\def\final{1}
\def\final{1}

\documentclass[11pt]{article}
\usepackage{fullpage}
\usepackage{enumitem,amssymb}
\usepackage{authblk}
\usepackage[all,cmtip]{xy}
\usepackage{url}
\usepackage{tcs-macros}
\usepackage{quantum-macros} 

\usepackage{fancybox}

\numberwithin{equation}{section}

\newcommand{\Hmin}{H_{\rm min}}
\newcommand{\Hhill}{H^{\mathsf{HILL}}}
\newcommand{\Hhillr}{H^{\mathsf{rHILL}}}
\newcommand{\Hmetric}{H^{\mathsf{metric}}}
\newcommand{\Hmetricr}{H^{\mathsf{metric\mbox{-}rlx}}}
\newcommand{\Hguess}{H^{\mathsf{guess}}}
\newcommand{\Pguess}{P^{\mathrm{guess}}}
\newcommand{\Dhillone}{D^{\mathsf{HILL\mbox{-}1}}}
\newcommand{\Dhilltwo}{D^{\mathsf{HILL\mbox{-}2}}}
\newcommand{\Dmetricone}{D^{\mathsf{metric\mbox{-}1}}}
\newcommand{\Dmetrictwo}{D^{\mathsf{metric\mbox{-}2}}}
\newcommand{\Dpseudo}{D^{\mathsf{pseudo}}}
\newcommand{\KL}[2]{D\left(#1\middle\|#2\right)}

\newcommand{\QCktMaxSat}{\textprob{QCkt-Max-Sat}}
\newcommand{\QCktTom}{\textprob{QCkt-Tomography}}
\newcommand{\QCktValue}{\textprob{QCkt-Value}}
\newcommand{\GapAmp}{\textprob{Gap-Amplification}}

\newcommand{\SC}{\mathsf{SC}}
\newcommand{\tK}{\tilde{K}}
\newcommand{\tX}{\tilde{X}}
\newcommand{\tG}{\tilde{G}}

\newcommand{\tlda}{\tilde{\lda}}
\newcommand{\Ext}{\mathsf{Ext}}
\newcommand{\Prg}{\mathsf{Prg}}

\newcommand{\secp}{\kappa} 
\newcommand{\net}[1]{\cN\left(#1\right)}
\newcommand{\bpovm}[1]{\Gamma\left(#1\right)}

\usepackage{environ}
\newcommand{\repeattheorem}[1]{%
  \begingroup
  \renewcommand{\thetheorem}{\ref{#1}}%
  \expandafter\expandafter\expandafter\theorem
  \csname reptheorem@#1\endcsname
  \endtheorem
  \endgroup
}
\NewEnviron{reptheorem}[1]{%
  \global\expandafter\xdef\csname reptheorem@#1\endcsname{%
    \unexpanded\expandafter{\BODY}%
  }%
  \expandafter\theorem\BODY\unskip\label{#1}\endtheorem
}

\newcommand{\repeatproposition}[1]{%
  \begingroup
  \renewcommand{\theproposition}{\ref{#1}}%
  \expandafter\expandafter\expandafter\proposition
  \csname repproposition@#1\endcsname
  \endproposition
  \endgroup
}
\NewEnviron{repproposition}[1]{%
  \global\expandafter\xdef\csname repproposition@#1\endcsname{%
    \unexpanded\expandafter{\BODY}%
  }%
  \expandafter\proposition\BODY\unskip\label{#1}\endproposition
}

\newcommand{\repeatlemma}[1]{%
  \begingroup
  \renewcommand{\thelemma}{\ref{#1}}%
  \expandafter\expandafter\expandafter\lemma
  \csname replemma@#1\endcsname
  \endlemma
  \endgroup
}
\NewEnviron{replemma}[1]{%
  \global\expandafter\xdef\csname replemma@#1\endcsname{%
    \unexpanded\expandafter{\BODY}%
  }%
  \expandafter\lemma\BODY\unskip\label{#1}\endlemma
}

\author[1]{
  \small Yi-Hsiu Chen\thanks{\texttt{yihsiuchen@g.harvard.edu}.
  Supported by NSF grant CCF-1420938 and work done in part while visiting the Institute of Information Science, Academia Sinica, Taiwan.}}
\author[2]{
  \small Kai-Min Chung\thanks{\texttt{kmchung@iis.sinica.edu.tw}.
  Supported by 2016 Academia Sinica Career Development Award under Grant no. 23-17 and the Ministry of Science and Technology, Taiwan under Grant no. MOST 103-2221-E-001-022-MY3.
This work was done in part while KMC was visiting the Simons Institute for the Theory of Computing, supported in part by the Simons Foundation and by the DIMACS/Simons Collaboration in Cryptography through NSF grant CNS-1523467.}}
\author[2]{
  \small Ching-Yi Lai\thanks{\texttt{cylai0616@iis.sinica.edu.tw}.}}
\author[3]{Salil P.~Vadhan\thanks{\texttt{salil-vadhan@harvard.edu}.
  Work done in part while visiting the Shing-Tung Yau Center and the Department of Applied Mathematics at National Chiao-Tung University, Taiwan. Supported by NSF grant CCF-1420938 and a Simons Investigator Award.}}
\author[4]{
  \small Xiaodi Wu\thanks{\texttt{xwu@cs.umd.edu}.}}

\affil[1]{Harvard John A. Paulson School Of Engineering And Applied Sciences, Harvard University, USA}
\affil[2]{Institute of Information Science, Academia Sinica, Taipei, Taiwan}
\affil[3]{Computer Science and Applied Mathematics, Harvard University, USA}
\affil[4]{Department of Computer Science, Institute for Advanced Computer Studies, and Joint Center for Quantum Information and Computer Science, University of Maryland, USA}

\title{Computational Notions of Quantum Min-Entropy}

\begin{document}

\hypersetup{pageanchor=false}
\pagenumbering{gobble}
\maketitle

\begin{abstract}
We initiate the study of computational entropy in the {\em quantum} setting. We investigate to what extent the classical notions of computational entropy generalize to the quantum setting, and whether quantum analogues of classical theorems hold. Our main results are as follows. (1) The classical Leakage Chain Rule for pseudoentropy can be extended to the case that the leakage information is quantum (while the source remains classical). Specifically, if the source has pseudoentropy at least $k$, then it has pseudoentropy at least $k-\ell$ conditioned on an $\ell$-qubit leakage. (2) As an application of the Leakage Chain Rule, we construct the first quantum leakage-resilient stream-cipher in the bounded-quantum-storage model, assuming the existence of a quantum-secure pseudorandom generator. (3) We show that the general form of the classical Dense Model Theorem (interpreted as the equivalence between two definitions of pseudo-relative-min-entropy) does {\em not} extend to quantum states. Along the way, we develop quantum analogues of some classical techniques (\eg~the Leakage Simulation Lemma, which is proven by a Non-uniform Min-Max Theorem or Boosting). On the other hand, we also identify some classical techniques (\eg~Gap Amplification) that do not work in the quantum setting. Moreover, we introduce a variety of notions that combine quantum information and quantum complexity, and this raises several directions for future work.
\end{abstract}


\clearpage
\tableofcontents

\clearpage
\hypersetup{pageanchor=false}
\pagenumbering{arabic}
\setcounter{page}{1}

\section{Introduction}\label{sec:intro}

Computational notions of entropy have many applications in cryptography and complexity theory.
These notions measure how much (min-)entropy a source $X$ has from the eyes of a computationally bounded party who may hold certain ``leakage information'' $B$ that is correlated with $X$.
They have several applications in cryptography, such as leakage-resilient cryptography~\cite{DziembowskiP08}, memory delegation~\cite{ChungKLR11}, deterministic encryption~\cite{FOR12}, zero-knowledge~\cite{ChungLP15},
pseudorandom generators~\cite{HastadILL99} and other cryptographic primitives~\cite{HRVW09}, and also have close connections to important results in complexity theory, such as Impagliazzo's hardcore lemma~\cite{Impagliazzo95}, and in additive number theory, such as the Dense Model Theorem~\cite{green2008primes,TZ08,ReingoldTTV08}.

In this work, we initiate the study of computational entropy in the quantum setting, where $X$ and $B$ may become quantum states and the computationally bounded observer is modeled as a small quantum circuit.
We find that some classical phenomena have (nontrivial) extensions to the quantum setting, but for others, the quantum setting behaves quite differently and we can even prove that the natural analogues of classical theorems are false.
As an application of some of our results, we construct a quantum leakage-resilient stream-cipher in the bounded-quantum-storage model, assuming the existence of a quantum-secure pseudorandom generator.
We expect that computational notions of quantum entropy will find other natural applications in quantum cryptography.
Moreover, by blending quantum information theory and quantum complexity theory, our study may provide new insights and perspectives in both of these areas.


In the rest of the introduction, we give a high-level overview of our work, highlight some of our interesting findings.
{sec:quantum-indistinguishability}.


\subsection{Brief Review of Quantum Information and Computation}\label{subsec:review-qi-qc}

Recall that a {\em pure state} in an $n$-qubit quantum system is a unit vector $\ket{\psi}\in\complex^{2^n}$.   The standard (``computational'') basis is denoted by $\{\ket{x} : x\in \zo^n\}$ and represents the set of classical bit strings $x\in \zo^n$.
Until they are {\em measured} (observed), quantum systems evolve via unitary operations ($2^n\times 2^n$ complex matrices $U$ such that $UU^\dag=I$, where $U^\dag$ is the conjugate transpose).
A projective {\em binary measurement} on the quantum system is given by a linear subspace $A \subseteq \complex^{2^n}$.
If the system is in state $\ket{\psi}\in\complex^{2^n}$, then the result of the measurement is determined by the decomposition $\ket{\psi} = \ket{\psi}_A + \ket{\psi}_{A^\perp}$, where $\ket{\psi}_A$ is the orthogonal projection of $\ket{\psi}$ to $A$.  With probability $\|\ket{\psi}_A\|_2^2$, the measurement returns 1 and the system collapses to state $\ket{\psi}_A/\|\ket{\psi}_A\|_2$, and with probability
$\|\ket{\psi}_{A^\perp}\|_2^2$, the measurement returns 0 and the system collapses to state $\ket{\psi}_{A^\perp}/\|\ket{\psi}_{A^\perp}\|_2$.
We abuse notation and write $A(\ket{\psi})$ to denote the $\zo$ random variable that is the outcome of the measurement.
There is a more general form of binary measurement (described by a ``projective operator value measurement'' (POVM)), but we only need a projective binary measurement to discuss most concepts in the introduction, and defer the definition of POVM to where we need it.

A {\em mixed state} $\rho$ of a quantum system can be specified by a probability distribution $\{p_i\}$ over pure states $\{\ket{\psi_i}\}$.  If we evolve $\rho$ by applying a unitary transformation $U$,
it will be in the mixed state given by distribution $\{p_i\}$ over the pure states $\{U\ket{\psi_i}\}$.
If instead we perform a measurement $A$ on such a mixed state $\rho$, then, by definition, $\Pr[A(\rho)=1] = \sum_i p_i\cdot \Pr[A(\ket{\psi_i})=1] = \sum_i p_i\cdot \|\ket{\psi_i}_A\|_2^2$.
The representation of a mixed state as a probability distribution over pure states is not unique, in that two such representations can yield exactly the same behavior under all sequences of unitary transformations and measurements.\footnote{A unique representation of a mixed state is given by its {\em density matrix} $\sum_i p_i \ketbra{\psi_i}$, which is a $2^n\times 2^n$ positive semidefinite matrix of trace one, and thus we use the density matrix formalism in the technical sections of the paper.}
For example, the {\em maximally mixed state} $\sigma_{\mix}$ is defined as the uniform distribution over the standard classical basis $\{ \ket{x} : x\in \zo^n \}$, but using any orthonormal basis of $\complex^{2^n}$ yields an equivalent mixed state (and thus all of them are regarded as the same mixed state $\sigma_{\mix}$).

Recall that the {\em min-entropy} of a classical random variable $X$ is given by
\[\Hmin(X) =\min_x \log(1/\Pr[X=x]) = \frac{1}{\log\left(\max_x \Pr[A_x(X)=1]\right)},\]
where $A_x$ is the indicator function for $x$.
When we have a mixed quantum state $\rho_X$ instead of a classical random variable $X$, we generalize from indicator functions to one-dimensional binary measurements~\cite{Renner05}.
That is, if $\rho_X$ is a mixed quantum state, then:
\[\Hmin(X)_\rho = \frac{1}{\log\left(\max_{\ket{\psi}} \pr{A_{\ket{\psi}}(\rho)=1}\right)},\]
where $A_{\ket{\psi}}$ is the binary measurement given by the one-dimensional subspace spanned by $\ket{\psi}$.  This generalizes the classical definition.
If $\rho$ is given by a distribution $\{p_x\}$ over the classical basis $\{\ket{x}\}$, then the maximum value of $\Pr[A_{\ket{\psi}}(\rho)=1] = \sum_x p_x |\braket{\psi}{x}|^2$ is $\max_x p_x$, obtained by taking $\ket{\psi} = \ket{y}$ for $y=\mathrm{argmax}_x p_x$.
On the other hand, if $\rho$ is a pure state, with all of its probability on a single unit vector $\ket{\phi}$, then the maximum probability is 1 (yielding {\em zero} min-entropy), obtained by taking $\psi = \phi$.

Informally, a {\em quantum circuit} computes on a quantum state (which may be a classical input $\ket{x}$ for $x\in \zo^n$) by applying a sequence of {\em local gates}, which are unitary transformations and measurements that apply to only a constant number of qubits in the state.
Quantum circuits are also allowed extra {\em ancilla} qubits (in addition to the $n$ input qubits).
We usually require those ancilla bits are initialized to be $\ket{0^n}$.
The {\em size} of a quantum circuit is the number of gates.

\subsection{Quantum Computational Notions}\label{subsec:comp-notion}

\paragraph{Quantum Indistinguishability.}
In many applications of cryptography and complexity theory, we only require the security against adversaries with restricted power.
One of the most common restrictions is considering only polynomial time bounded circuits/algorithms.
Certainly, ``polynomial time'' is meaningful only when we consider an ensemble of inputs and circuits.

In the classical world, there are two different computational models that are widely studied.
First, in the {\em nonuniform computation} model, circuits can depend on the input size, while in the {\em uniform computation} model, the same algorithm is used for inputs of any size, or equivalently, there is an algorithm that can generate the ensemble of circuits.
Once the universal gate set is fixed, we can define the size of a circuit.
Then both models can be extended to the quantum setting naturally by replacing circuits with quantum circuits.
In this article, we mostly focus on the nonuniform settings, as adversaries have more power in this model.
Consider two quantum state ensembles $\{\rho_{n}\}$ and $\{\sigma_{n}\}$ where $n$ bounds the number of qubits in $\rho_n$ and $\sigma_n$ and serves as the security parameter.
We say $\{\rho_{n}\}$ and $\{\sigma_{n}\}$ are quantum-indistinguishable if for every $\poly(n)$-size nonuniform quantum algorithm $\{A_n\}$, we have $|\Pr[A_n(\rho_n) = 1] - \Pr[A_n(\sigma_n) = 1]| \leq \negl(n)$.
Sometimes, we consider the asymptotic setting implicitly by omitting the index $n$.

Classically, an equivalent way to define a nonuniform circuit ensemble is giving a uniform algorithm (\eg~a Turing machine) advice strings which only depend on input lengths.
In the quantum setting, this formation of uniform algorithms with advice matches the above definition of nonuniform quantum circuits if we restrict the advice strings to be classical.
But one can consider an even more general computational model by giving the circuits quantum advice.
A simple way to incorporate quantum advice is to allow the quantum circuits have some of ancilla bits initialized to be the quantum advice.
In this model, the quantum analogue of the classical complexity class $\Ppoly$ is $\BQPqpoly$, which was defined by Nishimura and Yamakami \cite{NY04}.
An intriguing and well known question is whether quantum advice provides more power in computation?
\ie~does $\BQPqpoly = \BQPpoly$ (and whether $\QMA = \QCMA$).
One can also define the indistinguishability with quantum advice.
Some of our results hold in this model as well.
For the sake of simplicity, in the rest of the introduction, we only consider classical advice when it is not specified explicitly.

\paragraph{Pseudorandom States.}

In the classical setting, a distribution $X_n$ over $\zo^n$ is {\em pseudorandom} if $X$ is computationally indistinguishable from the uniform distribution $U_n$.
Namely, for every probabilistic $\poly(\secp)$-size circuit $D$, the distinguishing advantage of $D$ is negligible, \ie~$|\Pr[D(X)=1] - \Pr[D(U_n) = 1] | \leq \negl(\secp)$, where $\secp$ is a security parameter.
That $X$ is pseudorandom means $X$ has full $n$ bits of computational entropy.
Pseudorandomness is a fundamental notion pervasive in cryptography.

In the quantum setting, the classical uniform distribution is represented as the maximally mixed state $\sigma_{\mix}$, as defined earlier.
Thus, the notion of pseudorandomness generalizes naturally: a mixed quantum state $\rho$ is {\em pseudorandom} if it is computationally indistinguishable from $\sigma_{\mix}$ to quantum distinguishers.
That is, for every  $\poly(\secp)$-size quantum circuit $D$, $|\Pr[D(\rho)=1] - \Pr[D(\sigma_{\mix}) = 1] | \leq \negl(\secp)$.
We note that pseudorandomness of {\em classical} distributions against quantum distinguishers has been studied extensively in the context of post-quantum cryptography (\eg~\cite{Song14} and references therein).
The difference here is that we consider pseudorandomness for quantum states.

An interesting observation is that there exist {\em pure} states that are pseudorandom, due to Bremner, Mora and Winter~\cite{BMW08}, and Gross, Flammia and Eisert~\cite{GFE09}.
In~\cite{BMW08}, the existence of a pseudorandom pure state was viewed as a negative result, showing that random pure states are not useful for efficient quantum computation, since they can be replaced by uniform classical bits.
However, from the perspective of pseudorandomness and computational entropy, it is a {\em positive} result, asserting the existence of a pseudorandom state that has {\em zero} entropy (as pure states have zero entropy).
This is a sharp contrast from the classical setting, where a classical distribution needs min-entropy at least $\omega(\log n)$ to be pseudorandom.
Indeed, the existence of pseudorandom pure states reveals a sharp contrast between the quantum setting and the classical setting: it says that a quantum state with {\em zero} entropy (\ie~a pure state) can be pseudorandom, whereas a classical distribution needs min-entropy at least $\omega(\log n)$ to be pseudorandom.
We remark that this separation relies on our choice to consider quantum distinguisher with only classical advice.
(If we allow quantum advice, it is an interesting open problem to determine the existence of a pseudorandom state of entropy $o(n)$)

As discussed in Section~\ref{subsec:dmt} below, we use the existence of pseudorandom pure states to show the separation of two classically equivalent notions of ``computational relative entropy'' in the quantum world.

\paragraph{Computational Quantum (Min-)Entropy}
We next investigate computational notions of entropy in the quantum setting.
One of the most natural ways to define the computational min-entropy is that we say a state has {\em computational min-entropy at least $k$} if it is indistinguishable (by $poly$-size quantum circuits) from a state with entropy at least $k$.
If $k$ equals the number of qubits of the state, then this is simply the definition of pseudorandomness described above, as the maximally mixed state is the unique state of (min-)entropy $k$.

In the classical settings, this definition was first proposed by H\aa stad, Impagliazzo, Levin and Luby~\cite{HastadILL99}, constructing a pseudorandom generator from any one-way function.
There are a  number of other ways to define computational min-entropy with many different applications and interesting connections to other fields.
In Section~\ref{sec:comp-entropy}, we give a systematic overview of classical definitions and known relations among them and discuss quantum analogues of these definitions and our findings.

In the rest of introduction, we focus on the HILL-style computational entropy definitions, which are the most widely used notions in the classical setting.
Under this notion, we prove the quantum analogues of the ``Leakage Chain Rule for HILL pseudoentropy'' which has cryptographic applications.
There remain many interesting open problems about understanding and relating quantum analogues of other definitions of pseudoentropy.

\subsection{Leakage Chain Rule}

\paragraph{Conditional Min-Entropy.}
Many applications require measures of {\em conditional} entropy $H(X|B)$, where $B$ often
represents some {\em leakage} about $X$.
A popular and useful measure of conditional min-entropy in the classical setting is the notion of {\em average min-entropy} by \cite{DodisORS08}, which has a nice operational meaning in terms of the guessing probability:
Let $(X, B)$ be a joint distribution over $\zo^{n+\ell}$.
The \emph{guessing probability} of $X$ conditioned on $B$ is defined as the maximum probability that an algorithm can guess $X$ correctly given $B$.
That is, $\Pguess(X|B) \eqdef \max_A \Pr[ A(B) = X ]$,
where the maximum is taken over \emph{all} (even computationally unbounded) algorithms $A$.
Then the {\em conditional min-entropy} (as known as {\em average min-entropy}) of $X$ given $B$ is defined as $\Hmin(X|B) = -\log (\Pguess(X|B))$.

The definition of conditional min-entropy $H(X|B)_{\rho}$ for bipartite quantum states $\rho_{XB}$ was given by \cite{Renner05}, which generalizes the aforementioned definition of average min-entropy as well as our earlier definition for (non-conditional) min-entropy of quantum states (in Section~\ref{subsec:review-qi-qc}).
A natural way to generalize the guessing game is, given $B$, the guesser apply a POVM indexed by a vector $\ket{\psi}$, then do the binary measurement on $X$ part given by the one-dimensional subspace spanned by $\ket{\psi}$.
Then the guessing probability is the probability of getting $1$.
For the special case of classical $X$ and quantum $B$, which is called a classical-quantum-state ({\em cq-state}), K\"{o}nig, Renner and Schaffner proved the generalized guessing game described above captures the conditional min-entropy definition~\cite{KonigRS09}.
When two parts are quantum (a qq-state), the guessing probability may give higher entropy then Renner's definition. (Instead, an operational interpretation of Renner's definition is as the maximum achievable singlet fraction~\cite{KonigRS09}.)
In fact, when $\rho_{XB}$ is entangled, the conditional min-entropy can be negative, which is impossible to be captured by a guessing probability.

The cq-state case is particularly useful in quantum cryptography, such as quantum key distribution (QKD)~\cite{benn-84,Renner05,VV13}, device-independent cryptography~\cite{VV12, MS, CSW:QIP}, and  quantum-proof randomness extractors~\cite{DePVR12}.
Also it has a more natural operational interpretation.
Thus, we focus on conditional min-entropy for cq-states in this paper, and leave the study of conditional min-entropy for qq-states and computational analogues for future work.

\paragraph{Conditional Pseudoentropy.}
Classically, for a joint distribution $(X,B)$, we say that $X$ conditioned on $B$ has {\em conditional relaxed HILL pseudo(min-)entropy} at least $k$ if there exists a distribution $(X',B')$ that is computationally indistinguishable from $(X,B)$ with $\Hmin(X'|B') \geq k$.
(This definition is called {\em relaxed} HILL pseudoentropy because we do not require that $B'$ is identically distributed to $B$.
For short, we will write rHILL to indicate that we are working with the relaxed definition.)

In the quantum setting, let $\rho_{XB}\in\cX\ot\cB$ be a bipartite cq-state with $n+\ell$ qubits.
We say that $X$ conditioned on $B$ has {\em conditional quantum rHILL pseudo(min-)entropy} at least $k$ (informally written as $\Hhillr(X|B)_\rho \geq k$) if there exists a quantum state $\sigma_{XB}$ such that (i) $\Hmin(X|B)_{\sigma} \geq k$ and (ii) $\rho_{XB}$ and $\sigma_{XB}$ are computationally indistinguishable by all $\poly(\secp)$-size quantum distinguishers, where again $\secp$ is the security parameter.

\paragraph{Leakage Chain Rule for quantum HILL pseudoentropy.}   The classical Leakage Chain Rule for rHILL pseudoentropy, first proved by~\cite{DziembowskiP08,ReingoldTTV08}
and improved by~\cite{FullerR12,Skorski13}, states that for a joint distribution $(X,Z,B)$ where $B$ consists of $\ell = O(\log \secp)$ bits,
\[\Hhillr(X|Z) \geq k \Rightarrow \Hhillr(X|Z, B) \geq k - \ell.\]
(Note that under standard cryptographic assumptions, the analogous statement for (standard) HILL pseudoentropy is false~\cite{KrennPW13}.)
The leakage chain rule is an important property for pseudoentropy and has a number of applications in cryptography, such as leakage-resilient cryptography~\cite{DziembowskiP08}, memory delegation~\cite{ChungKLR11}, and deterministic encryption~\cite{fuller2015unified}.

In this paper, we prove that this Leakage Chain Rule can be generalized to handle quantum leakage $B$ when both the source $X$ and the prior leakage $Z$ remain classical.

\begin{theorem} [Quantum Leakage Chain Rule; informal] \label{thm:intro:chain-rule}
Let $\rho_{XZB}$ be a ccq-state, where $X$ and $Z$ are classical and $B$ consists of $\ell$ qubits, for $\ell=O(\log \secp)$, where $\secp$ is the security parameter. Then
$$\Hhillr(X|Z)_{\rho} \geq k \Rightarrow \Hhillr(X|Z, B)_{\rho} \geq k - \ell.$$
\end{theorem}

Theorem~\ref{thm:intro:chain-rule} is proved by a quantum generalization of the \emph{Leakage Simulation Lemma}~\cite{TrevisanTV09,JP14,ChungLP15,VZ13} to its quantum analogue (which implies Theorem~\ref{thm:intro:chain-rule} immediately).
There are two types of classical proofs for the Leakage Simulation Lemma: one based on Boosting and one based on the Min-Max Theorem.
We develop quantum analogues of both proofs.
We will demonstrate the boosting proof in Section~\ref{subsec:boosting-proof}, and leave the min-max proof in Appendix~\ref{app:min-max-proof} as it is more involved.
Both proofs also rely on efficient algorithms for quantum tasks such as {\em POVM tomography} and {\em quantum circuit synthesis} to construct efficient reductions.
This leads to a variant of POVM tomography problem that merits further study.

An interesting open question is to prove the leakage chain rule when the source $X$ and/or the prior leakage $Z$ are quantum.
In particular, handling a prior quantum leakage $Z$ seems important for applications to leakage-resilient cryptography with quantum leakage.
This is not likely to be a direct generalization of Theorem~\ref{thm:intro:chain-rule} as even the information theoretic leakage chain rule loses $2\ell$~\cite{WTHR11} rather than $\ell$ bits of entropy.
We leave an investigation of computational analogues of this Leakage Chain Rule to future work.
In Section~\ref{ssec:intro-barrier}, we discuss a general barrier to further generalizing our proof to handling quantum $X$ and $Z$ as well as to generalizing many other proofs of classical theorems.

\subsection{Main Techniques}

We first review a common technique for proving some classical theorems about indistinguishability and computational entropy (\eg~Impagliazzo's Hardcore Lemma~\cite{Impagliazzo95}, Regularity Lemma~\cite{TrevisanTV09} and Dense Model Theorem~\cite{ReingoldTTV08}).
In those proofs, there is usually a step of ``switching quantifiers'', and that is a place where the min-max theorem comes in.
For instance, one wants to prove that $X$ not having high HILL entropy against $s$-size circuits implies $X$ does not having high metric entropy~\cite{BarakSW03}.
From the assumption, we know for all distributions $Y$ with high min-entropy, there exists an $s$-size circuit to distinguish $X$ and $Y$.
We view it as a game between Player~1 who provides a distribution, and Player~2 who provides a distinguisher, so Player~2 always has a certain payoff by the assumption.
Then by von Neumann's Min-Max Theorem, Player~2 has a universal ``mixed strategy,'' which is a distribution over distinguishers that achieves high payoff for every distribution from Player~1's strategy.
It turns out that this statement almost gives us the desired conclusion, except we only got a {\em distribution} over $s$-size distinguishers instead of a small size distinguisher.
One common remedy is a ``sampling''.
If we sample polynomially many distinguishers from the distribution over distinguishers, then by a Chernoff bound, with high probability, the average of the sampled distinguishers performs well for any fixed $x$ in Player~1's strategy space.
Then by a union bound, the samples work well for all $x$.
Therefore, we can obtain a universal distinguisher with complexity roughly $s$ times the (polynomial) number of samples we need.

Suppose we would like to generalize the above statement to the quantum settings, in which a strategy of Player~1 is a quantum state rather than a classical distribution.
An immediate obstacle is that we cannot union bound over all possible pure states.
A naive approach is to union bound over an ``$\eps$-net'' of all possible pure states.
In this way, one can still obtain a nonuniform min-max theorem.
However, since the net is doubly exponentially large (c.f.~Proposition~\ref{prop:enet-vec}), the complexity of the universal circuit is too large for some applications.
Therefore, another view of the min-max theorem that generalizes to the quantum setting better is desired.

Freund and Schapire \cite{FS99} showed how to use the multiplicative weight update (MWU) method to prove von Neumann's Min-Max Theorem.
Moreover, the proof provides a constructive way to approximate the universal strategy.
Inspired by Barak Hardt and Kale~\cite{BHK09}, Vadhan and Zheng~\cite{VZ13} adopted the framework to show a {\em uniform min-max theorem}.
Contrary to the above sampling method, which can only show the {\em existence} of the low complexity universal strategy, the uniform version provides an efficient way to find a universal strategy.
As applications, they proved some aforementioned theorems (\eg~Impagliazzo's Hardcore Lemma and Dense Model Theorem) with uniform security using the uniform min-max theorem.
From this perspective of MWU, a (mixed) strategy of Player~1 is viewed as a weight vector.
During the MWU procedure, we maintain the weight vectors over the support of $X$.
In each round, we update the weight vector according to some loss function.
In an extension to the quantum setting, the Player~1's strategy is a quantum state, which can be represented as a density matrix.
Luckily, there is a generalization of MWU, called the matrix multiplicative update (MMWU) method~\cite{Kale07}.
In an MMWU procedure, instead of maintaining a weight vector, we keep updating a weight matrix, which is semidefinite positive.
The weight matrix can be seen as a quantum state in Player~1's strategy space as long as it is Hermitian.
In this work, we prove the quantum min-max theorem via MMWU and then are able to apply it to connect different pseudoentropy notions and the Quantum Leakage Chain Rule.

In $\cite{VZ13}$, the purpose of using the MWU approach was to obtain a constructive version of the min-max theorem.
Since we use a similar argument in the quantum setting, in fact, we can also have a uniform version of the quantum min-max theorem.
However, we have not found any further applications of the constructive version, so we still only state it as the quantum nonuniform min-max theorem for simplicity.

\subsection{Application to Quantum Leakage-Resilient Stream-Ciphers}

In this section, we demonstrate an application of computational quantum entropy to leakage-resilient cryptography, where we seek to construct cryptographic protocols that maintain the security even if the side information about the honest parties' secrets leaks to an adversary.
Specifically, we construct a leakage-resilient stream-cipher that is secure against \emph{quantum leakage}.

Classical leakage-resilient stream-ciphers were investigated in the seminal work of Dziembowski and Pietrzak~\cite{DziembowskiP08}, where they consider the security of a stream-cipher $\SC$ in the ``only computation leaks'' model~\cite{MicaliR04} with continual leakage.
Specifically, let $S_i$ denote the secret state of $\SC$.
At each round $i$ when the stream cipher evaluates $(S_{i+1}, X_{i+1}) = \SC(S_i)$, an adversary can adaptively choose any leakage function $f_i$ and learn the output of $f_i$ applied to the part of $S_i$ involved in the computation of $\SC(S_i)$.
They assume that the leakage functions are efficient and of bounded output length $\ell = O(\log \secp)$,\footnote{Note that both assumptions are necessary.
Without the efficiency assumption, the leakage function can invert the secret state and leak on the initial secret $S_0$ bit by bit.
Without the length bound, the adversary can learn the entire new secret state.} and proved the security property: the output of the $i$-th round remains pseudorandom given the output and leakage of the first $i-1$ rounds.
Note that even though the length of each leakage is bounded, in total the adversary can collect a long leakage accumulated over many rounds.

Dziembowski and Pietrzak~\cite{DziembowskiP08} gave the first construction of a leakage-resilient stream-cipher based on randomness extractors and pseudorandom generators (PRGs), and proved the security using the classical Leakage Chain Rule for HILL pseudoentropy.
Pietrzak~\cite{Pie09} gave a simpler construction based on any weak pseudorandom functions (weak PRFs), and Jetchev and Pietrzak~\cite{JP14} gave an improved analysis of~\cite{Pie09} using the classical Leakage Simulation Lemma.

Now we consider the case that the leakage is quantum (while the stream-cipher remains classical).
Namely, the output of the leakage functions is a bounded-length quantum state.
It is conceivable that such an attack may exist in the future with the emergence of quantum computers.
We also view this as a natural theoretical question that investigates problem information through a cryptographic lens.
We show that the construction of Dziembowski and Pietrzak~\cite{DziembowskiP08} remains secure against quantum leakage in the bounded-quantum-storage model~\cite{DamgardFSS05,KonigT08,WehnerW08,Unruh11}, where the adversary has a limited quantum memory (but no restriction on its classical memory.
The model is investigated in the literature as a way to bypass impossibility results~\cite{DamgardFSS05,WehnerW08,Unruh11} or to prove security~\cite{KonigT08}.


\begin{theorem}[Quantum Leakage-Resilient Stream-Cipher; informal]
  Assuming the existence of quantum-secure pseudorandom generators against quantum distinguisher with quantum advice, there exists quantum leakage-resilient stream-cipher secure against bounded-quantum-storage adversaries with $O(\log \secp)$ quantum memory and $\poly(\secp)$ circuit size, where $\secp$ is the security parameter.
\end{theorem}

Note that both bounds on the leakage and quantum storage are logarithmic in the security of the underlying primitives. If the PRG has exponential security, then the leakage and adversary's quantum storage can be linear in the size of the secret state.

When proving the quantum security of classical cryptographic constructions (\eg~construction of PRGs), it often suffices to assume quantum security of the underlying primitives (\eg~one-way functions or lattice assumptions), since typically the security reductions are ``nice'' and can be directly carried through in the quantum setting. (See the work of Song~\cite{Song14} for a nice framework formalizing this observation.)
However, this is not the case here due to the presence of quantum information.\footnote{There are several other challenging cases such as when the reduction needs to rewind the adversary~\cite{Watrous09,Unruh12}, or when the setting involves oracles~\cite{BonehDFLSZ11,Zhandry12}.}
Several issues arise when we generalize the classical proofs to handle quantum leakage.
We provide detailed discussion on the issues we encounter in Section~\ref{sec:app-resilient}, and here
 we only explain where the bounded-quantum-storage limitation comes from.

The classical proof of the security iteratively argues that the secret state $S_i$ of the stream-cipher remains pseudorandom to the adversary.
When the adversary leaks up to $O(\log \secp)$ bits on $S_i$, the Leakage Chain Rule for HILL pseudoentropy says that $S_i$ retains high HILL pseudoentropy given the adversary's view.
Thus, if the stream-cipher applies an extractor and a PRG, $S_i$ can be replaced with a new state $S_{i+1}$ that is pseudorandom (while also generating some pseudorandom output bits for the stream-cipher).
The same argument works for a single round of leakage in the quantum setting using our Leakage Lemma for Quantum HILL Pseudoentropy (Theorem~\ref{thm:intro:chain-rule}.
But over many rounds, the adversary can accumulate long quantum auxiliary information $Z$, and we do not know how to prove the leakage lemma in such a case. (The lemma does hold if $Z$ is classical.)
Assuming bounded quantum storage prevents this accumulation (but still allows the accumulation of classical leakage information). Our proof formally proceeds by using our Quantum Leakage Simulation Lemma to efficiently simulate the adversary's quantum state at each round.
We leave it as an interesting open question to identify versions of Leakage Chain Rule or Leakage Simulation Lemma that are sufficient to prove the security of quantum leakage-resilient stream-cipher against a general adversary.

\subsection{The Dense Model Theorem and Computational Relative Min-Entropy}\label{subsec:dmt}
First, we review the two possible definitions of computational relative entropy in the classical case.
Let $X$ and $Y$ be two distributions.
The {\em relative min-entropy} from $X$ to $Y$ is defined as $\max_{x\in\Supp(X)}\log(\pr{X = x}/\pr{Y = x})$.
That is, $D(X\| Y) \leq \lda$ iff $X$ is $2^{-\lda}$-dense in $Y$. \ie~$\pr{Y = x} \geq 2^{-\lda}\pr{X = x}$ for all $x$.
And if $X$ is distributed over $n$-bit strings, then $D(X\|U_n) \leq \lda$ iff $\Hmin(X) \geq n-\lda$.
We say the {\em HILL-1 relative min-entropy} from $X$ to $Y$ is at most $\lambda$ (informally written as $\Dhillone(X\|Y) \leq \lambda$) if there exists a distribution $X'$ computationally indistinguishable from $X$ such that $D_{\infty}(X'\|Y) \leq \lambda$.
We say that {\em HILL-2 relative min-entropy} from $X$ to $Y$ is at most $\lambda$ (informally written as $\Dhilltwo(X\|Y) \leq \lambda$) if there exists a distribution $Y'$ computationally indistinguishable from $Y$ such that $D_{\infty}(X\|Y') \leq \lambda$.
These definitions naturally generalize to the quantum setting by replacing the distributions with quantum states and taking the computational indistinguishability to be against polynomial-sized quantum distinguishers (with classical advice).

With the above definitions, the Dense Model Theorem of \cite{ReingoldTTV08} can be equivalently stated as $\Dhilltwo(X\|Y) \leq \lambda \Rightarrow \Dhillone(X\|Y) \leq \lambda$ for $\lambda = O(\log \secp)$, where $\secp$ is the security parameter.
We show that the proof of the Dense Model Theorem can be adapted to show the converse: $\Dhillone(X\|Y) \leq \lambda \Rightarrow \Dhilltwo(X\|Y) \leq \lambda$ for $\lambda = O(\log \secp)$ (see Lemma~\ref{lemma:reverse-dmt}).
Therefore, the two notions of computational relative min-entropy are equivalent in the classical setting. In contrast, in the quantum setting, we show a separation of the two notions of HILL quantum relative min-entropy.

\begin{theorem}[Separation of HILL quantum relative min-entropy]\label{thm:intro-separation}
  For a security parameter $\secp$ and every $n = \omega(\log \secp)$, there exist two $n$-qubit mixed quantum states $\rho, \sigma\in\cH$  such that
  \[\Dhilltwo(\rho\|\sigma) \leq 1 \mbox{ but } \Dhillone(\rho\|\sigma) = \infty.\]
\end{theorem}

The counterexample is based on the existence of pseudorandom pure states \cite{BMW08,GFE09} (see Section~\ref{subsec:comp-notion}).
Specifically, we take $\sigma$ to be a pseudorandom pure state, and $\rho$ be a classical distribution with the first bit equal to $1$, and the remaining $n-1$ bits being uniform.
First, $\Dhilltwo(\rho\|\sigma) \leq 1$ since $D_{\infty}(\rho \| \sigma_{\mix}) = 1$ and $\sigma$ is pseudorandom.
On the other hand, the fact that $\sigma$ is pure implies $D_{\infty}(\tau\|\sigma) = \infty$ for any $\tau \neq \sigma$. Also, $\rho$ and $\sigma$ can be distinguished by simply measuring and outputting the first bit. Therefore, $\Dhillone(\rho\|\sigma) = \infty$.

\subsection{Barrier Result: Impossibility of Quantum Gap Amplification} \label{ssec:intro-barrier}

Herein, we discuss the general barrier to extending proofs of many classical proofs to the quantum setting.
Let us consider the following quantum gap amplification problem: Let $p > q > \eps \in (0,1)$ be parameters, and $A$ a quantum algorithm with $n$ qubits input and binary output.
The task here is to perform a reduction that transforms $A$ into an ``amplified'' algorithm $R$ such that for every \emph{pure state} input $\ket{\psi}$, if $\Pr[A(\ket{\psi}) = 1] \geq p$, then $\Pr[R(\ket{\psi}) = 1] \geq p'$ for $p'>p$; if $\Pr[A(\ket{\psi}) = 1] \leq q$, then $\Pr[R(\ket{\psi}) = 1] \leq q'$ for $q'<q$.
A typical setting is $p=2/3$, $q=1/3$, $p'=1-\eps$, and $q'=\eps$.
We show that the task is impossible to achieve even if the input is a single qubit.

\begin{theorem}[Impossibility of Quantum Gap Amplification; informal]
  For every $p'>p > q >q' \in (0,1)$, there exists a quantum algorithm $A$ with $1$ qubit input and binary output such that there does not exist any quantum algorithm $R$ that solves the gap amplification problem defined above with respect to $p$, $q$, $p'$, $q'$, and $A$.
\end{theorem}
This impossibility can be viewed as a strengthening of the no-cloning theorem~\cite{WZ82}, which states that an unknown quantum pure state cannot be cloned perfectly.
Indeed, if the input state could be cloned perfectly, the gap amplification problem can be solved by ordinary repetitions and ruling by the majority or a threshold.

For a complexity theory interpretation, we can consider a class of promise problems, where the instances are pure quantum states, and a promise problem is in the class if there is a (possibly resource-bounded) quantum algorithm $A$ that distinguishes YES instances from NO instances with error probability bounded by, say, 1/3.
Consequently, such a class is not robust to the choice of error probability; different choices yield different classes.
This is in contrast to complexity classes such as $\mathrm{BPP}$, $\mathrm{BQP}$, and $\mathrm{QMA}$, where error reduction is possible. The key difference is that, in these classes the input is classical, and hence can be cloned.
In the case of $\mathrm{QMA}$, there is a quantum witness, but the witness for the amplified algorithm can be allowed to consist of many copies of the original witness.
(There is also a non-trivial way to do error reduction for $\mathrm{QMA}$ with single quantum witness~\cite{QMA05}. It circumvents the impossibility since the amplification is only defined with respect to the correct quantum witness.)

To see the impossibility, we sketch a simple argument for the case $p = 0.4$, $q = 0.3$, and $\eps = 0.1$.
Consider a quantum algorithm $A$ that takes a single qubit input $\rho$, measures it in the computational basis, and outputs the outcome, and the following four pure states: $\ket{0}$, $\ket{1}$, $\ket{+}=(\ket{0}+\ket{1})/\sqrt{2}$, and $\ket{-}=(\ket{0}-\ket{1})/\sqrt{2}$.
$A$ outputs 1 with probability 1 on $\ket{1}$, $0$ on $\ket{0}$, and $0.5$ on $\ket{+}$ or $\ket{-}$.
Thus, by definition, $R$ has to accept $\ket{0}$ with probability at most $0.1$, and accept $\ket{1},\ket{+}, \ket{-}$ with probability at least $0.9$.
Now consider the behavior of $R$ on the maximally mixed state $\sigma_{\mix}$, which can be equivalently described as the uniform distribution over $\ket{0}$ and $\ket{1}$, or as the uniform distribution over $\ket{+}$ and $\ket{-}$.
This means that the acceptance probability of $R$ on $\sigma_{\mix}$ is the average of its acceptance probabilities on $\ket{0}$ and $\ket{1}$, but also the average of its acceptance probabilities on $\ket{+}$ and $\ket{-}$.
However these two averages are at most $.55$ and at least $.9$, respectively.

Now, the barrier follows by observing that gap amplification is a common ingredient in many proofs of classical theorem based on reductions.
When a reduction takes unknown quantum states as input, the impossibility of quantum gap amplification implies that the classical reduction proof that uses gap amplification cannot be directly generalized to the quantum setting.
For example, in proving the Leakage Chain Rule for rHILL pseudoentropy, we need to construct a reduction $R$ showing that $X$ conditioned on $Z$ has less than $k$ bits of rHILL pseudoentropy; here $R$ takes both $X$ and $Z$ as input.
Thus, it is difficult to generalize classical proofs of the Leakage Chain Rule (all of which utilize gap amplification) when $X$ or $Z$ is a quantum state. Nevertheless, we emphasize that this does not imply that quantum analogues of classical theorems are false, but that different techniques must be introduced.

\section{Preliminaries}\label{sec:preliminary}

\subsection{Quantum Information}

\begin{trivlist}
\item \textbf{Quantum states.}
  We begin with some notation.
  Suppose $\cX$ is a complex vector space with inner product. A vector in $\cX$  is denoted by $\ket{v}$ and
  its  conjugate transpose  is denoted by $\bra{\psi}=\ket{\psi}^{\dag}$.
  The inner product and outer product of two vectors $\ket{v}, \ket{w}\in\cX$ are denoted by $\braket{v}{w}$ and $\ket{v}\bra{w}$, respectively.
  The norm of $\ket{v}$ is defined by
  \[ \norm{\ket{v}}_2 =\sqrt{\braket{v}{v}}. \]
  The set of all unit vectors in $\cX$ is denoted by $\ball{\cX}$.
  Let $\lin{\cX}$ denote the set of all linear operators on $\cX$.
  Let $\herm{\cX}$ denote the set of all Hermitian operators on space $\cX$, \ie~$\herm{\cX}=\{X \in \lin{\cX}: X^\dag=X\}$, where $X^{\dag}$ is the conjugate transpose of $X$.
  The Hilbert-Schmidt inner product on $\lin{\cX}$ is defined by
  \[\textstyle \ip{X}{Y}=\tr{X^{\dag}Y},  \forall X,Y \in \lin{\cX}.\]

  A quantum state space $\cX$ of $m$ qubits is the complex Euclidean vector space $\cX=\complex^{2^m}$.
  An $m$-qubit quantum state is represented by a \emph{density operator} $\rho$ in $\herm{\cX}$, which is a positive semidefinite Hermitian operator on $\cX$ with trace one.
  When $\rho$ is of rank one, it refers to a \emph{pure} quantum state, which can also be represented by a unit vector $\ket{\psi}$ in $\ball{\cX}$.
  In that case, the density operator $\rho$ can be written as $\ketbra{\psi}$.
  Otherwise, the density operator $\rho$ refers to a \emph{mixed} quantum state.
  Thus in an appropriate basis (anyone that diagonalizes $\rho$) we can think of $\rho$ as a classical distribution  on the basis elements.
  In general, the expression is not unique. 
  The set of all quantum density operators on $\cX$ is denoted by
  \[\density{\cX}:=\{ \rho \in \herm{\cX}:  \rho\geq 0, \tr{\rho}=1\}=\mbox{ConvexHull}\{\ketbra{\psi}:\ket{\psi}\in \ball{\cX}\},\]
  where the notation $\rho\geq 0$ means that $\rho$ is positive semidefinite. Likewise, $\sigma\geq \rho$ means that $\sigma-\rho$ is positive semidefinite.
  Let $\I_\cX$ denote the identity operator on $\cX$ (or $\I_d$ when the dimension of $\cX$ is known to be $d$).
  Then the \emph{maximally mixed state} in $\cX$ is $\sigma_{\mix} = \frac{1}{\dim(\cX)}\I_\cX$,
  where $\dim(\cX)$ is the dimension of $\cX$.

  For a quantum system $X$, its state space is denoted by $\cX$ and similarly for others.
  The state space of the composite system of two quantum systems  $X$ and $Y$ is their tensor product space $\cX\ot\cY$, and similarly for multiple systems.
  For a multi-partite state, \eg~$\rho_{XYZ} \in \density{\cX \ot \cY \ot \cZ}$, its reduced state on some subsystem is represented by the same state with the corresponding subscript.
  For example, the reduced (marginal) state on system $\cX$ of $\rho_{XYZ}$ is $\rho_X=\tr[\cY\cZ]{\rho_{XYZ}}$, where $\tr[\cY\cZ]{\cdot}$ denotes the \emph{partial trace} operation over the system $Y \ot Z$.
  That is, $\tr[\cY\cZ]{\ket{x_1}\bra{x_2}\ot \ket{y_1}\bra{y_2}\ot\ket{z_1}\bra{z_2}}=\ket{x_1}\bra{x_2}\tr{\ket{y_1}\bra{y_2}\ot\ket{z_1}\bra{z_2}}$, where $\ket{x_i}, \ket{y_i},\ket{z_i}$ for $i=1, 2$ are vectors in $\ball{\cX},\ball{\cY},\ball{\cZ}$, respectively.
  It can be verified that if $\rho_{XYZ}\in\density{\cX\otimes \cY\otimes \cZ}$, then $\rho_X\in\density{\cX}$.
  When all subscript letters are omitted, the notation represents the original state (\eg~$\rho=\rho_{XYZ}$).
  Any bipartite state $\rho_{XY}$ is called a \emph{product} state if and only if $\rho_{XY}=\rho_{X} \ot \rho_{Y}$.

  If a system $X$ is a discrete random variable $X$ with distribution $p_x=\Pr[X=x]$, it can be represented by a density operator $\rho$ over a state space $\cX$ so that $\rho= \sum_x p_x \ketbra{x}$ where $\{\ket{x}\}$ is an orthonormal basis of $\cX$.
  When restricted to the basis $\ket{x}$, we will say that the system $X$ is classical.
  A \emph{classical-quantum-}state, or \emph{cq-state} $\rho \in \density{\cX \ot \cY}$ indicates that subsystem $X$ is classical and subsystem $Y$ is quantum.
  We use lower case letters to denote specific values assigned to the classical part of a state.
  Then a cq-state can be represented by $\rho_{XY}=\sum_x p_x \ketbra{x} \ot \rho_Y(x)$, where $p_x=\pr{X=x}$ and $\rho_Y(x)\in\density{\cY}$.
  It is also easy to see that the marginal state $\rho_{Y}$ is $\sum_x p_x \rho_Y(x)$.

\item \textbf{Quantum measurements.}
  Let $\Sigma$ be a finite nonempty set of {\em measurement outcomes}.
  A {\em positive-operator valued measure (POVM)} on the state space $\cX$ with outcomes in $\Sigma$ is a collection of positive semidefinite operators $\{\Pi_a : a \in \Sigma\}$ such that $\sum_{a\in \Sigma} \Pi_a=\I_\cX$.
  When this POVM is applied to a quantum state $\rho\in \density{\cA}$, the probability of obtaining outcome $a \in \Sigma$ is  $\ip{\rho}{\Pi_a}$.
  If outcome $a$ is observed, the quantum state $\rho$ will collapse to the state $\sqrt{\Pi_a}\rho\sqrt{\Pi_a} / \ip{\rho}{\Pi_a}$, where $\sqrt{A}= \sum_i \sqrt{\lambda_i}\ket{a_i}\bra{a_i}$ if $A$ is a normal operator that admits a spectral decomposition $A=\sum_i {\lambda_i}\ket{a_i}\bra{a_i}$.

\item \textbf{Matrix Norms.}
  The \emph{trace norm} of $X\in \lin{\cX}$ is defined as
  \[\textstyle\trdist{X}=\tr{\sqrt{X^{\dag}X}}.\]
  One important measure on the distance between two quantum states $\rho, \sigma\in \density{\cX}$ is the \emph{trace distance} $T(\rho, \sigma)$, defined by
  \begin{align}\label{eq:trace dist}
    T(\rho,\sigma) =\frac{1}{2}\trdist{\rho-\sigma}.
  \end{align}
  The following inequality says that the trace distance of two quantum states is an upper bound on the difference of their probabilities of obtaining the same measurement outcome~\cite{NC00}:
  \begin{align}
    T(\rho,\sigma)=\max_{0\leq\Pi\leq \I} \tr{\Pi (\rho-\sigma)}. \label{eq:trace-distance-bound}
  \end{align}

  The \emph{operator norm} of $X \in\lin{\cX}$ is
  \begin{align*}
    \opnorm{X} &= \sup\left\{\norm{X\ket{v}}_2 : \ket{v} \in \cX \mbox{ with } \norm{\ket{v}}_2=1\right\}.
  \end{align*}
  When $X$ is Hermitian, the operator norm of $X$ coincides with the following quantity.
  \begin{align}
    \opnorm{X} &= \sup_{\rho \in \density{\cX}} |\ip{X}{\rho}| \label{eqn:opnorm}\\
    &= \sup_{\ket{\psi} \in \ball{\cX}} |\bra{\psi}X\ket{\psi}|\notag\\
    &=\lambda_{\max}(X),\notag
  \end{align}
  where $\lambda_{\max}(X)$ is the largest eigenvalue of $X$.
  Once we fix an orthonormal basis $\{\ket{(i)}\}$ of $\cX$, we can consider the \emph{max norm} of any $X \in \lin{\cX}$, defined as
  $\norm{X}_{\max}= \max_{i,j} |X_{ij}|$ where $X_{ij}= \bra{(i)}X\ket{j}$.
  We can connect $\norm{X}_{\max}$ to $\opnorm{X}$ by the following inequality.
  \begin{align} \label{eqn:max_norm}
    \opnorm{X} \leq \dim(\cX) \cdot \norm{X}_{\max}.
  \end{align}
  For readers who are interested in the operational interpretation of these norms, we refer them to, for example,~\cite{HJ86}.

\item \textbf{Quantum Circuits.}
  The evolution of a closed quantum system $X$ is described by a unitary operator $U\in\lin{\cX}$, \ie~the operator $U$ satisfying $UU^\dag=U^\dag U=\I_\cX$.
  The quantum system then evolves from state $\rho\in\density{\cX}$ to $U\rho U^{\dag}\in\density{\cX}$.
  If $\rho=\ketbra{\psi}$, then $U\rho U^{\dag}=U\ketbra{\psi} U^{\dag}$ with $U\ket{\psi}\in \ball{\cX}$.
  Herein we consider a multiple system, where each subsystem is a two-dimensional quantum system $\bbC^2$ with an ordered computational basis $\{\ket{0},\ket{1}\}$.
  A quantum state in $\density{\bbC^2}$ is called a \emph{qubit} (quantum bit), as opposed to a classical bit $0$ or $1$.
  Thus an $m$-qubit state space is $\bbC^{2^m}$ with a computational basis $\{\ket{(i)}: i=0,\dots, 2^{m}-1\}$.
  Simple unitary operators that act non-trivially on a constant number of qubits are called \emph{elementary quantum gates}.
  A set of elementary quantum gates is called \emph{universal} if any unitary operator can be approximated arbitrarily close by a composition of gates from this set.
  Let us fix one universal gate set for the remaining of this paper.

  Let $\cW=\cX \ot \cA=\mathbb{C}^{2^m}$ denote the work space of a quantum circuit $C$, which is an $m$-qubit space that consists of both an $\ell$-qubit input space $\cX=\mathbb{C}^{2^\ell}$, taking some quantum/classical input $\rho\in\density{\cX}$, and some $m-\ell$ ancilla qubits that are initialized to some state $\tau\in\density{\cA}$.
  Depends on the model, $\tau$ could be $\ket{0_{\cA}}$ or any quantum state.
  A quantum circuit $C$ is a sequence of elementary quantum gates from the universal gate set, followed by some measurements. (In general, measurements can be deferred to the end of quantum circuits \cite{NC00})
  That is, $C$ applies a unitary $U_C = U_1U_2\cdots U_s$ where $U_i$ denotes the $i$th gate and $s$ is the number of elementary quantum gates, and the performs some measurements.
  We say the size of the quantum circuit $C$ is $s$.
  The number of quantum circuits of size $s$ is $s^{O(s)}$.

\item \textbf{Quantum Distinguisher.}
  In cryptography, we usually have a circuit with binary output as a distinguisher to determine how close two random variables are respect to the circuit.
  Here we define the quantum analogue.
  A \emph{quantum distinguisher} is a quantum circuit with binary measurement outcome 0 or 1.
  Without loss of generosity, we assume that after applying a unitary $U_C$.
  That is, we measure $\rho'= U_C\left(\rho\ot \tau\right) U_C^\dag$ according to the POVM $\{\Pi_0, \Pi_1\}$, where $\Pi_i= \ketbra{(i)}\ot \I_{2^{m-1}}$.
  Thus
  \begin{align*}
    \Pr[C \text{ outputs $i$ on input $\rho$}]
    =&\ip{\rho'}{\Pi_i}\\
    =& \ip{U_C\left(\rho\ot \tau\right) U_C^\dag}{\ketbra{(i)}\ot \I_{2^m-1}}\\
    =&\ip{\rho\ot \tau }{U_C^\dag\left(\ketbra{(i)}\ot \I_{2^{m-1}}\right)U_C}\\
    =& \ip{\rho }{\Pi_i'},
  \end{align*}
  where
  \begin{equation} \label{eqn:circuit-povm}
    \Pi_i = \tr[\cA]{\I_\cX \ot \tau \left (U^\dagger_C (\ketbra{(i)} \ot \I_{-1}) U_C \right) \I_\cX \ot \tau }, \forall i \in \zo,
  \end{equation}

  Consequently, this quantum circuit is equivalent to perform a POVM $\{\Pi_0', \Pi_1'\}$ on the input space $\cX$ as above.
  For our purpose, a \emph{quantum distinguishers} will be considered as a binary POVM on the input space $\cX$.
  In this case, since $\Pi_0 + \Pi_1 = \I$, the POVM can be fully determined by $\Pi_1$.
  For convenience, we use the term {\em BPOVM} to represent the positive-operator $\Pi_1$ in a POVM.
  Especially, we can describe a quantum distinguisher by a BPOVM and vice versa.

  One can easily generalize the binary output to larger domains.
  In that case, any quantum circuit can still be effectively deemed as a general POVM with a large outcome set.
  When the input state is a product state of a classical input $x$ and a quantum input $\rho$, we abuse the notation as follows
  \[C(x, \rho)= C(\ketbra{x}\ot\rho).\]

  We also consider more general quantum circuits that output general quantum states.
  These circuits can be deemed as mappings from $\density{\cX}$ to $\density{\cY}$, where $\cX$ is the input space and $\cY$ is the output space.
  (In general, these mappings are called super-operators from $\lin{\cX}$ to $\lin{\cY}$.)
  Similar to quantum distinguishers, a general quantum circuit $C$ applies a unitary $U_C$ on the space $\cW=\cX \ot \cA$ consisting of an input and ancillas, and perform some measurements on $\cB$.
  Then it outputs a state in space $\cY$ where $\cW=\cY \ot \cB$ is the decomposition of the space after applying $U_C$.
  We abuse the notation again for convenience as
  $$\rho\mapsto C(\rho)\in \density{\cY}$$
  for input $\rho \in \density{\cX}$. 

\end{trivlist}

\subsection{Information-Theoretic Notions}


We will first define the relative min-entropy (a.k.a max-divergence) of two quantum states, which will be used to define the conditional quantum min-entropy \cite{RW05}.
Relative min-entropy can be seen as a distance between two quantum states.
The distance measures, in logarithm, of how much more likely an event happens for one state than the other.

\begin{definition}[Quantum relative min-entropy]\label{defn:relative-min-entropy-psd}
  Let $\rho$ and $\sigma$ be two density operators  on a space $\cH$. The relative min-entropy between two quantum states $\rho$ and $\sigma$ is defined as
  \begin{align*}
    D_{\infty}(\rho\|\sigma) \eqdef \inf\{\lda\in\mathbb{R}: \rho \leq 2^{\lda}\sigma\}.
  \end{align*}
\end{definition}

Equivalently, quantum relative min-entropy can be defined in an operational way using binary measurement.
\begin{proposition}\label{prop:relative-min-operational}
  Let $\rho$ and $\sigma$ be density operators on a state space $\cH$. Then
  \begin{align*}
    D_{\infty}(\rho\|\sigma) = \inf\{\lambda\in\mathbb{R}: \forall A:\density{\cH}\to\bin,  \ex{A(\rho)} \leq 2^{\lambda}\ex{A(\sigma)} \}.
  \end{align*}
\end{proposition}
\begin{proof}
  It suffices to show that for $\gamma>0$, $\gamma\sigma \geq \rho$ if and only if $ \gamma \ex{A(\sigma)}\geq \ex{A(\rho)}$ for every quantum circuit $A$.
  Let $\Pi$ be the BPOVM of $A$, then $\ex{A(\gamma \sigma -\rho)}= \ip{\Pi}{\gamma \sigma-\rho}$.

  If $\gamma\sigma-\rho \geq 0$, then $\ip{\Pi}{\gamma \sigma-\rho}\geq 0$, since $\ip{A}{B}\geq 0$ for $A,B\geq 0$.
  On the other hand, suppose $\ip{\Pi}{\gamma \sigma-\rho}\geq 0$ for any BPOVM $\Pi$.
  Then particularly, take $\Pi = \ketbra{\psi}$, $\ip{\ketbra{\psi}}{\gamma \sigma-\rho} = \bra{\psi}\gamma \sigma-\rho\ket{\psi}\geq 0$ for any $\ket{\psi}\in\ball{\cH}$.
  This concludes the proof.
\end{proof}

The definition of quantum relative-min entropy agrees with the definition of classical max divergence definition when the two quantum states are equivalent to classical random variables.
\begin{proposition}\label{prop:relative-min-match}
  If $\rho$ and $\sigma$ are mixed quantum states corresponding to two classical random variables $X$ and $Y$ respectively. Then $$D_{\infty}(\rho\|\sigma) = \log \max_{x\in\Supp(Y)}(\Pr[X = x]/\Pr[Y = x]).$$
\end{proposition}
\begin{proof}
  By the assumption, we can write $\rho = \sum_{x}p_i\ketbra{x}$ and $\sigma = \sum_{x}q_i\ketbra{x}$ for $\Pr[X = x] = p_x, Pr[Y = x] = q_x$.
  Then
  \[D_{\infty}(\rho\|\sigma) = \inf_\lambda\{\forall i\;\;p_i \leq 2^{\lda} q_i\} = \log \max_i {p_i/q_i} = \log \max_x(\Pr[X = x]/\Pr[Y = x]).\]
\end{proof}

\begin{definition}[Conditional quantum min-entropy]\label{def:min-entropy2}
  Let $\rho = \rho_{XB} \in\density{\cX\ot\cB}$ be a density operator describing a bipartite quantum system $(X, B)$.
  The min-entropy of system $X$ conditioned on system $B$ is defined as
  \begin{align*}
    \Hmin(X|B)_{\rho} \eqdef \log|\cX| -\inf_{\sigma_{B}\in\density{\cB}}
    \left\{
      D_{\infty}\left(\rho_{XB}\middle\|\frac{1}{|\cX|}\I_{X}\otimes \sigma_B\right)
    \right\}.
  \end{align*}

\end{definition}
\begin{proposition}\label{prop:min-match}
  If $X$ and $B$ are classical random variables, then
  \[\Hmin(X|B)_{\rho} = \frac{1}{\log\sum_{b}\max_{x}p_{xb}} = \Hmin(X|B).\]
\end{proposition}
\begin{proof}
  Since $X, B$ are classical random variables, we abuse the notation $\cX, \cB$ to be the finite spaces that $X$ and $B$ are distributed over, respectively.
  Let $U_{\cX}$ be the uniform distribution over the set $\cX$.
  Then by Definition~\ref{def:min-entropy2},
  \begin{align*}
    \Hmin(X|B)_{\rho} &= \log|\cX| - \inf_{Q:\mbox{ dist.~over }\cB}\{D_{\infty}(X, B\|U_{\cX}, Q)\}\\
    &= \log|\cX| - \log\inf_{Q:\mbox{ dist.~over }\cB}\left\{\max_{x, b}\frac{p_{xb}}{\Pr[(U_{\cX}, Q) = (x, b)]}\right\}\\
    &= - \log\inf_{\sum_b q_b = 1}\left\{\max_b \frac{\max_x p_{xb}}{q_b}\right\},
  \end{align*}
  where $q_b = \Pr[Q = b]$.
  The minimum happens when $(\max_x p_{xb})/q_b$ are equal for all $b\in\cB$ (Can be seen by Lagrange multiplier).
  Then we get
  \[\Hmin(X|B)_{\rho} = \inf_{\sum_b q_b = 1}\left\{\max_b \frac{\max_x p_{xb}}{q_b}\right\} = -\log\left(\sum_{b}\max_x p_{xb}\right),\]
  which is exactly the definition of average min-entropy in the classical case \cite{DodisORS08}.
\end{proof}
While min-entropy in Definition~\ref{def:min-entropy2} describes the average case that one can guess about a random source,
we may also define the worst-case min-entropy as in the classical case~\cite{RW05}.
\begin{definition}[Worst-case conditional quantum min-entropy]\label{def:worst-min}
  Let $\rho = \rho_{XB} \in\density{\cX\ot\cB}$ be a density operator describing a bipartite quantum system $(X, B)$.
  The \emph{worst-case} min-entropy of system $X$ conditioned on system $B$ is defined as
  \begin{align*}
    \Hmin^{\mathrm{wst}}(X|B)_{\rho} \eqdef \log|\cX| -
      D_{\infty}\left(\rho_{XB}\middle\|\frac{1}{|\cX|}\I_{X}\otimes \rho_B\right),
  \end{align*}
  where $\rho_B=\tr[\cX]{\rho_{XB}}$.
\end{definition}
If $X$ and $B$ are classical random variables, Definition~\ref{def:worst-min} agrees with the classical definition of the worst-case min-entropy $\Hmin^{\mathrm{wst}}(X|B) = -\log\max_{x, b}\Pr[X = x|B = b]$.
Also it is clear that the two entropies are equal when system $B$ is trivial: $ \Hmin^{\mathrm{wst}}(X)_{\rho}= \Hmin(X)_{\rho}$.
The following chain rule provides a connection between these two entropies.
\begin{proposition}\label{prop:chain-worst-min}
  Let $\rho = \rho_{XB} \in\density{\cX\ot\cB}$ be a density operator describing a bipartite quantum system $(X, B)$.
  Suppose $B$ is an $\ell$-qubit system.
  Then
  \[ \Hmin(X|B)_{\rho}\geq    \Hmin^{\mathrm{wst}}(X,B)_{\rho}- \ell. \]
\end{proposition}
\begin{proof}
  Let $\lambda=\log|\cX| - \Hmin^{\mathrm{wst}}(X,B)_\rho$. That is,
  \[\rho_{XB}\leq 2^{\lambda}\frac{\I_X}{|\cX|}\ot\I_B = 2^{\lambda+\ell}\frac{\I_X}{|\cX|}\ot\frac{\I_B}{2^\ell}. \]
  Thus by Definition~\ref{def:min-entropy2},    $\Hmin(X|B)_{\rho}\geq    \Hmin^{\mathrm{wst}}(X,B)_{\rho}- \ell.$

\end{proof}

Another way to define min entropy is to use guessing probability.
Here we only consider the case that $\rho_{XB}$ is a cq-state: $\rho_{XB} = \sum_{x}p_x\ketbra{x}\otimes \rho^x_B$ for $\rho^x_B\in \density{\cB}$.
The probability of guessing $X$ correctly given $B$ by a given quantum circuit $C$ is 
\begin{align*}
  \Pguess_C (X|B)_{\rho} = \sum_{x} \Pr[X = x]\ip{\Pi_x}{\rho^{x}_B},  \label{eq:pguess_rho}
\end{align*}
where $\{\Pi_x\}$ is the effective POVM for $C$, demonstrating the guessing strategy.
Accordingly,  the probability of guessing $X$ correctly given $B$ is defined as
\begin{equation}
\Pguess (X|B)_{\rho} = \max_{C} \Pguess_C (X|B)_{\rho},
\end{equation}
where the maximization is taken over arbitrary quantum circuits $C$ of unbounded size.
As in the purely classical case~\cite{DodisORS08}, the guessing probability captures the conditional min-entropy of $X$ given~$B$:
\begin{lemma}[\cite{KonigRS09}] \label{lemma:guess-def-equivalent}
Suppose $\rho_{XB}$ is a cq-state on the space $\cX\ot\cB$. Then
$$\Hmin(X|B)_{\rho} = -\log(\Pguess(X|B)).$$
\end{lemma} 

\subsection{Leakage Chain Rules}

One of our main result is the leakage chain rule for computational quantum min-entropy.
The information-theoretic version of the Leakage Chain Rule is a necessary step in our proof.

\begin{theorem}[{\cite[Lemma~13]{WTHR11}} Leakage chain rule for quantum min-entropy]\label{thm:quantum-chain-rule}
  Let $\rho = \rho_{XZB}$ be a state on the space $\cX\ot\cZ\ot\cB$.
  Let $d= \min\{\dim(\cX\ot \cZ), \dim(\cB)\}$ and $\ell=\log \dim(\cB)$.
  Then
  \[\Hmin(X|ZB)_{\rho} \geq \Hmin(X|Z)_\rho - 2\ell.\]
\end{theorem}

We remark that Theorem~\ref{thm:quantum-chain-rule} is tight.
In the case that the dimensions of $\cX$ and $\cB$ are both $2^{\ell}$, $X$ and $B$ are maximally entangled, it is easy to verify that $\Hmin(X|B)_{\rho} = -\ell$ and $\Hmin(X) = \ell$.
Also, the factor $2$ is crucial for the application in proving the lower bound of superdense coding.
For more detail, see~\cite{chainrulesuperdense}.

On the other hand, when $B$ is separable from $(X, Z)$, the loss of the entropy is at most $\ell$ instead of $2\ell$.
Specifically, when the leakage $B$ is classical, there is no entanglement between $Z$ and $B$.

\begin{theorem}[Chain rule for quantum min-entropy of separable states]\label{thm:quantum-chain-rule-sep}
  Let $\rho = \rho_{XZB}$ be a separable state on the space $(\cX\ot\cZ)\ot\cB$.
  Namely, $\rho_{XZB}=\sum_{k} p_{k}\rho_{XZ}^k  \ot \rho_B^k$.
  Suppose $B$ is an $\ell$-qubit system.
  Then
  \[\Hmin(X|ZB)_{\rho} \geq \Hmin(X|Z)_\rho - \ell.\]
\end{theorem}

\subsection{Matrix Multiplicative Weights Update Method}

Matrix Multiplicative Weights Update (MMWU) Method is a generalization of multiplicative weights algorithm.
It is a fundamental technique in learning theory and provides a way to approximate a unknown function efficiently in an iterative way.
Numerous applications on faster algorithms for SDPs~\cite{AHS12, Kale07} and classical simulations of quantum complexity classes~\cite{JainW09, JainJUW09, Wu10, GW13}.
In the matrix generalization, a weight vector is generalized to a weight matrix.
We will use the following abstraction of the MMWU method.

\begin{theorem}[Matrix multiplicative weights update method~{\cite[Thm.~10]{Kale07}}]
  \label{thm:mmwu}

  Fix $\eta\in(0,1/2)$ and let $L^{(1)},\dots,L^{(T)}$ be arbitrary $d\times d$ ``loss'' matrices with $-c_1\cdot\I_d\leq L^{(t)}\leq c_2\cdot\I_d~\forall t\in[T]$.
  Define $d\times d$ weight matrices $W^{(1)}, \dots, W^{(T)}$ and density operators $\rho^{(1)}, \dots, \rho^{(T)}\in\density{\bbC^d}$ by the following update algorithm.
  \begin{enumerate}
    \item Initialization: $W^{(1)} = \I_d$
    \item For $t = 1, \dots, T$,
    \begin{enumerate}
      \item Normalization: $\rho^{(t)}=W^{(t)}/\tr{W^{(t)}}$
      \item Update: $W^{(t+1)} = W^{(t)}\cdot \exp\left(-\eta L^{(t)}\right) $
    \end{enumerate}
  \end{enumerate}

  Then for all density operators $\sigma\in\density{\bbC^d}$, we have
  \[
    \frac{1}{T}\sum_{t=1}^T \ip{\rho^{(t)}}{L^{(t)}}
    \leq
    \frac{1}{T}\sum_{t=1}^T \ip{\sigma}{L^{(t)}} + (c_1+c_2)\left(\eta + \frac{\ln d}{\eta T}\right).
  \]
  Particularly, let $c_1$ and $c_2$ be constants, and $\eta = \sqrt{\ln d / T}$, we have
  \[
    \frac{1}{T}\sum_{t=1}^T \ip{\rho^{(t)}}{L^{(t)}}
    \leq
    \frac{1}{T}\sum_{t=1}^T \ip{\sigma}{L^{(t)}} + O\left(\sqrt{\frac{\log d}{T}}\right).
  \]
\end{theorem}

\subsection{Quantum Tomography}

In a quantum tomography problem, one wants to learn the behavior or even a description of a quantum circuit or quantum state.
In our applications, we specifically want to solve the following problem.
We show more details and the proofs of related lemmas in Appendix~\ref{app:tomography}.

\begin{definition}[$\QCktTom$ Problem]\label{def:qckt-tomography}
  The $\QCktTom(s, d, \eps, \gamma)$ problem is a computational problem defined as follows:
  \begin{itemize}
    \item Input: a description of a quantum circuit $C:\density{\bbC^{d}}\to\zo$ of size-$s$, and an error parameter $0< \eps <1$.
    \item Task: let $\Pi$ be the corresponding BPOVM of $C$.
      Output an explicit description (as matrices) of BPOVM $\tilde{\Pi}$ such that $\opnorm{\Pi-\tilde{\Pi}}\leq \eps$ with probability $1-\gamma$.
  \end{itemize}
\end{definition}

\begin{lemma}\label{lemma:qkct-tomography}
  There exists a (uniform) quantum algorithm that solves the $\QCktTom(s, d, \eps, \gamma)$ Problem in time
  $\poly(s, d, 1/\eps, \log(1/\gamma))$.
\end{lemma}

\section{Quantum Indistinguishability / Pseudorandomness}\label{sec:quantum-indistinguishability}


\subsection{Quantum Indistinguishability}
Computational indistinguishability is a fundamental concept of computational complexity and pseudorandomness.
It provides a relaxed way to describe the similarity of two random objects.
Informally, computational indistinguishability only requires that two random objects cannot be distinguished by efficient algorithms/circuits.
Two objects may be indistinguishable by a bounded algorithm even if their entropies difference is significant.

We use $(s, \eps)$-indistinguishability to describe two objects cannot be distinguished with advantage more than $\eps$ by all $s$-size circuits with ancilla bits initialized to $\ket{0}$s.
Contrary, we use $(s, \eps)^*$-indistinguishability when the ancilla bits can be initialized as any quantum states.
Formally, we have the following definitions.

\begin{definition}
  Quantum states $\rho$ and $\sigma$ on $\cH$ are $(s, \eps)$-quantum-indistinguishable if for all size-$s$ quantum distinguishers $D$ with ancilla bits initialized to $\ket{0}$s,
  \[|\ex{D(\rho)} - \ex{D(\sigma)}| \leq \eps.\]
  Moreover, we say that $\rho$ is an \emph{$(s, \eps)$-quantum-pseudorandom} state if $\rho$ is $(s, \eps)$quantum-indistinguishable from the maximally mixed state on $\cH$.
\end{definition}

\begin{definition}
  Quantum states $\rho$ and $\sigma$ on $\cH$ are $(s, \eps)^*$-quantum-indistinguishable if for all size-$s$ quantum distinguishers $D$ with arbitrary ancilla bits $\tau$,
  \[|\ex{D(\rho)} - \ex{D(\sigma)}| \leq \eps.\]
  Moreover, we say that $\rho$ is an \emph{$(s, \eps)^*$-quantum-pseudorandom} state if $\rho$ is $(s, \eps)^*$quantum-indistinguishable from the maximally mixed state on $\cH$.
\end{definition}

If we consider the indistinguishability in asymptotic settings,
in the case that the algorithm with classical advice, equivalently, it can be described as an ensemble of nonuniform quantum circuits $\{C_n\}_{n\in\bbN}$ where their ancilla bits are initialized to be all $\ket{0}$s.
On the other hand, the algorithm with quantum advice can be described as an ensemble of quantum circuits $\{C_n\}_{n\in\bbN}$ where their ancilla bits are initialized to be $\{\tau_n\}_{n\in\bbN}$.
Note that $\tau_n$ can only depend on $n$ and could be a state that cannot be generated efficiently.


\begin{definition}
  Let $s:\bbN\to\bbN$ and $\eps:\bbN\to\bbR$ be two functions.
  Let $\{\rho_n\}_{n\in\bbN}$ and $\{\sigma_n\}_{n\in\bbN}$ be two quantum state ensembles where $\rho_n, \sigma_n\in\density{\bbC^{2^n}}$.
  We say $\{\rho_n\}_{n\in\bbN}$ and $\{\sigma_n\}_{n\in\bbN}$ are $(s, \eps)$-quantum-indistinguishable (resp., $(s, \eps)^*$-quantum-indistinguishable), if for every non-uniform ensemble of quantum circuit with classical (resp. quantum) advice $\{C_n\}_{n\in\bbN}$, and the size of $C_n$ is at most $s(n)$, we have for all $n\in\bbN$,
  \[ |\pr{C_n(\rho_n) = 1} - \pr{C_n(\sigma_n) = 1}| \leq \eps(n). \]
\end{definition}

Now lets consider the uniform setting.
A \emph{quantum polynomial time algorithm} is a uniform ensemble of quantum circuit $\{C_n\}_{n\in\bbN}$, which can be generated by a Turing machine in $\poly(n)$ time.
Without loss of generosity, we can still assume the ancilla bits are initialized to all $\ket{0}$s.
If two ensembles are indistinguishable by all polynomial-time algorithms, they are computationally indistinguishable. Formally,

\begin{definition}
  Let $\{\rho_n\}_{n\in\bbN}$ and $\{\sigma_n\}_{n\in\bbN}$ be two quantum state ensembles where $\rho_n, \sigma_n\in\density{\bbC^{2^n}}$.
  We say $\{\rho_n\}_{n\in\bbN}$ and $\{\sigma_n\}_{n\in\bbN}$ are quantum-indistinguishable, if for every quantum polynomial time algorithm $\{C_n\}_{n\in\bbN}$ and every polynomial function $p(\cdot)$, for all but finitely many $n$, we have
  \[ |\pr{C_n(\rho_n) = 1} - \pr{C_n(\sigma_n) = 1}| \leq \frac{1}{p(n)}. \]
\end{definition}

\subsection{Pseudorandom Pure State}

Classically, there exists a random variable ensemble $\{X_n\in\zo^n\}_{n\in\bbN}$ that is far from the uniform distribution $U_n$ but are computationally indistinguishable from $U_n$ by any algorithm runs in $\poly(n)$ time.
Note that the min-entropy of such $X_n$ strings is at least $\omega(\log n)$.
In the quantum setting, we have a pseudorandom state as an analogue.
We will state the results in non-uniform settings.
Also, we fix the security parameter for simplicity.
The following result of Bremner, Mora, and Winter~\cite{BMW08} (stated in our language) says that a pseudorandom object can be a pure state, which has \emph{zero} entropy.

\begin{theorem}[\cite{BMW08}]
  For all $s\in\mathbb{N}$ and $\eps > 0$, there exists an $(s, \eps)$-quantum-pseudorandom pure state $\rho = \ketbra{\psi}\in \density{\bbC^{2^m}}$ with $m = O(\log(s/\eps))$ that is $(s, \eps)$-pseudorandom.
  Furthermore, such $\rho$ can be obtained by a uniformly random pure state $\ketbra{\psi}\in \density{\bbC^{2^m}}$ with all but $2^{-\Omega(2^m)}$ probability.
\end{theorem}

We emphasize that the existence of a pseudorandom pure state only holds when we consider distinguishers without quantum advice.
Otherwise, for a pure state $\rho$, one can hardwire the same state as an advice, then it can be distinguished from a maximally mixed state by using Swap Test.

Here we provide a simpler sampling method to show the existence of pseudorandom pure state, whose coefficients are all $2^{-m/2}$ or $2^{m/2}$.

\begin{theorem}\label{thm:pure-indist}
  For all $s\in\mathbb{N}$ and $\eps > 0$, there exists $m = O(\log(s/\eps))$ such that, if we uniformly sample $(\alpha_1, \dots \alpha_{2^m})$ from $\{-2^{-m/2}, -2^{-m/2}\}^{2^m}$ and let $\rho = \sum_{i = 1}^{2^m}\alpha_i \ket{(i)}$, then with all but $2^{-\Omega(2^m)}$ probability, $\rho$ is an $(s, \eps)$-quantum-pseudorandom pure state.
\end{theorem}

\begin{proof}
  Let $A:\density{\bbC^{2^m}}\to\bin$ be some fixed quantum distinguisher corresponding to a BPOVM $\Pi$.
  Then
    \[\pr{A(\rho_\mix^{(m)}) = 1} = \frac{1}{2^m}\ip{\Pi}{\I_{2^m}}=\frac{1}{2^m} \sum_{(i)}  \bra{(i)}\Pi\ket{(i)},\]
  where the probability is over the measurement taken by $A$.
  For a fixed $\rho = \ketbra{\psi}$ with $\ket{\psi} = \sum_{(i)}\alpha_i \ket{(i)}$, we have
  \begin{align*}
    \Pr[A(\rho) = 1] = &\ip{\Pi_1}{\ketbra{\psi}}
    = \sum_{i, j\in[2^m]}\alpha_i^*\alpha_j \bra{(i)}\Pi_1\ket{j} = \sum_{(i)}|\alpha_i|^2 \bra{(i)}\Pi_1\ket{(i)}+\sum_{i\neq j}\alpha_i^*\alpha_j \bra{(i)}\Pi_1\ket{j},
  \end{align*}
  Taking expectation over $\alpha = (\alpha_1, \dots, \alpha_{2^m})$, we have
  \begin{align*}
    \ex[\alpha]{\Pr[A(\rho) = 1]}
    =& \ex[\alpha]{\sum_i|\alpha_i|^2 \bra{(i)}\Pi_1\ket{(i)}} + \ex[\alpha]{\sum_{i\neq j}\alpha_i^*\alpha_j \bra{(i)}\Pi_1\ket{j}} \\
    =& \frac{1}{2^m}\sum_{(i)}\bra{(i)}\Pi_1\ket{(i)} = \pr{A(\rho_\mix^{(m)}) = 1}.
  \end{align*}
  For a fixed distinguisher $A$, define the function on a hypercube $f:\{-2^{-m/2}, -2^{-m/2}\}^{2^m}\to {[0, 1]}$ as
  \[f(\alpha_1, \cdots, \alpha_{2^m}) = \pr{A(\rho) = 1} = \alpha\Pi\alpha^{\dagger}.\]
  Now we are going to show the concentration using Talagrand's inequality.
  To that end, we first find the \emph{Lipschitz constant} $\eta$ of the function $f$.
  \begin{align*}
    |f(\alpha)-f(\beta)|
    =   & | \alpha\Pi\alpha^{\dagger} - \beta\Pi\beta^{\dagger}|\\
    \leq& | \alpha\Pi\alpha^{\dagger} - \alpha\Pi\beta^{\dagger}
           +\alpha\Pi\beta^{\dagger} - \beta\Pi\beta^{\dagger}|\\
    \leq& \norm{\alpha\Pi}_2 \cdot \norm{\alpha-\beta}_2
    +     \norm{\alpha^{\dagger}-\beta^{\dagger}}_2 \cdot \norm{\Pi\beta^{\dagger}}_2\\
    \leq & 2 \opnorm{\Pi} \norm{\alpha-\beta}_2\\
    \leq & 2 \norm{\alpha-\beta}_2,
  \end{align*}
  where the second inequality follows from the Cauchy-Schwartz inequality of the Hilbert-Schmidt inner product of two operators.
  Therefore the Lipschitz constant $\eta$ of the function $f$ is at most $2$.
  Also, $f$ is a convex function.
  \begin{align*}
    f((\alpha+\beta)/2) &\leq f((\alpha+\beta)/2) + f((\alpha-\beta)/2)
    = \frac{1}{2} \left(
      \alpha\Pi\alpha^{\dagger}
      + \beta\Pi\beta^{\dagger}
    \right)
    = \frac{1}{2}(f(\alpha) + f(\beta))
  \end{align*}
  Now we are ready to apply the Talagrand's concentration inequality:
  Let $\alpha_1, \cdots, \alpha_M$ be independent random variables with $|X_i| \leq K$ and $f:\bbR^n\to{\bbR}$ be a $\eta$-Lipschitz convex function, then there exists a constant $C$ such that for all $t$
  \[\pr[\alpha]{|f(\alpha) - \ex{f(\alpha)}|\geq Kt}\leq 2^{-Ct^2/\eta^2}.\]
  Take $M = 2^m$, $K = 2^{-m/2}$, $t = 2^{m/2}\eps$ and $\eta = 2$, we have
  \[\pr[\alpha]{\left|\Pr[A(\rho) = 1] - \pr{A(\sigma_{\mix}^{(m)}) = 1}\right| \geq \eps} = \pr[\alpha]{|f(\alpha) - \ex{f(\alpha)}|\geq \eps}\leq 2^{-\Omega(\eps^2\cdot2^m)}.\]
  There are only $s^{O(s)} = 2^{O(s\log s)}$ many different quantum circuits of size $s$.
  We can choose some $m = O(\log(s/\eps))$ such then by union bound, there exists a quantum state $\rho$ such that for every quantum circuit $A$ of size $s$, we have
  \[\left|\pr{A(\rho) = 1} - \pr{A(\sigma_{\mix}^{(m)}) = 1}\right| \geq \eps.\]
\end{proof}

An interesting follow-up question is that whether we can {\em explicitly} generate pseudorandom pure states, say as the output of a small quantum circuit (with no measurements) on input $\ket{0^n}$ --- which we could think of as a ``seedless'' pseudorandom generator.
If the generator is of polynomial size, then its output cannot be pseudorandom against all polynomial-sized distinguishers, because (measurement-free) quantum computation is reversible.
But if we allow the generator circuit to be larger than the distinguishers then it is conceivable to have a pseudorandom pure state as output.
As aforementioned, in \cite{BHH12, BHH16}, they use probabilistic method to show the existence of a generator circuit of size $n^{11k+9}$ that can fool all $n^{k}$-size quantum distinguishers.
It would be interesting to construct such generators explicitly under plausible (quantum) complexity assumptions.

\section{Quantum Computational Entropy}\label{sec:comp-entropy}
A uniform distribution can be seen as a special case that it has a full amount of (min-) entropy.
Similar to the definition of pseudorandomness, one can naturally generalize the concept of entropy in information theory to {\em pseudoentropies} in computational settings.
The pseudoentropy notions are useful in cryptography and constructions of pseudorandom generators because it suffices to handle adversaries with adversaries with limited computational resources.

As one can expect, there are different ways to define computational (min-) entropies.
In this section, we will explore some of the possible relaxations of min-entropy and show some connections.
First, we introduce the quantum analogue of non-uniform min-max theorem~\cite{VZ13}, which is an elementary tool of proving the relationship between entropies and the Leakage Chain Rule of a computational min-entropy in Section~\ref{sec:leakage-chain-rule}.

\subsection{Quantum Non-Uniform Min-Max Theorem}\label{subsec:min-max}
We begin with von Neumann's Min-Max Theorem for the zero sum game with two players.
Let the strategy spaces of Player~1 and Player~2 be $\cA$ and $\cB$, respectively, and the payoff function be $g:\cA\times\cB\to{[-1, 1]}$.
The theorem says that if for every mixed strategy $A\in\Conv(\cA)$, Player~2 can respond $b\in\cB$ so that the expected payoff $\ex[a\from A]{g(a, b)} \geq p$, then Player~2 has an universal mixed strategy $B\in\Conv(\cB)$ that guarantees the same payoff regardless of the strategy of Player~1.
Namely, for all $a\in\cA, \ex[b\from B]{g(a, b)} \geq p$.
In many applications in cryptography and complexity theory, (\eg~\cite{Impagliazzo95,ReingoldTTV08,DziembowskiP08,GentryW11,VZ12}), people consider the strategy space $\cA$ to be a set of distributions over $\zo^n$.
Moreover, those applications require not only the existence of a universal mixed strategy, but also a strategy with low complexity (measured in the number of pure strategies of Player~2).

In this section, we generalize the classical non-uniform Min-max theorem in~\cite{VZ12} to the quantum setting.
Specifically, the game we consider has the following structure.
The strategy space $\cA$ is in $\density{\bbC^d}$.
The payoff function $g:\Conv(\cA)\times\cB\to{[0, 1]}$ is restricted to the form $g(a, b) = \ip{a}{f(b)}$ where $f$ is a function maps $\cB$ to a $d$-dimension matrix $M$ with $0 \leq M \leq \I_d$.
Note that if we restrict both player 1's strategy and $M$ to be diagonal matrices, then it replicates the above classical definition with $d = 2^n$.

\begin{theorem}[Quantum Non-uniform Min-Max Theorem]\label{thm:min-max}
  Consider the above quantum zero-sum game.
  Suppose that for every mixed strategies $a\in\Conv(\cA)$ of Player 1, there exists a pure strategy $b \in \cB$ such that $g(a, b) \geq p$.
  Then for every $\eps \in (0, 1/2)$, there exists some mixed strategy $\hat{B}$ of Player 2 such that for every strategy $a\in\cA$ of Player 1, $\ex[b\from\hat{B}]{g(a, b)} \geq p-\eps$.
  Moreover, $\hat{B}$ is the uniform distribution over a multi-set $S$ consists of at most $O\left(\log d/\eps^2\right)$ strategies in $\cB$.
\end{theorem}

Notice that we only assume $\Conv(A)$ is contained in $\density{\bbC^d}$ but not equal.
Indeed, for the application in Section~\ref{ssec:comp-min-entropy}, $\Conv(A)$ will be the set of all high entropy quantum states.
As a result, directly applying the MMWU algorithm shown Theorem~\ref{thm:mmwu} by letting $\rho^{(t)} = a^{(t)}$ does not suffice for proving this theorem, because there is no guarantee that the (normalized) weight matrix still belongs $\Conv(A)$ after the update.
To fix that, we have to ``project'' the weight matrix back to the set $\Conv(A)$.

\begin{definition}[KL Divergence of Quantum States]
  Let $\rho$ and $\sigma$ be two density matrices, the KL divergence (relative entropy) between them is defined as
  \[ \KL{\rho}{\sigma} = \tr{\rho(\log \rho - \log \sigma)}.\]
\end{definition}

\begin{definition}[KL Projection]\label{def:KL-projection}
  Let $\sigma$ be quantum states on $\density{\bbC^d}$, which contains a convex set $\cA$.
  $\rho^*$ is a KL projection of $\sigma$ on $\cA$ if
  \[\rho^* = \arg\min_{\rho\in\cA} \KL{\rho}{\sigma}\]
\end{definition}

\begin{proof}
  We consider the following MMWU procedure (modified from ones in~\cite{TRW05,WK12}) to obtain the multi-set $S$.

  \begin{center}
    \begin{pseudocode}[shadowbox]{Procedure}{} \label{pro:mmwu-minmax}
      \begin{enumerate*}
        \item Initially, $W^{(1)} = \I_d$, $a^{(1)} = \frac{1}{d}\I_d$.
        \item For $t = 1, \dots, T$,
        \begin{enumerate*}
          \item Obtain the best strategy of Player~2:
            \[b^{(t)} = \argmax_{b\in\cB} g(a, b) = \argmax_{b\in\cB} \ip{a}{f(b)}.\]
          \item Update the weight matrix: $a^{(t+1)'} = \exp\left(\log a^{(t)} -\eta f(b^{t})\right)$ where $\eta = \sqrt{\ln d/T} < 1$.
          \item Do a normalization to get a density matrix: $a^{(t+1)''}=a^{(t+1)'}/\tr{a^{(t+1)'}}$.
          \item Let $a^{(t+1)}$ be a KL projection of $a^{(t+1)''}$ on $\Conv(\cA)$.
        \end{enumerate*}
        \item Output $S = \{b^{(1), \dots, b^{(T)}}\}$.
      \end{enumerate*}
    \end{pseudocode}
  \end{center}

  We emphasize that the algorithm differs from the one in Theorem~\ref{thm:mmwu} in two places.
  First, as mentioned before, we have the step (d): project the weight matrix back to the convex set.
  Second, in the step (b), instead of multiply the update matrix to $a^{(t)}$, we put $a^{(t)}$ in the exponent.
  Note that they are not equivalent because matrix does not commute in multiplication generally.
  This modification is needed if we want to prove the same property using KL-divergence.

  Considering the above procedure, similar to the Theorem~\ref{thm:mmwu}, we have
  \begin{align*}
    \frac{1}{T}\sum_{t = 1}^{T}\ip{a}{f(b^{(t)})}
    & \geq \frac{1}{T}\sum_{t = 1}^{T}\ip{a^{(i)}}{f(b^{(t)})} - O\left(\sqrt{\log d/T}\right)
  \end{align*}

  Without the projection step, the above inequality was shown
  In~\cite{TRW05,WK12}, they showed the above inequality when the procedure without step (d).
  Even though, when there is the projection step, the proof much follows the ones in~\cite{TRW05,WK12}.
  For the completeness, we prove it as Lemma~\ref{lemma:mmwu-kl-proof}.

  Additional to the inequality, by the fact that in the step 2(b), $b^{(t)}$ is the best strategy against $a^{(t)}$, we have for all $a\in\cA$,
  \begin{align*}
    \frac{1}{T}\sum_{t = 1}^{T}\ip{a}{f(b^{(t)})}
    & \geq \frac{1}{T}\sum_{t = 1}^{T}\ip{a^{(i)}}{f(b^{(t)})} - O\left(\sqrt{\log d/T}\right)\\
    & \geq p - O\left(\sqrt{\log d/T}\right)
  \end{align*}
  Set $T = O(\log d/\eps^2)$, we get
  \[\forall a,~\ex[b\from S]{g(a, b)} = \frac{1}{T}\sum_{t = 1}^{T}\ip{a}{f(b^{(t)})} \geq p - \eps.\]
\end{proof}

\subsection{Computational Min-entropy}\label{ssec:comp-min-entropy}

In the classical setting, a definition of computational entropy was given by H\aa stad \textit{et.~al.}~\cite{HastadILL99}.
It says that a random variable $X$ has HILL pseudo-min-entropy (HILL pseudoentropy for short) at least $k$ if it is indistinguishable from some random variable $Y$ with (true) min-entropy (with $B$ trivial in Definition~\ref{def:min-entropy2}) at least $k$.
Another natural definition of computational entropy is the \emph{Metric pseudo-(min)-entropy} which switches the quantifiers in the definition of HILL pseudoentropy.
That is, $X$ has metric pseudoentropy at least $k$ if, for every efficient distinguisher, there exists a random variable $Y$ with min entropy at least $k$ such that $X$ and $Y$ cannot be distinguished by the distinguisher.
One can also define computational entropies (when $X$ is classical) via guessing probabilities.
Recall that using guessing probability, one can equivalently define the conditional min-entropy (cf. Lemma~\ref{lemma:guess-def-equivalent}).
We can also get a relaxed notion by restricting the complexity of guessing algorithms, and we call it \emph{guessing pseudoentropy}.
Below, we formally define the quantum analogues of those relaxed notions.

\begin{definition}[Conditional (relaxed-)HILL pseudoentropy]
  Let $\rho = \rho_{XB}$ be a bipartite quantum state in $\density{\cX\otimes\cB}$.
  We say $X$ conditioned on $B$ has \emph{$(s, \eps)$-relaxed-HILL pseudoentropy} $\Hhillr_{s, \eps}(X|B)_{\rho} \geq k$ if there exists a bipartite quantum state $\sigma_{XB}$ on $\density{\cX\otimes\cB}$ such that
  (i) $\Hmin(X|B)_{\sigma} \geq k$, and
  (ii) $\rho_{XB}$ and $\sigma_{XB}$ are $(s, \eps)$-indistinguishable.
  In addition, if $\tr[\cX]{\rho_{XB}} = \tr[\cX]{\sigma_{XB}}$, we say $X$ conditioned on $B$ has (regular) \emph{HILL pseudoentropy} $\Hhill_{s, \eps}(X|B)_{\rho} \geq k$.
\end{definition}

As in the definition of conditional relaxed-HILL pseudoentropy~\cite{HsiaoLR07}, we do not require the reduced states $\rho_B$ and $\sigma_B$ being equal in relaxed-HILL pseudoentropy.
In the classical case, the relaxed HILL notion satisfies a chain rule even when a prior knowledge $Z$ is present, while for the regular HILL pseudoentropy, a counterexample exists (under a standard assumption)~\cite{KrennPW13}.
Another remark here is that when the length of $B$ is $O(\log n)$, it is not hard to see that the two definitions are equivalent in the classical case.
However, we do not know whether that is still the case if $B$ is a quantum state of $O(\log n)$ qubits.

\begin{definition}[Conditional (relaxed-)metric pseudoentropy]
  Let $\rho = \rho_{XB}$ be a bipartite quantum state in $\density{\cX\ot\cB}$.
  We say that $X$ conditioned on $B$ has \emph{$(s, \eps)$-relaxed-metric pseudoentropy} $\Hmetricr_{s, \eps}(X|B)_{\rho} \geq k$ if for all size-$s$ quantum distinguisher $D$, there exists a bipartite quantum state $\sigma_{YC}$ on $\density{\cX\ot\cB}$ such that
  (i) $\Hmin(X|B)_\sigma \geq k$ and
  (ii) $|\ex{D(\rho_{XB})} - \ex{D(\sigma_{XB})}| < \eps$.
  In addition, if $\tr[\cX]{\rho_{XB}} = \tr[\cX]{\sigma_{XB}}$, we say $X$ conditioned on $B$ has (regular) \emph{metric pseudoentropy} $\Hmetric_{s, \eps}(X|B)_{\rho} \geq k$.
\end{definition}

\begin{definition}[Guessing pseudoentropy]
  Let $\rho_{XB} = \sum_{x\in\bin^n}p_x\ketbra{x}\otimes\rho^x_B$ be a cq-state.
  We say that $X$ conditioned on $B$ has \emph{$(s, \eps)$-quantum guessing pseudoentropy} $\Hguess_{s, \eps}(X|B)_\rho \geq k$ if for every quantum circuit $A$ of size~$s$, the probability of guessing X correctly given B by the circuit $A$ is $\Pguess_A(X|B) \leq 2^{-k}+\eps$.
\end{definition}

\subsection{Connections between computational notions}\label{subsec:connection}

\paragraph{HILL pseudoentropy v.s. metric pseudoentropy}
In the classical case, it is known that the HILL and metric entropies are interchangeable~\cite{BarakSW03} up to some degradation in the size of distinguishers.
With the equivalence, metric pseudoentropy is a useful intermediate notion to obtain tighter security proof in a number of cases~(\eg~\cite{DziembowskiP08,fuller2015unified}).
Here we will show the analogue transformation in the quantum case.

\begin{theorem}[(relaxed-)HILL $\Leftrightarrow$ (relaxed-)metric]\label{thm:quantum-metric-hill}
  Let $\rho_{XB}$ be a bipartite quantum system in $\density{\cX\ot\cB}$ where $\dim(\cX) = N$ and $\dim(\cB) = L$.
  If $\Hmetric_{s, \eps}(X|B)_{\rho} \geq k$ (resp., $\Hmetricr_{s,\eps}(X|B)_{\rho} \geq k$), then $\Hhill_{s',\eps'}(X|B) \geq k$ (resp., $\Hhillr_{s',\eps'}(X|B) \geq k$), where $\eps' = \eps+\delta$ and $s = s'\cdot O\left((\log N + \log L)/\delta^2\right)$ for any $\delta>0$.
\end{theorem}
\begin{proof}
  For the sake of contradiction, let $\Hhill_{s', \eps'}(X|B)_{\rho} < k$.
  Then for all $\sigma_{XB}\in\density{\cX\ot\cB}$ with $\Hmin(X|B)_{\sigma} \geq k$ and $\rho_{B} = \sigma_{B}$, there exists a quantum distinguisher $D:\density{\cX\ot\cB}\to{\zo}$ of size $s'$ such that
  \[\ex{D(\rho_{XB})} - \ex{D(\sigma_{XB})} > \eps'.\]
  We define the following zero-sum game:
  \begin{itemize}
    \item The strategy space of Player~1 is $\cA = \{\sigma_{XB}\in\density{\cX\ot\cB}:\Hmin(X|B)_{\sigma}\geq k\}$.
    \item The strategy space of Player~2 $\cB$ is a set of all distinguishers $D:\density{\cX\ot\cB}\to{\zo}$ of size $s'$.
    \item For the payoff function $g:\cA\times\cB\to{[0, 1]}$, we first define the auxiliary mapping $f$.
      For an input distinguisher $D\in\cB$, let $\Pi$ be its BPOVM, we let
      \[f(D) = \frac{1}{2}\left(\ex{D(\rho_{XB})}\cdot\I_{\dim(N+L)} - \Pi + \I_d\right).\]
      Then
      \[g(\sigma_{XB}, D) = \ip{\sigma_{XB}}{f(D)} = \left(\ex{D(\rho_{XB})} - \ex{D(\sigma_{XB})}+1\right)/2.\]
  \end{itemize}
  Note that the strategy space $\cA$ is convex, and by the assumption, for all $\sigma_{XB}\in\cA=\Conv(\cA)$, there exists a distinguisher $D\in\cB$ such that $g(\sigma_{XB}, D) > (1+\eps')/2$.
  By Theorem~\ref{thm:min-max}, there exists a circuit $\tilde{D}$ of size $s'\cdot O\left((\log N+\log L)/\delta^2\right)$ such that for all $\sigma_{XB}$ with $\Hmin(X|B)_{\sigma}\geq k$
  \begin{align*}
    \left(\ex{\tilde{D}(\rho_{XB})} - \ex{\tilde{D}(\sigma_{XB})}+1\right)/2  & > (1+\eps')/2-\delta/2\\
    \ex{\tilde{D}(\rho_{XB})} - \ex{\tilde{D}(\sigma_{XB})} & > \delta
  \end{align*}

  which contradict the assumption $\Hmetric_{s, \eps}(X|B)_{\rho} \geq k$.
\end{proof}

\begin{remark}
  In the above discussion, we define the computational entropies and state the theorems only respect to quantum distinguishers with classical advice.
  One can also consider HILL/metric pseudoentropy respect to quantum distinguishers with quantum advice.
  The transformation between metric and HILL pseudoentropy still works in this model, since in the proof, we view distinguishers as BPOVM matrices without utilizing the fact that the advice is classical.
\end{remark}

\paragraph{Guessing pseudoentropy v.s. HILL pseudoentropy}

Vadhan and Zheng showed that in the classical case, the HILL pseudoentropy and the guessing pseudoentropy are equivalent when $n$ is logarithmic in the security parameter~\cite{VZ12}.
In fact when $n=1$, the equivalence between the HILL pseudoentropy and the guessing pseudoentropy implies Impagliazzo's Hardcore Theorem \cite{Impagliazzo95} and \textit{vice versa}.
However, in the quantum case, we do not know whether these two definitions are equivalent.
All the proofs suffer the same barrier, and we will mention it in Section~\ref{sec:barrier}.
Briefly speaking, a proof cannot be extended to the quantum case if it relies on estimating the acceptance probability of a given quantum state.
Therefore, connections between guessing pseudoentropy and other pseudoentropy notions remain as interesting open problems.

\section{Quantum Relative Min-Entropy}\label{sec:relative-min-entropy}

In this section, we consider relative min-entropy in the quantum setting.
By definition, it can be seen as a generalization of min-entropy.
That is, the relative min-entropy between a quantum state $\rho\in\density{\cC^{2^n}}$ and a $2^n$-dimension maximally mixed state is exactly $n$ minus the min-entropy of $\rho$.
As for min-entropy notions, we can also consider the computational relaxations of relative min-entropy.
Since relative min-entropy defines a ``distance'' between two states, there are more possible ways to define its relaxed notions.
Classically, some relations between different relaxations are described by the Dense Model Theorem~\cite{ReingoldTTV08}.
In Section~\ref{ssec:classical-dmt}, we review the theorem and consider a variation to show more connections among the various notions.
For the quantum case, we show in Section~\ref{ssec:quantum-dmt} that some notions are not equivalent, which can be interpreted as that a ``Quantum Dense Model Theorem'' does not hold.

\subsection{Definition}
Following the HILL-type generalization, there are already two ways to generalize the relative min-entropy (Definition~\ref{defn:relative-min-entropy-psd}) to computational notions.
First, we can say if $\rho$ has computational relative min-entropy with respect to $\sigma$, then there exists $\rho'$ that is indistinguishable from $\rho$, but has bounded relative min-entropy respective to $\sigma$ entropy. Or, we can have $\sigma'$ that is indistinguishable from $\sigma$ as the bridge.

\begin{definition}[HILL-1 relative min-entropy]\label{def:hill-rho}
  Let $\rho$ and $\sigma$ be density operators of the same system. $\Dhillone_{s, \eps}(\rho\|\sigma) \leq \lambda$ if there exists $\rho'$ that is $(s, \eps)$-indistinguishable from $\rho$ and $D_{\infty}(\rho'\|\sigma) \leq \lambda$.
\end{definition}
\begin{definition}[HILL-2 relative min-entropy]\label{def:hill-sigma}
  Let $\rho$ and $\sigma$ be density operators of the same system. $\Dhilltwo_{s, \eps}(\rho\|\sigma) \leq \lambda$ if there exists $\sigma'$ that is $(s, \eps)$-indistinguishable from $\rho$ and $D_{\infty}(\rho\|\sigma') \leq \lambda$.
\end{definition}

By switching the quantifiers, we can also have two metric-type generalizations.

\begin{definition}[metric-1 relative min-entropy]\label{def:metric-rho}
  Let $\rho$ and $\sigma$ be density operators of the same system. $\Dmetricone_{s, \eps}(\rho\|\sigma) \leq \lambda$ if for all $s$-size quantum distinguisher $A$, there exists $\rho'$ such that (i) $D_{\infty}(\rho'\|\sigma) \leq \lambda$. (ii) $\ex{A(\rho)} - \ex{A(\rho')} < \eps$.
\end{definition}
\begin{definition}[metric-2 relative min-entropy]\label{def:metric-sigma}
  Let $\rho$ and $\sigma$ be density operators of the same system. $\Dmetrictwo_{s, \eps}(\rho\|\sigma) \leq \lambda$ if for all $s$-size quantum distinguisher $A$, there exists $\sigma'$ such that (i) $D_{\infty}(\rho\|\sigma') \leq \lambda$. (ii) $\ex{A(\sigma)} - \ex{A(\sigma')} < \eps$.
\end{definition}

Another approach is to follow the ``guessing-type'' generalization.
As in min-entropy, there is an equivalent way to define relative min-entropy, using unbounded circuits.
Naturally, we can also relax the definition by restricting the size of the circuits.

\begin{definition}[Pseudo relative min-entropy]\label{def:pseudo-relative-min}
  Let $\rho$ and $\sigma$ be density operators of the same system.
  Then $\Dpseudo_{s, \eps}(\rho\|\sigma) \leq \lambda$ if for all $s$-size quantum distinguisher $A$, we have $\ex{A(\rho)} \leq 2^{\lambda}\ex{A(\sigma)} + \eps$.
\end{definition}

By the definitions, we immediately have the following relations.
\begin{proposition}
  Let $\sigma_{\rm mix}$ be the maximally mixed state in $\density{\cC^M}$
  For any $k, s\in\bbN$, $\eps, \lambda > 0$ and quantum states $\rho, \sigma\in \density{\cC^M}$, we have
  \begin{enumerate}
    \item $\Dhilltwo_{s, \eps}(\rho\| \sigma_{\rm mix}) \leq \log M - k$ if and only if $\Hhill_{s, \eps}(\rho) \geq k$.
    \item $\Dmetrictwo_{s, \eps}(\rho\| \sigma_{\rm mix}) \leq \log M - k$ if and only if $\Hmetric_{s, \eps}(\rho) \geq k$.
    \item If $\Dhillone_{s, \eps}(\rho\| \sigma) \leq \lambda$, then $\Dmetricone_{s, \eps}(\rho\| \sigma) \leq \lambda$ and $\Dpseudo_{s, \eps}(\rho\| \sigma) \leq \lambda$.
    \item If $\Dhilltwo_{s, \eps}(\rho\| \sigma) \leq \lambda$, then $\Dmetrictwo_{s, \eps}(\rho\| \sigma) \leq \lambda$ and $\Dpseudo_{s, \eps}(\rho\| \sigma) \leq \lambda$.
  \end{enumerate}
\end{proposition}

Since we can switch the quantifiers using the quantum min-max theorem (Theorem~\ref{thm:min-max}), the HILL-type and metric-type relative min-entropies are also interchangeable up to some parameter loss.

\begin{theorem}
  Let $\sigma$ and $\rho$ be quantum states in $\density{\cH}$ where $\dim(\cH) = N$.
  If $\Dmetricone_{s, \eps}(\rho\| \sigma)\leq \lda$ (resp., $\Dmetrictwo_{s, \eps}(\rho\| \sigma)\leq \lda$), then $\Dhillone_{s', \eps'}(\rho\| \sigma)\leq \lda$ (resp., $\Dhilltwo_{s', \eps'}(\rho\| \sigma)\leq \lda$, where $\eps' = 2\eps$ and $s = s'\cdot O(\log N/\eps^2)$.
\end{theorem}
\begin{proof}
  Suppose for contradiction that $\Dhillone_{s', \eps'}(\rho\| \sigma) > \lda$, then for all $\rho'$ with $D(\rho'\|\sigma) > \lda$, there exists an $s'$-size distinguisher $A$ such that $\ex{A(\rho)} - \ex{A(\rho')} > \eps'$.
  We consider the following zero-sum game:
  \begin{enumerate}
    \item The strategy space of Player~1 is $\{\rho'\in\density{\cH}:D(\rho'\|\sigma) > \lda\}$.
    \item The strategy space of Player~2 is a set of all $s'$-size distinguisher $A:\density{\cH}\to\zo$.
    \item Let the BPOVM of a distinguisher $A$ be $\Pi$. Define the auxiliary mapping $f$, which maps $A$ to $\ex{A(\rho)}\I_{\dim(\cH)} - \Pi$.
      Then the payoff for the player strategies $\rho'$ and $A$ is
      \[g(\rho', A) = \ip{\rho'}{f(A)} = \ex{A(\rho)}-\ex{A(\sigma)}.\]
  \end{enumerate}
  By the nonuniform Quantum Min-max Theorem (Theorem~\ref{thm:min-max}), there exists a universal distinguisher of size $s = s'\cdot (\log(N)/\eps^2)$ such that for all $\rho'$ with $D(\rho'\|\sigma) > \lda$,
  \[\ex{A(\rho)}-\ex{A(\sigma)} > \eps' - \eps = \eps.\]
  By the definition of metric relative entropy, we get $\Dmetricone_{s, \eps}(\rho, \sigma) > \lda$, which yields a contradiction.
  Similarly, we can have the same argument for type-2 notions.
  Note that, the strategy space of Player~1 is convex in either case, so Theorem~\ref{thm:min-max} can be applied.
\end{proof}

In the rest of the section, we will only focus on the HILL-type and guessing notions.

\subsection{Classical Dense Model Theorem}\label{ssec:classical-dmt}
In the classical case, the relation between HILL-1, HILL-2 and pseudo relative min-entropies are partly captured by Dense Model Theorem.
Specifically, the Dense Model Theorem~\cite{ReingoldTTV08} indicates that HILL-2 relative min-entropy implies HILL-1 relative min-entropy.
Also the \emph{strong Dense Model Theorem}\footnote{Strong Dense Model Theorem refers the theorem as Theorem~\ref{thm:classical-dmt}. But instead of having $X$ such that $Z$ dense in $X$ and $Y$ computationally close to $X$, we only need the ``pseudo-dense'' condition between $Z$ and $Y$.} says that pseudo relative min-entropy implies HILL-1 relative min-entropy~\cite{MironovPRV09}.
Here we additionally show that HILL-1 relative min-entropy also implies HILL-2 relative min-entropy.
Therefore, all those three notions are equivalent in the classical setting. (See Figure~\ref{fig:classical-dmt} for their relationships)

\begin{figure}
  \begin{center}
    \centerline { \xymatrix {
      *+[F]{\mbox{HILL-2}}
        \ar[rrr]^{\mbox{\small By definition}}
        \ar@(u,u)[rrrrrrrr]^{\mbox{\small Dense Model Theorem}}
      & & & *+[F]{\mbox{pseudo}}
        \ar[rrrrr]^{\mbox{\small Strong Dense Model Theorem}}
      & & & & & *+[F]{\mbox{HILL-1}}
        \ar@(d, d)[llllllll]^{\mbox{\small Lemma~\ref{lemma:reverse-dmt}}}
    } }
  \end{center}
  \caption{Relationships between computational relative min-entropies in the classical setting}
  \label{fig:classical-dmt}
\end{figure}
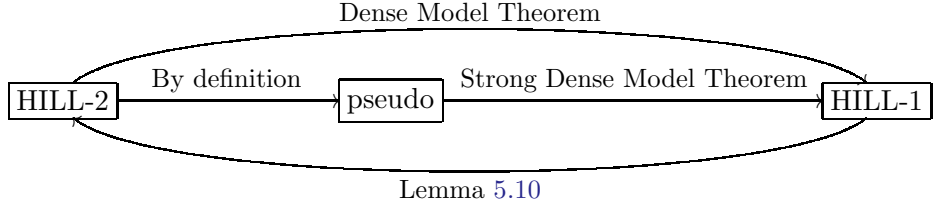

Recall the definition of ``density''.
Suppose $X, Y$ are distributions over $\bin^n$.
We say $X$ is $\delta$-dense in $Y$ if
\[\forall x\in\bin^n,~\Pr[X = x] \leq \frac{1}{\delta}\Pr[Y = x].\]
Then $X$ is $2^{-\lambda}$-dense in $Y$ if and only if $D_{\infty}(X\|Y) \leq \lambda$.

The statement of Dense Model Theorem is as follows.

\begin{theorem}[Dense Model Theorem~\cite{ReingoldTTV08}]\label{thm:classical-dmt}
  For any $s, n\in\bbN$ and $0 < \eps, \delta < 1$, let $X, Y, Z$ be three distributions over $\bin^n$ such that $X$ and $Y$ are $(s, \eps)$-indistinguishable and $Z$ is $\delta$-dense in $X$.
  Then there exists a distribution $M$ over $\bin^n$ such that $M$ is $\delta$-dense in $Y$ and $M$ is $(s', \eps')$-indistinguishable form $Z$, where $s' = \poly(s, 1/\eps, \log(1/\delta))$ and $\eps' = O(\eps/\delta)$.
\end{theorem}

Figure~\ref{fig:QDMT} is the visualization of the relationships between the distributions in Dense Model Theorem.
The theorem gives the positive answer of the existence of the distribution $M$.

\begin{figure*}[ht]
  \[\xymatrix{
      X \ar[d]_{\delta\mbox{-dense}} \ar@{~}[rr]_{\text{comp. indist.}} & & Y\ar[d]^{\delta\mbox{-dense}}\\
      Z\ar@{~}[rr]_{\text{comp. indist.}} & & *+[F]{ M ?}
    }\]
  \caption{Relation diagram of the Dense Model theorem.}
  \label{fig:QDMT}
\end{figure*}
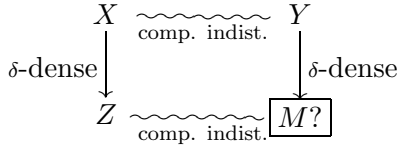

\begin{corollary}
  For any $s, n\in\bbN$, $0 < \eps < 1$ and $\lambda > 0$, let $X, Y$ be two distributions over $\bin^n$ such that $\Dhilltwo_{s, \eps}(X\|Y) \leq \lambda$, then $\Dhillone_{s', \eps'} \leq \lambda$ where $s' = \poly(s, 1/\eps, \lambda)$ and $\eps' = O(\eps \cdot 2^{\lambda})$.
\end{corollary}

Following a similar proof technique, Mironov Pandey Reingold and Vadhan also showed that pseudo relative min-entropy also imply HILL-1 relative min-entropy~\cite{MironovPRV09} up to some parameter loss (we call it  ``Strong Dense Model Theorem'').
However, to the best of our knowledge, it has not been shown that whether pseudo relative min-entropy is a weaker notion than HILL-2 relative min-entropy.
In fact, we can show that HILL-1 relative min-entropy also implies HILL-2 relative min-entropy.
We state the lemma in the density language.
\begin{lemma}\label{lemma:reverse-dmt}
  Let $Y, Z, M$ be three distributions over $\bin^n$ such that $M$ is $\delta$-dense in $Y$ and $Z$ and $M$ are $(s, \eps)$-indistinguishable.
  Then there exists a distribution $X$ over $\bin^n$ such that $Z$ is $\delta$-dense in $X$ and $X$ is $(s, \eps)$-indistinguishable from $Y$.
\end{lemma}
\begin{proof}
  Since $M$ is $\delta$-dense in $Y$, there exists a distribution $Q$ such that $Y = {\delta}M + (1-\delta)Q$.
  Let $X = \delta Z + (1-\delta)Q$, and then clearly $Z$ is $\delta$-dense in $X$.
  Moreover, $X$ and $Y$ are $(s, \eps)$-indistinguishable due to the indistinguishability between $M$ and $Z$.
\end{proof}

Therefore, all the three notions, pseudo, HILL-1 and HILL-2 relative-min entropies are all equivalent up to some parameter losses in the classical case.

\subsection{Impossibility of Quantum Dense Model Theorem}\label{ssec:quantum-dmt}
As discussed previously, HILL-1, HILL-2 and pseudo relative min-entropies are equivalent for classical distributions and the relative entropy bound is logarithmic.
In this section we will show a separation between the $\Dhillone$ and $\Dhilltwo$ relative min-entropies for quantum states.
More specifically, we show that there exist quantum states $\rho$ and $\sigma$ such that $\Dhillone(\rho\|\sigma) < \log(1/\delta)$ but $\Dhilltwo(\rho\|\sigma) = \infty$.
To this end, we use the language of Dense Model theorem.

We first define the notion of density for quantum states.
\begin{definition}[$\delta$-dense] \label{def:d-dense}
  Suppose $0 < \delta \leq 1$.
  A quantum state $\sigma \in\density{\cH}$ is \em{$\delta$-dense} in another quantum state $\rho \in\density{\cH}$ if $\sigma \leq \frac{1}{\delta}\rho$.
\end{definition}

By Definitions \ref{def:d-dense} and \ref{defn:relative-min-entropy-psd}, saying $\sigma$ is $\delta$-dense in $\rho$ is equivalent to saying that $\sigma$ has relative min-entropy at most $\log(1/\delta)$ with respect to $\rho$.
Note that if $\sigma$ and $\rho$ are classical distributions, Definition~\ref{def:d-dense} matches the classical definition.

\newcommand{\rhoX}{\rho^{X}}
\newcommand{\rhoY}{\rho^{Y}}
\newcommand{\rhoZ}{\rho^{Z}}
\newcommand{\rhoM}{\rho^{M}}
To be consistent with the classical notation, $\rhoX, \rhoY, \rhoZ$ and $\rhoM$ are different quantum states in the same space $\density{\bbC^N}$.
Recall the Dense Model Theorem statement and what the counterexample should achieve to show the non-existence of Quantum Dense Model Theorem.
Suppose $\rhoX$ and $\rhoY$ are two computationally indistinguishable quantum states and $\rhoZ$ is a quantum state that is $\delta$-dense in $\rhoX$.
A Quantum Dense Model Theorem would imply that there exists $\rhoM$ that is $\delta$-dense in $\rhoY$ and indistinguishable from $\rhoZ$.
However, we show that this is false by constructing $\rhoX$, $\rhoY$, and $\rhoZ$ such that for every $\rhoM$ that is $\delta$-dense in $\rho_Y$, it can be distinguished from $\rhoZ$.

Our counterexample is based on the following two observations: 1) the only state that is dense in a pure state is the pure state itself; 2) there exists a pseudorandom pure state (Theorem~\ref{thm:pure-indist}).

The following theorem says that a Quantum Dense Model Theorem does not exist.

\begin{theorem} \label{thm:qdmt}
  Given $s, \eps > 0$,  for $0 < \delta < 1 - 4\eps$, integers $m_1, m_2 > O(\log(s/\eps)$, and $n = m_1+m_2$, there exist quantum states $\rhoX$, $\rhoY$, $\rhoZ\in \density{\bbC^{2^n}}$ such that $\Hmin(X)_{\rhoX} = m_1, \Hmin(Y)_{\rhoY} = m_2$, $\rhoZ$ is $\delta$-dense in $\rhoX$, and $\rhoX$, $\rhoY$ are $(s, \eps)$-indistinguishable,
  but for every quantum state $\rhoM$ which is $\delta$-dense in $\rhoY$, there exists a constant-size quantum distinguisher $A$ such that
  $\size{\ex{A(\rhoM)}-\ex{A(\rhoZ)}}>\eps$.
\end{theorem}
\begin{proof}
  We first have the following two claims:

  \begin{claim}\label{claim:pure-dense}
    Suppose $\ket{\psi}\in \ball{\cH_1}$ and $\rho\in\density{\cH_2}$. Let $0 < \delta\leq 1$.
    Then a density operator in  $\density{\bbC^{M}\ot \bbC^N}$ that is $\delta$-dense in $\ketbra{\psi}\ot\rho$ must be of the form $\ketbra{\psi}\ot\sigma$, where $\sigma$ is
    $\delta$-dense in $\rho$.
  \end{claim}
  \begin{proof}[proof of Claim]
    Let $\rho = \sum_i {p_i}\ketbra{\psi_i}$ be the spectral decomposition of $\rho$.
    Then $\ketbra{\psi}\ot\rho=  \sum_i {p_i}\ketbra{\psi, \psi_i}$, where $\ket{\psi,\psi_i}=\ket{\psi}\ot\ket{\psi_i}$ for short.
    Suppose $\sigma' = \sum_{(i)}q_i\ketbra{\phi_i}\in \density{\cH_1\otimes\cH_2}$ is $\delta$-dense in $\ketbra{\psi}\ot\rho$.
    Assume that $q_j > 0$ and $\tr[\cH_2]{\ketbra{\phi_j}}\neq \ketbra{\psi}$.

    Let $\ket{v} = \frac{\ket{\phi_{j}}-\braket{\psi,\psi_i}{\phi_j}\ket{\psi,\psi_i}}{\|\ket{\phi_{j}}-\braket{\psi,\psi_i}{\phi_j}\ket{\psi,\psi_i}\|}$, which is a unit vector orthogonal to $\ket{\psi,\psi_l}$ for all $l$ but not $\ket{\phi_j}$.
    Then $\bra{v}\sigma'\ket{v}\geq q_j>0$ but $\bra{v}\left(\ketbra{\psi}\ot\rho \right)\ket{v}=0$, which contradicts the assumption that $\sigma'\leq \frac{1}{\delta} \ketbra{\psi}\ot\rho$ for $\delta> 0$.
  \end{proof}

  \begin{claim}\label{claim:ind_state}
    Let $\rho_1\in\density{\bbC^{2^{m_1}}}$ and  $\rho_2\in\density{\bbC^{2^{m_2}}}$ be two pure quantum states that are $(s, \eps)$-pseudorandom. Then $\rho_1\otimes \sigma_{\mix}^{(m_1)}$ and $\sigma_{\mix}^{(m_2)}\otimes \rho_2$ are $(s-O(\max\{m_1, m_2\}), 2\eps)$-indistinguishable.
  \end{claim}
  \begin{proof}[proof of Claim]
    Since it only takes $O(m_2)$ many ancilla qubits and $O(m_2)$ many Hadamard gates to prepare a maximally mixed state, $\rho_1\otimes \sigma_{\mix}^{(m_2)}$ and $\sigma_{\mix}^{(m_1)}\otimes \sigma_{\mix}^{(m_2)}$ are $(s-O(m_2), \eps)$-indistinguishable.
    Similarly, $\sigma_{\mix}^{(m_1)}\otimes \rho_2$ and $\sigma_{\mix}^{(m_1)}\otimes \sigma_{\mix}^{(m_2)}$ are $(s-O(m_1), \eps)$-indistinguishable from $\sigma_{\mix}^{(m_1)}\otimes \sigma_{\mix}^{(m_2)}$.
    Therefore, $\rho_1\otimes \sigma_{\mix}^{(m_2)}$ and $\sigma_{\mix}^{(m_1)}\otimes \rho_2$ are $(s-O(\max\{m_1, m_2\}), 2\eps)$-indistinguishable from each other.
  \end{proof}

  By Claim~\ref{claim:ind_state}, there exist $\rho_X= \sigma_{\mix}^{(m_1)}\otimes \rho_2$ and $\rho_Y=\rho_1\otimes \sigma_{\mix}^{(m_2)}$ that are $(s, \eps)$-indistinguishable where $\rho_1$ and $\rho_2$ are pure pseudorandom states.
  Thus the entropies of $\rhoX$ and $\rhoY$ are $m_1$ and $m_2$, respectively.

  Denote $\tau_0 = \frac{1}{2^{m_1-1}}\ketbra{0}\ot\I_{2^{m_1-1}}$ and $\tau_1 = \frac{1}{2^{m_1-1}}\ketbra{1}\ot\I_{2^{m_1-1}}$.
  Let
  \[\rhoZ= \left(\min\left\{1, \frac{1}{2\delta}\right\}\tau_0 + \max\left\{0, 1-\frac{1}{2\delta}\right\}\tau_1\right)\otimes \rho_2.\]
  Then $\rhoZ$ is $\delta$-dense in $\rhoX$.
  By Claim~\ref{claim:pure-dense}, for every $\rhoM$ that is $\delta$-dense in $\rhoY$, $\rhoM$ must be of the form $\rho_1\otimes \sigma_2$ for $\sigma_2\in\density{\bbC^{2^{m_2}}}$.

  Now we define a quantum distinguisher $A$ with BPOVM $\Pi = \ketbra{0}\ot\I_{2^{m_1-1}} \ot \I_{2^{m_2}}$.
  Thus
  \[\ex{A(\rhoZ)}= \ip{\Pi}{\rhoZ}=\min\left\{1, 1/2\delta \right\} > \frac{1}{2}+2\eps.\]
  On the other hand, since $\rho_1$ is $(s, \eps)$ pseudorandom,
  \[\ex{A(\rhoM)}= \ip{\Pi}{\rhoM} = \ip{\ketbra{0}\ot\I_{2^{m_1-1}}}{\rho_1} \leq \frac{1}{2} + \eps.\]
  Therefore, the quantum distinguisher $A$ is as desired.
\end{proof}

\begin{corollary}
  Given $s\in\cN$, $\eps > 0$, for $0 < \delta < 1\eps$ and $n > O(\log(s/\eps))$, there exist  quantum states $\rho, \sigma \in\density{\bbC^{2^n}}$ such that $\Dhilltwo_{s, \eps}(\rho\|\sigma) < \log(1/\delta)$ but $\Dhillone_{O(1), O(1)}(\rho\|\sigma) = \infty$.
\end{corollary}

\section{Computational Leakage Chain Rule}\label{sec:leakage-chain-rule}

In this section, we will prove the leakage chain rule for quantum relaxed-HILL pseudoentropy with quantum leakage.
First, we recall the statement of the classical leakage chain rule~\cite{DziembowskiP08}.
Suppose there is a joint distribution $(X, Z, B)$ (which will be a quantum state $\rho_{XZB}$ in the quantum setting) where $X$, $Z$ and $B$ are viewed as a source, prior knowledge, and leakage, respectively.
The leakage chain rule says that, if the pseudoentropy of $X$ conditioned on $Z$ is at least $k$, then the pseudoentropy of $X$ conditioned on both $Z$ and $B$ retains at least $k-\ell$, where $\ell$ is the length of $B$.
In the asymptotic setting, we will focus on the case that the leakage $B$ is small ($O(\log \secp)$ qubits) and the length of $X$ and $Z$ could be $\poly(\secp)$ bits, where $\secp$ is the security parameter.
The leakage chain rule cannot hold when $B$ is too long.
For example, assume the existence of one-way function. Let $X = G(B)$ where $B$ is a uniform distribution and $G$ is a pseudorandom generator, and $Z$ is empty.
Then $X$ has a full HILL entropy, while $X|B$ has zero.

If there is no prior knowledge $Z$ and $|B| = O(\log n)$, then HILL pseudoentropy and relaxed-HILL pseudoentropy are equivalent.
If prior information $Z$ is allowed,  Krenn \emph{et~al.}~\cite{KrennPW13} showed that the leakage lemma is unlikely to hold for standard HILL pseudoentropy.
Specifically, assuming the existence of a perfect binding commitment scheme, they constructed random variables $X, Z,$ and $B$, where $B$ is a single random bit, such that $\Hhill(X|Z)\geq n$, but $\Hhill(X|Z, B) \leq 1$.
On the other hand, we know that the leakage chain rule holds for classical relaxed-HILL pseudoentropy~\cite{DziembowskiP08, ReingoldTTV08, GentryW11}.
Therefore, as a first-step study, we aim for proving a quantum leakage chain rule for relaxed-HILL pseudoentropy.

One main obstacle of proving the leakage chain rule for quantum entropies is that, contrary to the classical case, we cannot consider a probability conditioned on a fixed leakage $b$ (\ie~$\Pr[X|B = b]$) which blocks the possibility of handling cases with different $b$ separately.
In fact, various proofs of the chain rule for entropies (\eg~\cite{ChungKLR11,JP14}) rely on this property, including both HILL and relaxed-HILL entropies.
Eventually, we adopt the proof via the \emph{Leakage Simulation Lemma}.

\paragraph{Leakage Simulation Lemma.}

The leakage simulation lemma says that given a joint distribution $(X, B)$, there is a simulator circuit $C$ of small complexity such that $(X, B)$ and $(X, C(X))$ are indistinguishable.
As we will show in Section~\ref{ssec:leakage-chain-rule}, this lemma immediately implies Gentry and Wichs' simulation lemma~\cite{GentryW11} and hence the chain rule.
Those lemmas also have other interesting applications (\eg~\cite{GentryW11, ChungLP15, TrevisanTV09}).

\begin{reptheorem}{thm:leak-sim}
  Let $\rho_{XB} = \sum_{x\in\bin^n}p_x\ketbra{x}\ot\rho^x_B\in\zo^n\times\density{\bbC^{2^\ell}}$ be a cq-state with $n$ classical bits and $\ell$ qubits.
  For any $s\in\bbN$ and $\eps > 0$, there exists a quantum circuit $C:\bin^n\to\density{\bbC^{2^\ell}}$ of size $s' = \poly(s, n, 2^\ell, 1/\eps)$ such that the cq-state $\sum_{x\in\bin^n}p_x \ketbra{x}\otimes C(x)$ and $\rho_{XB}$ are $(s, \eps)$-indistinguishable.
\end{reptheorem}

There are two ways to prove the quantum simulation leakage lemma: one is based on the boosting (multiplicative weight update) technique, and the other is via the (nonuniform) Min-max Theorem.
We illustrate the former in the main body (Section~\ref{subsec:boosting-proof}) and leave the latter in Appendix~\ref{app:min-max-proof}.
In both proofs, the techniques that we need to convert a given circuit into its corresponding BPOVM, called quantum tomography, are detailed in Appendix~\ref{app:tomography}.

\subsection{Leakage Chain Rule}\label{ssec:leakage-chain-rule}

\begin{theorem}\label{thm:chain_rlt_hill}
  For any $n, m, \ell, s'\in\bbN$ and $\eps > 0$, the following holds for $s = \poly(s', n, 2^\ell, 1/\eps)$ and $\eps' = O(\eps)$.
  Let $\rho_{XZB} = \sum_{(x,z)\in\bin^{n+m}}p_{xz} \ketbra{x}\ot \ketbra{z}\otimes\rho^{xz}_B$ be a ccq-state with $n+m$ classical bits and ${\ell}$ qubits.
  If $\Hhillr_{s, \eps}(X|Z) \geq k$,
  then we have $\Hhillr_{s', \eps'}(X|Z, B)\geq k-\ell$.
\end{theorem}
\begin{proof}
  We use the following lemma as an intermediate step to derive the Leakage Chain Rule.

  \begin{lemma}[Generalization of \cite{GentryW11} Lemma~3.2]\label{lemma:gw-sim}
    For any $n, \ell, s'\in\bbN$ and $\eps > 0$, the following holds for $s = \poly(s', n, 2^\ell, 1/\eps)$ and $\eps' = O(\eps)$.
    Let $\rho_{XB} = \sum_{x\in\bin^n}p_x \ketbra{x}\otimes\rho^x_B$ be a cq-state with $n$ classical bits and ${\ell}$ qubits.
    For every $Y$ that is $(s, \eps)$-indistinguishable from $X$, there exists a (possibly inefficient) quantum circuit $C$ such that the cq-state
    \[\rho_{YC} = \sum_{y\in\bin^n}q_y \ketbra{y}\otimes C(y)\]
    and $\rho_{XB}$ are $(s', \eps')$-quantum-indistinguishable where $q_y = \pr{Y = y}$.
  \end{lemma}

  \begin{proof}[proof of Lemma~\ref{lemma:gw-sim}]
    By Theorem~\ref{thm:leak-sim}, there exists a circuit $C:\bin^n\to\density{\bbC^{2^{\ell}}}$ with size $t = s' + \poly(s', n, 2^{\ell}, 1/\eps)$ such that $\rho_{XB}$ and $\rho_{XC(X)}$ are $(s', \eps)$-indistinguishable.
    Set $s = 2t$.
    Since $X$ and $Y$ are $(s, \eps)$-indistinguishable, $\sum_{x\in\zo^n}p_x\ketbra{x}\ot C(x)$ and $\sum_{y\in\zo^n}q_y\ketbra{y}\ot C(y)$ are $(s-t = t, \eps)$-indistinguishable.
    By the transitivity of indistinguishability, $\rho_{XB}$ and $\rho_{YC}$ are $(\min\{s', t\}=s', 2\eps)$-indistinguishable.

  \end{proof}
  Once we have Lemma~\ref{lemma:gw-sim}, we can derive the chain rule for quantum relaxed HILL entropy from the chain rule of quantum min-entropy.
  Since $\rho_{XZB}$ s a ccq-state, which is separable in on the space $(\cX\ot\cZ)\ot\cB$, losing only $\ell$ bits HILL entropy instead of $2\ell$ as in Theorem~\ref{thm:quantum-chain-rule-sep} is possible.

  $\Hhillr(X|Z) \geq k$ implies there exists a joint distribution $(Y, Z')$ such that $(Y, Z')$ and $(X, Z)$ are $(s, \eps)$-indistinguishable and $\Hmin(Y|Z')\geq k$.
  By Lemma~\ref{lemma:gw-sim}, there exists an $\ell$-qubit quantum system $C$ such that $(X, Z, B)$ and $(Y, Z', C)$ are $(s', \eps')$-indistinguishable where $s = \poly(s', n, 2^{\ell}, 1/\eps)$ and $\eps' = O(\eps)$.
  By the chain rule of quantum min-entropy (Theorem~\ref{thm:quantum-chain-rule-sep}), $\Hmin(Y|Z', C)\geq k-\ell$, which implies $\Hhillr(X|Z, B)_{s', \eps'}\geq k-\ell$.
\end{proof}

We have proved the quantum leakage chain rule for ccq-states.
However, due to some barriers that we will mention in Section~\ref{sec:barrier} and other obstacles, it is still open for the cqq-state case (the prior knowledge $Z$ is quantum).
The quantum leakage chain rule for cqq-states is desired because it might help one to improve the leakage-resilient protocol~\cite{DziembowskiP08} to be secure against a quantum leakage.

\begin{openprob}
  Let $\rho_{XBZ} = \sum_{(x)\in\bin^{n}}p_{x} \ketbra{x} \ot\rho^{x}_{BZ}$ be a cqq-state with $n$ classical bits and ${\ell+m}$ qubits.
  If $\Hhillr_{s, \eps}(X|Z) \geq k$, can we show that $\Hhillr_{s', \eps'}(X|Z, B)\geq k-\ell$ for some $s' = \poly(s, n, 2^\ell, 1/\eps))$ and $\eps' = O(\eps)$?
\end{openprob}

On the other hand, it is not known whether quantum HILL and relaxed-HILL entropies are equivalent, even for the case that $B$ is a single qubit.
Thus our result does not imply a chain rule for quantum HILL entropies without prior knowledge.
\begin{openprob}
  Let $\rho_{XZ}$ be a cq-state with ${n+1}$ qubits.
  If $\Hhillr_{s, \eps}(X|Z) \geq k$, can we show that $\Hhill_{s', \eps'}(X|Z)\geq k$ for some $s' = \poly(s, n, 1/\eps))$ and $\eps' = O(\eps)$?
\end{openprob}

\subsection{Leakage Simulation Lemma by MMWU}\label{subsec:boosting-proof}
This section dedicates to a proof of the Leakage Simulation Lemma.
In this proof, we generalize the idea in \cite{VZ13} to the quantum setting.
The overview of the algorithm is as follows.
We will construct a simulator via MMWU method.
Initially, the simulator outputs a maximally mixed state.
If there exists a distinguisher can distinguish the output from $\rho_{XB}$, we use the best (or almost) distinguisher to update the simulator.
Guaranteed by the MMWU method, we will get the desired simulator within polynomially many rounds.

However, due to the fact that quantum circuits are random, we cannot only consider deterministic distinguishers that output $\{0, 1\}$.
Therefore, we adopt other techniques and have the simulator circuit much more complicated than in the classical setting.
Roughly speaking, even though we cannot hope for the quantum circuits being deterministic, we can make it output the same result with high probability if we repeat more times and randomly shift the estimations.

In the following proof, we will picture the idea of building the simulator via Procedure~\ref{pro:mmwu-sim} and its approximation version, Procedure~\ref{pro:mmwu-sim-approx}.
Those procedures are merely the abstracts for the simulator and do not handle the circuit complexity issue.
Then we will show how to construct the simulator circuit in Procedure~\ref{pro:mmwu-sim-ckt}.

\begin{proof}
  Let $d = 2^{\ell}$ be the dimension of the quantum space.
  Recall that we can use a $d$ by $d$ BPOVM matrix to describe a quantum distinguisher on $\density{\bbC^d}$.
  Follow the same idea, we can also use a set of $d$ by $d$ BPOVM matrices $\{\Pi_x\}_{x\in\zo^n}$ where $\Pi_x$ is the BPOVM matrix of the quantum circuit $A(x, \cdot)$ to characterize a quantum distinguisher of domain $\zo^n\ot\density{\bbC^d}$.
  As a result, for an cq-state $\rho_{XB} = \sum_{x\in\bin^n}p_x\ketbra{x}\ot\rho^x_B\in\zo^n\ot\density{\bbC^{d}}$, we have
  \[\ex{A(\rho_{XB})} = \sum_{x\in\zo^n}p_x \ex{A(\ketbra{x}\ot\rho^x_B)} = \sum_{x\in\zo^n}p_x\ip{\Pi_x}{\rho^x_B} = \ex[x\sim X]{\ip{\Pi_x}{\rho^x_B}}.\]

  First, we consider the following MMWU procedure for a given $x$.

  \begin{center}
  \begin{pseudocode}[shadowbox]{Procedure}{} \label{pro:mmwu-sim}
    \textbf{Input:} $x\in\zo^n$ and an error parameter $\eps > 0$\\
    \textbf{Output:} A quantum state $\sigma\in\density{\bbC^d}$
    \begin{enumerate*}
      \item Choosing $T = O(\log d/\eps^2)$ and $\eta = \sqrt{\ln d/T}$.
      \item Initially, let $W^{(1)}_x = \I_d$.
      \item For $t = 1, \dots, T$,
      \begin{enumerate*}
        \item Let $\sigma^{(t)}_x = W^{(t)}_x/\tr{W^{(t)}_x}$.
        \item Let $A^{(t)}$ be the best $s$-size distinguisher in distinguishing $\rho_{XB}$ and $\sum_{x\in\zo^n}p_x\ketbra{x}\ot\sigma^{(t)}_x$.
          Namely,
          \[A^{(t)} = \argmax_{A \mbox{ of size } s} \ex[x\sim X]{A(\ketbra{x}\ot\rho^x_B)} - \ex[x\sim X]{A(\ketbra{x}\ot\sigma^{(t)}_x)}.\]
        \item Denote the corresponding BPOVM matrices of $A^{(t)}$ to be $\{\Pi^{(t)}_x\}_{x\in\zo^n}$, then define the loss matrix
          \[L^{(t)}_x = \ex[x\sim X]{\ip{\Pi^{(t)}_x}{\rho^x_B}}\I_d - \Pi^{(t)}_x.\]
        \item Let $W^{(t+1)}_x = W^{(t)}_x\cdot\exp\left(-\eta L^{(t)}_x\right)$.
      \end{enumerate*}
      \item Output $\sigma_x = \frac{1}{T}\sum_{t = 1}^T\sigma_x^{(t)}$.
    \end{enumerate*}
  \end{pseudocode}
  \end{center}

  \begin{claim}\label{claim:mmwu-sim}
    Let $\{\sigma_x\}_{x\in\zo^n}$ be the states obtained from Procedure~\ref{pro:mmwu-sim}.
    Then for every quantum distinguisher $A$ of size $s$, we have
    \begin{align*}
      \ex{A\left(\rho_{XB}\right)} - \ex{A\left(\sum_{x\in\zo^n}p_x\ketbra{x}\ot\sigma_x\right)} \leq \eps
    \end{align*}
  \end{claim}
  \begin{proof}[proof of Claim~\ref{claim:mmwu-sim}]
    By Theorem~\ref{thm:mmwu}, for any $\tau\in\density{\bbC^d}$
    \[\frac{1}{T}\sum_{t = 1}^T\ip{\sigma^{(t)}_x}{L^{(t)}_x} \leq \frac{1}{T}\sum_{t = 1}^T\ip{\tau}{L^{(t)}_x} + \eps\]
    when $T$ is properly chosen.
    Particularly, if we take $\tau = \rho^x_B$, we get
    \begin{align}\label{eq:boosting-mmwu}
      \frac{1}{T}\sum_{t = 1}^T\ip{\sigma^{(t)}_x}{L^{(t)}_x} \leq \ex[x\sim X]{\ip{\Pi^{(t)}_x}{\rho^x_B}} - \ip{\Pi^{(t)}_x}{\rho^x_B} + \eps.
    \end{align}
    By the definitions of $A^{(t)}$ and $L^{(t)}_x$, for any algorithm distinguisher $A$,
    \begin{align}\begin{split}\label{eq:boosting-def-a-l}
      \ex[x\sim X]{A(\ketbra{x}\ot\rho^x_B)} - A(\ketbra{x}\ot\sigma^{(t)}_x)
      & \leq \ex[x\sim X]{A^{(t)}(\ketbra{x}\ot\rho^x_B)} - A^{(t)}(\ketbra{x}\ot\sigma^{(t)}_x)\\
      & \leq \ip{\sigma^{(t)}_x}{L^{(t)}_x}.
    \end{split}\end{align}
    Let $\sigma_x = \frac{1}{T}\sum_{t=1}^T\sigma^{(t)}_x$.
    Combine Equation~(\ref{eq:boosting-mmwu}) and (\ref{eq:boosting-def-a-l}), and take the expectation over $x$ from $X$, we get
    \begin{align*}
      & \ex{A\left(\rho_{XB}\right)} - \ex{A\left(\sum_{x\in\zo^n}p_x\ketbra{x}\ot\sigma_x\right)}\\
      =& \sum_{x\in\zo^n}p_x\cdot \frac{1}{T}\sum_{t = 1}^{T}\left(\ex[x\sim X]{A(\ketbra{x}\ot\rho^x_B)} - A(\ketbra{x}\ot\sigma^{(t)}_x)\right) \leq \eps.
    \end{align*}
  \end{proof}

  In Protocol~\ref{pro:mmwu-sim}, one cannot easily obtain the precise BPOVM matrices $\Pi^{t}_x$ from $A^{(t)}$.
  First, it needs a quantum tomography technique to get BPOVM form of a distinguisher given the distinguisher (as a circuit or an oracle).
  Usually, the technique evolves sampling, so one cannot get a deterministic result.
  Second, due to the precision issue, one cannot hope to store or have a circuit outputs the precise BPOVM matrices.
  In the following procedure, we address the second concern.
  We will show that, approximations of the BPOVM matrices (even themselves are not valid BPOVM matrices) suffice for building a simulator.

  \begin{center}
  \begin{pseudocode}[shadowbox]{Procedure}{} \label{pro:mmwu-sim-approx}
    \textbf{Input:} $x\in\zo^n$ and an error parameter $\eps > 0$\\
    \textbf{Output:} A quantum state $\sigma\in\density{\bbC^d}$
    \begin{enumerate*}
      \item Choosing $T = O(\log d/\eps^2)$ and $\eta = \sqrt{\ln d/T}$.
      \item Initially, let $W^{(1)}_x = \I_d$.
      \item For $t = 1, \dots, T$,
      \begin{enumerate*}
        \item Let $\sigma^{(t)}_x = W^{(t)}_x/\tr{W^{(t)}_x}$.
        \item Let $A^{(t)}$ be the best $s$-size distinguisher in distinguishing $\rho_{XB}$ and $\sum_{x\in\zo^n}p_x\ketbra{x}\ot\sigma^{(t)}(x)$.
          Namely,
          \[A^{(t)} = \argmax_{A \mbox{ of size } s} \ex[x\sim X]{A(\ketbra{x}\ot\rho^x_B)} - \ex[x\sim X]{A(\ketbra{x}\ot\sigma^{(t)}_x)}.\]
        \item Denote the corresponding BPOVM matrices of $A^{(t)}$ to be $\left\{\Pi^{(t)}_x\right\}_{x\in\zo^n}$.
          Let $\tilde{\Pi}^{(t)}_x$ be a matrix such that $\opnorm{\Pi^{(t)}_x-\tilde{\Pi}^{(t)}_x} \leq \eps/4$, then define the loss matrix
          \[\tilde{L}^{(t)}_x = \ex[x\sim X]{\ip{\tilde{\Pi}^{(t)}_x}{\rho^x_B}}\I_d - \tilde{\Pi}^{(t)}_x.\]
        \item Let $W^{(t+1)}_x = W^{(t)}_x\cdot\exp\left(-\eta \tilde{L}^{(t)}_x\right)$.
      \end{enumerate*}
      \item Output $\sigma_x = \frac{1}{T}\sum_{t = 1}^T\sigma_x^{(t)}$.
    \end{enumerate*}
  \end{pseudocode}
  \end{center}
  \begin{claim}\label{claim:mmwu-sim-approx}
    Let $\{\sigma_x\}_{x\in\zo^n}$ be the states obtained from Procedure~\ref{pro:mmwu-sim-approx}.
    Then for every quantum distinguisher $A$ of size $s$, we have
    \begin{align*}
      \ex{A\left(\rho_{XB}\right)} - \ex{A\left(\sum_{x\in\zo^n}p_x\ketbra{x}\ot\sigma_x\right)} \leq \eps.
    \end{align*}
  \end{claim}
  \begin{proof}[proof of Claim~\ref{claim:mmwu-sim-approx}]
    Since $\opnorm{\Pi^{(t)}_x-\tilde{\Pi}^{(t)}_x} \leq \eps/4$, we have $-\frac{\eps}{4}\I_d\leq \tilde{\Pi}^{(t)}_x \leq \left(1+\frac{\eps}{4}\right)\I_d$ and so $-\frac{\eps}{4}\I_d\leq \tilde{L}^{(t)}_x \leq \left(2+\frac{\eps}{4}\right)\I_d$ for all $x\in\zo^n, t\in[T]$.
    By Theorem~\ref{thm:mmwu}, for any $\tau\in\density{\bbC^d}$
    \[\frac{1}{T}\sum_{t = 1}^T\ip{\sigma^{(t)}_x}{\tilde{L}^{(t)}_x} \leq \frac{1}{T}\sum_{t = 1}^T\ip{\tau}{\tilde{L}^{(t)}_x} + \frac{\eps}{2}\]
    when $T$ is properly chosen.
    Particularly, if we take $\tau = \rho^x_B$, we get
    \begin{align}\label{eq:boosting-mmwu-approx}
      \frac{1}{T}\sum_{t = 1}^T\ip{\sigma^{(t)}_x}{\tilde{L}^{(t)}_x} \leq \ex[x\sim X]{\ip{\Pi^{(t)}_x}{\rho^x_B}} - \ip{\Pi^{(t)}_x}{\rho^x_B} + \frac{\eps}{2}.
    \end{align}
    By the definitions of $A^{(t)}$ and $\tilde{L}^{(t)}_x$, for any algorithm distinguisher $A$,
    \begin{align}\begin{split}\label{eq:boosting-def-a-l-approx}
      \ex[x\sim X]{A(\ketbra{x}\ot\rho^x_B)} - A(\ketbra{x}\ot\sigma^{(t)}_x)
      & \leq \ex[x\sim X]{A^{(t)}(\ketbra{x}\ot\rho^x_B)} - A^{(t)}(\ketbra{x}\ot\sigma^{(t)}_x)\\
      & = \ex[x\sim X]{\ip{\Pi^{(t)}_x}{\rho^x_B}} - \ip{\Pi^{(t)}_x}{\sigma^{(t)}_x}\\
      & \leq \ex[x\sim X]{\ip{\tilde{\Pi}^{(t)}_x}{\rho^x_B}} - \ip{\tilde{\Pi}^{(t)}_x}{\sigma^{(t)}_x} + 2\cdot\frac{\eps}{4}\\
      & = \ip{\sigma^{(t)}_x}{L^{(t)}_x} + \frac{\eps}{2}.
    \end{split}\end{align}
    Let $\sigma_x = \frac{1}{T}\sum_{t=1}^T\sigma^{(t)}_x$.
    Combining Equation~(\ref{eq:boosting-mmwu-approx}) and (\ref{eq:boosting-def-a-l-approx}), and taking the expectation over $x$ from $X$, we get
    \begin{align*}
      & \ex{A\left(\rho_{XB}\right)} - \ex{A\left(\sum_{x\in\zo^n}p_x\ketbra{x}\ot\sigma_x\right)}\\
      =& \sum_{x\in\zo^n}p_x\cdot \frac{1}{T}\sum_{t = 1}^{T}\left(\ex[x\sim X]{A(\ketbra{x}\ot\rho^x_B)} - A(\ketbra{x}\ot\sigma^{(t)}_x)\right) \\
      \leq& \frac{\eps}{2} + \frac{\eps}{2} = \eps.
    \end{align*}
  \end{proof}

  Based on Procedure~\ref{pro:mmwu-sim-approx}, we are going to build a small circuit $C:\zo^n\to\density{\bbC^d}$ such that $C(x)$ will output approximates $\sigma_x$ with high probability.
  Intuitively, to calculate $\tilde{\Pi}^{(t)}_x$, one can simply use a tomography algorithm (see Appendix~\ref{app:tomography}) to obtain an approximation matrix $\tilde{\Pi}_x$.
  However, as mentioned before, the circuit is not deterministic, that means $A^{(t)}$ depends on $\tilde{\Pi}^{(t-1)}_x$ and we cannot afford to hardwire all possibilities of $A^{(t)}$.
  To deal with that, after finishing an tomography algorithm, we shift each entry (both real and imaginary parts) a small amount, and round up to certain precision.
  If there exists a shift such that with high probability, each number uniquely round to a number, then the tomography algorithm will output the same result with high probability.
  Concretely, we consider the following circuit.

  \begin{center}
  \begin{pseudocode}[shadowbox]{Procedure}{} \label{pro:mmwu-sim-ckt}
    \textbf{Input:} $x\in\zo^n$ and an error parameter $\eps > 0$\\
    \textbf{Output:} A quantum state $\sigma\in\density{\bbC^d}$
    \begin{enumerate*}
      \item Choosing $T = O(\log d/\eps^2)$ and let $\eta = \sqrt{\ln d/T}$, $\Delta = 24d^2T/\eps$.
      \item Let $\delta\in[\Delta]$ be a nonuniform advice for the circuit.
      \item Initially, let $C^{(1)}(x) = \I_d$.
      \item For $t = 1, \dots, T$,
      \begin{enumerate*}
        \item Let $\sigma^{(t)}_x = C^{(t)}(x)/\tr{C^{(t)}(x)}$.
        \item Let $A^{(t)}$ be the best $s$-size distinguisher in distinguishing $\rho_{XB}$ and $\sum_{x\in\zo^n}p_x\ketbra{x}\ot\sigma^{(t)}(x)$.
          Namely,
          \[A^{(t)} = \argmax_{A \mbox{ of size } s} \ex[x\sim X]{A(\ketbra{x}\ot\rho^x_B)} - \ex[x\sim X]{A(\ketbra{x}\ot\sigma^{(t)}_x)}.\]
        \item Let the circuit $C^{A^{(t)}(x, \cdot)}_{\rm tom}$ be the circuit do the following.
        \begin{enumerate}
          \item Solve the quantum tomography problem (Definition~\ref{def:qckt-tomography}) $\QCktTom(s, d, \eps/16d, \eps/4T)$ for the distinguisher $A^{(t)}(x, \cdot)$ and get the approximation matrix $M^{(t)}_x$ (The problem and the algorithms are defined in Appendix~\ref{app:tomography}.
          \item Add $\frac{\delta}{\Delta}\cdot\frac{\eps}{16d}(1+i)$ to the entries in upper triangle and $\frac{\delta}{\Delta}\cdot\frac{\eps}{16d}(1-i)$ to ones in lower triangle in $M$ and round each number to closest multiple of $\eps/16d$.
          \item Output the result matrix as $\tilde{\Pi}^{(t)}_x$.
        \end{enumerate}

        \item Calculate the matrix
          \[\tilde{L}^{(t)}_x = \ex[x\sim X]{\ip{\tilde{\Pi}^{(t)}_x}{\rho^x_B}}\I_d - \tilde{\Pi}^{(t)}_x.\]
        \item Let $C^{(t+1)}(x) = C^{(t)}(x)\cdot\exp\left(-\eta \tilde{L}^{(t)}_x\right)$.
      \end{enumerate*}
      \item Based on the descriptions of quantum states $C^{(t)}(x)$, output an $\eps/4$-approximation of $\sigma_x = \frac{1}{T}\sum_{t = 1}^T\sigma_x^{(t)}$.
    \end{enumerate*}
  \end{pseudocode}
  \end{center}

  \begin{claim}\label{claim:crazy}
    For a fixed $x$, if we randomly sample $\delta$ from $1, \dots, \Delta$, then with probability at least $1-\eps/2$ over the choice of $\delta$ and the executions of the circuits $C_{\rm tom}^{A^{(t)}(x, \cdot)}$ for $t = 1, \dots, T$, for each $t$, the output matrices $\Pi_{x}^{(t)}$ is fixed and $\opnorm{\tilde{\Pi}^{(t)}_x - \Pi^{(t)}_x} \leq \eps/4$ where $\Pi_x^{(t)}$ is the BPOVM matrix of the distinguisher $A^{(t)}(x, \cdot)$.
  \end{claim}
  \begin{proof}[proof of Claim~\ref{claim:crazy}]
    Assume that $\norm{M^{(t)}_x - \Pi^{(t)}_x}_{\max} \leq \eps/16d$.
    After shifting every number of $M^{(t)}_x$ with the amount at most $\eps/16d$ and rounding every number to a closest multiple of $\eps/16d$ to get $\tilde{\Pi}^{(t)}_x$, we have $\norm{M^{(t)}_x-\tilde{\Pi}^{(t)}_x}_{\max} \leq 2\sqrt{2}\eps/16d$, so $\norm{\tilde{\Pi}^{(t)}_x - \Pi^{(t)}_x}_{\max} \leq \eps/4d$.
    Then by Equation~(\ref{eqn:max_norm}), $\opnorm{\tilde{\Pi}^{(t)}_x - \Pi^{(t)}_x} \leq \eps/4$.

    Now we show that, conditioned on $M^{(t)}_x$ and $\Pi^{(t)}_x$ are entry-wise close, with high probability over the the shifting, the output of $C_{\rm tom}^{A^{(t)}(x, \cdot)}$ is fixed.
    Observe that, for a given number, if we shift the amount $\frac{\delta}{\Delta}\frac{\eps}{16d}$ where $\delta$ is uniformly and randomly chosen from $1, \dots, \Delta$.
    Then the probability that it locates less than $\frac{1}{\Delta}\frac{\eps}{16d}$ distance to some multiple of $\frac{\eps}{16d}$ is at most
    $3/\Delta \leq \eps/8d^2T$.

    In the procedure, we do the shifting and rounding for $2\cdot d^2 T$ real numbers.
    By the union bound, with probability at least $1 - \frac{\eps}{8d^2T}\cdot 2d^2T = 1 - \eps/4$ over choosing the shifting, every number will end up to a fixed value (since they all are far from multiples of $\frac{\eps}{16d}$).

    Recall that the tomography algorithm also only guarantee that with probability $1 - \eps/4T$, we have $\norm{M^{(t)}_x - \Pi^{(t)}_x}_{\max} \leq \eps/16d$.
    Again, union bound this two events, we have the claim.
  \end{proof}
  From the above claim, we know that a random shifting works for a fixed $x$ with probability at least $1-\eps/2$.
  Simply by an averaging argument, we know there exists a shift $\delta$ such that it works for $x$ with probability at least $1-\eps/2$ when $x$ is chosen from the distribution $X$.
  Let $W$ be the set of all $x$ that the shift $\delta$ works.
  Now we can apply Claim~\ref{claim:mmwu-sim-approx} with parameter $\eps/4$,
  \begin{align*}
    & \ex{A\left(\rho_{XB}\right)} - \ex{A\left(\sum_{x\in\zo^n}p_x\ketbra{x}\ot \sigma_x\right)}\\
    \leq& \pr{A\left(\rho_{XB}\right) = 1 \wedge x\in W}  - \pr{A\left(\sum_{x\in\zo^n}p_x\ketbra{x}\ot \sigma_x\right)=1 \wedge x\in W} + \pr{x\notin W}\\
    \leq& \eps/2 + \eps/4 = 3\eps/4.\\
  \end{align*}
  Eventually, the trace distance between the quantum states outputted by the circuit $C(x)$ and $\sigma_x$ is at most $\eps/4$.
  Then,
  \begin{align*}
    & \ex{A\left(\rho_{XB}\right)} - \ex{A\left(\sum_{x\in\zo^n}p_x\ketbra{x}\ot C(x)\right)} \\
    \leq & \ex{A\left(\rho_{XB}\right)} - \ex{A\left(\sum_{x\in\zo^n}p_x\ketbra{x}\ot \sigma_x)\right)} + \frac{\eps}{4}\\
    \leq & 3\eps/4 + \eps/4 = \eps.\\
  \end{align*}

  By Lemma~\ref{lemma:qkct-tomography}, the size of the circuit $C^{A^{(t)}(x, \cdot)}_{\rm tom}$ is at most $\poly(s, d, \eps/16d, \log(\eps/4T)) = \poly(s, d, 1/\eps)$.
  The circuit $C^{A^{(t)}(x, \cdot)}_{\rm tom}$ is executed $T = O(\log d / \eps^2)$ times.
  In the final step, the quantum state can be constructed from its description by a circuit of size polynomial in its description length and dimension~\cite{SBM05}.
  Summarily, the complexity of the circuit $C$ is $\poly(s, d, 1/\eps) = \poly(s, 2^\ell, 1/\eps)$.
\end{proof}

\begin{remark}
  We can also consider the leakage chain rule in the model that quantum distinguishers are given quantum advice.
  The same proof is still applicable as long as the simulator is allowed to have a quantum advice, since the only place that we need a quantum advice is that we use distinguishers $A^{(t)}$ as an advice when constructing the simulator.
\end{remark}

\section{Application to Quantum Leakage-Resilient Cryptography}\label{sec:app-resilient}

Classically, important applications of the Leakage Simulation Lemma and the Leakage Chain Rule are to Leakage-Resilient Cryptography, which aims to construct secure cryptographic protocols even if side information about the honest parties' secrets leak to an adversary.
For instance, the security of leakage-resilient stream cipher based on any weak pseudorandom function (weak PRF) was proved using the classical Leakage Simulation Lemma~\cite{Pie09,JP14}, and the security of the construction based on a pseudorandom generator (PRG) was proved by the classical Leakage Chain Rule~\cite{DziembowskiP08}.

Here, we apply our Quantum Leakage Simulation Lemma to obtain a stream-cipher that is secure against quantum adversaries that can get quantum leakage as well as classical leakage, provided that the adversary has bounded quantum storage. (The classical storage of the adversary is unbounded.)
The construction is the same as in~\cite{DziembowskiP08} but instantiated with a PRG secure against quantum adversaries with quantum advice.
Several issues arise when we generalize the classical proofs to handle quantum leakage.
In particular, we do not know how to prove the security of the weak PRF-based construction of~\cite{Pie09} or the security against general quantum adversaries (with unbounded quantum storage).
Furthermore, our proof generalizes the classical proof of~\cite{JP14}, but with certain necessary modifications to make the proof go through in the quantum setting.
We discuss the issues we encounter after describing the construction of~\cite{DziembowskiP08}.

\subsection{Quantum Leakage-resilient Stream Cipher}\label{subsec:stream-cipher}

In this section, we generalize the leakage-resilient stream-cipher defined in~\cite{DziembowskiP08} to capture quantum leakage in the bounded-quantum-storage model.
We first review the classical model informally.
A \emph{stream-cipher} is given by a function $\SC:\zo^m\to\zo^m\times\zo^n$.
Suppose the initial internal state is $S^{(0)}\in\zo^m$.
In the $i$-th round, $(S^{(i)}, X^{(i)}) = \SC(S^{(i-1)})$ is computed.
When we recursively apply the function $\SC$, the internal state evolves and generates the output $X^{(1)}, X^{(2)}, \dots$.
Informally, a stream cipher $\SC$ is secure if for all $i$, $X^{(i)}$ is pseudorandom given $X^{(1)}, \dots, X^{(i-1)}$.

\paragraph{Classical Leakage-resilient Stream Cipher}

Recall that for classical \emph{leakage-resilient stream cipher}, at round $i$, the adversary learns not only the output $X^{(i)}$ but also some bounded length leakage about the internal state $S^{(i-1)}$ that used for generating $X^{(i)}$.
More precisely, the security is captured by the following game.
Suppose the adversary's memory right before round $i$ is $V^{(i-1)}$, which could contain all the leakages and outputs ahead this round and any information that can be inferred from them.
Also a leakage function $f^{(i)}$ \footnote{Here we use $f^{(i)}$ to denote a leakage function used in the $i$-th round, but not applying the function $i$ times.}
 with range at most $\ell$ bits is already chosen by the adversary.
During round $i$, when $(S^{(i)}, X^{(i)}) = \SC(S^{(i-1)})$ is computed, the adversary $A$ can learn some leakage information $\Lda^{(i)} = f^{(i)}(\hat{S}^{(i-1)})$ where $\hat{S}^{(i-1)}$ denotes the part of $S^{(i-1)}$ that is used for evaluating $\SC(S^{(i-1)})$ (namely, following the ``\emph{only computation leak}'' model~\cite{MicaliR04}).
We say the stream cipher is $(s, \eps, q, \ell)$-secure if for all $1\leq i\leq q$, no distinguisher of size $s$ can distinguish $X^{(i)}$ from a uniformly random string with advantage more than $\eps$ before round $i$ (before seeing $X^{(i)}$ and the leakage $\Lda^{(i)}$ of length at most $\ell$) with advantage more than $\eps$.

\paragraph{Bounded-quantum-storage Model}
In the classical case, there is no restriction on the size of adversary's memory, while in the bounded quantum storage model, we assume the adversary has an $\ell$-qubit quantum memory.
Here we choose $\ell$ to be same as the maximal length of a leakage in each round for convenience.
We use $(V^{(i-1)}, \tau^{(i-1)})\in\zo^*\ot\density{\bbC^{2^\ell}}$ to represent the memory state after round $(i-1)$.

The game in the quantum setting is different from the classical setting in the following way.
At round $i$, the adversary receives an $\ell_i$-qubit leakage according to the leakage function $f^{(i)}$ where $\ell_i \leq \ell$.
(Recall that $f^{(i)}$ can be chosen adaptively, so it can be correlated to $(V^{(i-1)}, \tau^{(i-1)})$.
Also, ($V^{(i-1)}, \tau^{(i-1)})$ is the memory state after the $f^{(i)}$ is decided.)
We use $\tau_{[j]}$ to denote the quantum state in the first $j$ bits system.
Namely, the last $\ell-j$ qubits of $\tau$ are traced out.
If we write a state $\tau$ as
\[\tau = \tau_{YZ} \in \density{\cY\ot\cZ} = \density{\bbC^{2^{\ell-\ell'}}\ot\bbC^{2^{\ell'}}},\]
then $\tau_{[\ell']} = \tr[\cZ]{\tau}$.
Suppose $\lda^{(i)} = f^{(i)}(\hat{S}^{(i-1)})$ is the (quantum) leakage in the $i$-th round.
After receiving the leakage, the adversary replaces the last $\ell_i$ qubits in its quantum memory by $\lda^{(i)}$.
The memory state then becomes $((V^{(i-1)}, X^{(i)}), \tau^{(i-1)}_{[\ell-\ell_i]}\ot\lda^{(i)})$.
Next, the adversary prepare the leakage function $f^{(i+1)}$ for the next round, and the memory state becomes $(V^{(i)}, \tau^{(i)})$.

Now we define a security game $G^q_0$ between a stream cipher $\SC$ and an adversary $A$.

\begin{center}
  \begin{Sbox}\begin{minipage}{\columnwidth-40pt}
    \begin{enumerate}
      \item Initially, the cipher randomly generates a secret state $S^{(0)} \in\zo^m$.
        Adversary $A$ generates the leakage function $f^{(1)}$.
        Let adversary's memory state be $(V^{(0)}, \tau^{(0)})$.
      \item For $i = 1, \dots, q-1$
      \begin{enumerate}
        \item $(S^{(i)}, X^{(i)}) = \SC(S^{(i-1)})$, where $S^{(i)}$ is the new secret state, and $X^{(i)}$ is the output at round $i$.
        \item An $\ell_i$-qubit quantum leakage $\lda^{(i)} = f_i(\hat{S}^{(i-1)})$ where $\ell_i\leq\ell$ is given to the adversary.
          $\hat{S}^{(i-1)}$ denotes the part of $S^{(i-1)}$ that is used for evaluating $\SC(S^{(i-1)})$.
        \item After seeing $X^{(i)}$ and placing $\lda^{(i)}$ in the last $\ell_i$ qubits of quantum memory, the adversary's memory state becomes $((V^{(i-1)}, X^{(i)}), \tau^{(i-1)}_{[\ell-\ell']}\ot\lda^{(i)})$.
        \item Adversary produces the leakage function $f^{(i+1)}$ for the next round.
          The memory state becomes $(V^{(i)}, \tau^{(i)})$.
      \end{enumerate}
      \item $(S^{(q)}, X^{(q)}) = \SC(S^{(q-1)})$.
      \item The adversary $A$ is given $X^{(q)}$ and outputs a bit.
        \end{enumerate}
  \end{minipage}\end{Sbox}
  \fbox{\TheSbox}
\end{center}

We use $A(G)$ to denote the output of the adversary $A$ in a game $G$.
The game $G$ implicitly depends on the stream cipher $\SC$.
Also, we say two games $G_1$ and $G_2$ are $\eps$-indistinguishable by $A$ if
\[\left|\pr{A(G_1) = 1} - \pr{A(G_2) = 1}\right| \leq \eps.\]
To define the security of a stream cipher, we also consider the game $\tG^q_0$, which is identical to the game $G^q_0$, except $X^{(q)}$ is replaced by a uniform random string with a same length at the end.
Now we are ready to define the security of a quantum leakage-resilient stream cipher:

\begin{definition}
  A quantum leakage-resilient stream cipher $\SC:\zo^m\to\zo^m\times\zo^n$ is $(s, \eps, q, \ell)$-secure in the bounded quantum storage model if for every quantum adversary $A$ of size $s$ with an $\ell$-qubit memory and every $q'\in[q]$, $G^{q'}_0$ and $\tG^{q'}_0$ are $\eps$-indistinguishable by $A$.
  Namely
  \[\left|\pr{A(G^{q'}_0) = 1} - \pr{A(\tG^{q'}_0) = 1}\right| \leq \eps.\]
\end{definition}

\subsection{Construction}
The construction follows the one in \cite{DziembowskiP08}.
First, we define a function $F:\zo^{k+n}\to\zo^{k+n}$, which serves as a building block of the construction.
\[F(K, X) = \Prg(\Ext(K, X), X)\]
where $\Ext:\zo^{k+n}\to\zo^m$ is a quantum-proof strong randomness extractor (\eg~Trevisan's extractor~\cite{Tre01, DePVR12}) and $\Prg:\zo^m\to\zo^{k+n}$ is a pseudorandom generator secure against quantum adversary.
The existence of quantum-secure PRGs can base on the quantum security of primitives implying PRG (\eg~lattice assumptions or quantum-secure one-way functions)~\cite{Song14}.
More specifically,
\begin{itemize}
  \item Quantum-proof strong randomness extractor~\cite{DePVR12}:

    We say $\Ext:\zo^{k+n}\to\zo^m$ is an \emph{$(k_{\Ext}, \eps_{\Ext})$-quantum-proof extractor} if for all cq-state $\rho_{KV}$ with $\Hmin(K|V)_{\rho} \geq k_{\Ext}$, the trace distance between two ccq-state
    \[(\Ext(K, U_n), U_n, V) \mbox{ and } (U_m, U_n, V)\]
    is at most $\eps_{\Ext}$ where $m = k_{\Ext} - 4\log(1/\eps_{\Ext}) - O(1)$. $U_m$ and $U_n$ are uniform distribution over $m$ and $n$ bits, respectively. (Two $U_n$s in the first state are the same sample from a uniform distribution.)
  \item Quantum-secure pseudorandom generator:

    We say $\Prg:\zo^m\to\zo^{n}$ is an $(s_{\Prg}, \eps_{\Prg})^*$-quantum-secure if for all quantum distinguisher $A$ of size $s_{\Prg}$ with quantum advice,
    \[\left|\pr{A(\Prg(U_m)) = 1} - \pr{A(U_{n}) = 1}\right| \leq \eps_{\Prg}.\]
\end{itemize}

Combine the properties of the extractor and the pseudorandom generator, we have the following claim:
\begin{claim}\label{claim:f}
  Let $\Ext:\zo^{k+n}\to\zo^m$ be a $(k_{\Ext}, \eps_{\Ext})$-quantum-proof extractor, $\Prg:\zo^{m+n}\to\zo^{k+n}$ be an $(s_{\Prg}, \eps_{\Prg})$-quantum-secure and define $F:\zo^{k+n}\to \zo^{k+n}$ to be $F(K, X) = \Prg(\Ext(K, X), X)$.
  If a cq-state $\rho_{KV}\in\zo^k\ot\density{\cH}$ satisfies $\Hmin(K|V)_{\rho} \geq k_{\Ext}$, then for all $(s_{\Prg})$-size quantum distinguisher $A$, we have
  \[\left|\pr[X\sim U_n]{A(F(K, X), V) = 1} - \pr{A(U_{k+n}, V) = 1}\right| \leq \eps_{\Ext}+\eps_{\Prg}.\]
\end{claim}
\begin{proof}
  \begin{align*}
    & \size{\pr[X\sim U_n]{A(F(K, X), V) = 1} - \pr{A(U_{k+n}, V) = 1}}\\
    = & \size{\pr[X\sim U_n]{A(\Prg(\Ext(K, X), X), V) = 1} - \pr{A(U_{k+n}, V) = 1}}\\
    \leq & \size{\pr[X\sim U_n]{A(\Prg(U_m, X), V) = 1} - \pr{A(U_{k+n}, V) = 1}} + \eps_{\Ext}\\
    \leq & \eps_{\Prg} + \eps_{\Ext}.
  \end{align*}
  The first inequality is because the trace distance between $(\Ext(K, X), X, V)$ and $(U_m, X, V)$ is at most $\eps_{\Ext}$.
  Applying a same function on two quantum states can only decrease the trace distance.
  The second inequality is due to the property of the quantum-secure pseudorandom generator defined above.
  Here $V$ can be seen as a quantum advice.
\end{proof}

Based on the function $F$, we define the $q$-round stream cipher $\SC$ as follows.
Let $S^{(i)} = (K^{(i)}, K^{(i+1)}, X^{(i)})$ where $K^{(i)}\in\zo^k, X^{(i)}\in\zo^n$.
Define
\[ \SC(S^{(i-1)}) = (S_i, X_i) = ((K^{(i)}, K^{(i+1)}, X^{(i)}), X^{(i)}) \mbox{ where } (K^{(i+1)}, X^{(i)}) = F(K^{(i-1)}, X^{(i-1)}). \]
Note that we repeat $X^{(i)}$ in the internal state $S^{(i)}$ just to make the definition consistent with the definition of stream cipher previously.
Clearly, $K_i$ is intact when calculating $(S^{(i)}, X^{(i)})$ from $S^{(i-1)}$, so $\hat{S^{(i-1)}} = K^{(i-1)}$.

Now we discuss the issues we encounter when we generalize the classical proofs of~\cite{DziembowskiP08,Pie09,JP14}.
The main problem is that many steps of the proofs require conditioning on the value of the leakage or the view of the adversary.
Again, we cannot do such conditioning for a quantum state, and this is reminiscent to the difficulty of proving that a randomness extractor is secure with quantum side information.
For example, the proofs of~\cite{DziembowskiP08,Pie09} based on the Leakage Chain Rule proceeds by conditioning on the value of the leakage in the previous rounds.
Also, the construction of~\cite{Pie09} relies on leakage-resilient weak PRFs, whose security proofs~\cite{Pie09,BDKPPSY11} require to condition on the value of the leakage.\footnote{More precisely, they show that the output of the weak PRF remains pseudorandom when the key has sufficiently high entropy, and use the fact that with high probability over the leakage value, the key conditioned on the leakage value has high entropy.}
Also, we mention that the security proof of~\cite{Pie09} for leakage-resilient weak PRF applies a  gap application procedure to, which suffers our barrier result in Section~\ref{sec:barrier}.

Fortunately, when we apply the proof of~\cite{JP14} based on the classical Leakage Simulation Lemma to the construction of~\cite{DziembowskiP08}, and replace the classical lemma with our quantum Leakage Simulation Lemma, we avoid most of the issues of such conditioning. Roughly speaking, the Leakage Simulation Lemma allows us to simulate the leakage as an efficient function of the view of the adversary without any conditioning. However, in one step of the proof we need to argue that the extractor $\Ext(K^{(i)}, X^{(i)})$ can extract (pseudo)entropy from $K^{(i)}$ using seed $X^{(i)}$, which requires to argue the independence of $K^{(i)}$ and $X^{(i)}$. The classical proof establishes independence by conditioning on the view of the adversary~\cite{DP07}. However, such conditional independence does not hold when the leakage becomes quantum. We resolve this issue by observing that the independence can be established in the hybrids, which is sufficient to carry the proof.

\begin{theorem}\label{thm:resilient}
  Let $\eps_{\Ext}=\eps_{\Prg}=\eps/8q$.
  There exists $s_{\Prg} = \poly(s, 2^{\ell}, 1/\eps, q, n, k)$ such that if $\Prg:\zo^m\to\zo^{k+n}$ is an $(s_{\Prg}, \eps_{\Prg})$-quantum-secure pseudorandom generator and $\Ext:\zo^{k+n}\to\zo^m$ is an $(\eps_{\Ext}, k-\ell)^*$-quantum-proof extractor, then the above construction for $\SC$ is a $(s, \eps, q, \ell)$-secure leakage-resilient stream cipher.
\end{theorem}

As one will see, the main reason that we can only handle the bounded-quantum-storage adversary is that we do not have a simulation leakage lemma for a quantum source.
We have to treat the whole quantum state maintained by the adversary as a leakage.

\begin{proof}
  Let $s_{\Prg} = s + s_{\SC} + s_g$ where $s_{\SC}$ is the circuit size of the $q$-round stream cipher and $s_g = \poly(s, s_{\SC}, n, k, 2^\ell, q, 1/\eps)$ is the circuit size of a "leakage simulator", which will be defined later.
  In this proof, for all $0\leq i \leq q$, $\tX^{(i)}$ and $\tK^{(i)}$ represent independent uniform distributions over $\zo^n$ and $\zo^k$, respectively.

  We will define the hybrid games $G^q_{0}, G^q_{0\to1}, G^q_{1}, \dots, G^q_{q-1}$.
  Then show that a bounded adversary cannot distinguish neighboring two games.

  \begin{itemize}
    \item {\bf From $G^q_{i}$ to $G^q_{i\to{i+1}}$}

      $G^q_{i}$ and $G^q_{i\to{i+1}}$ are identical until round $(i-1)$.
      At round $i$ of $G^q_{i\to{i+1}}$, after the adversary queries $f^{(i)}(\tK^{(i-1)})$, its quantum memory becomes $g^{(i)}(V^{(i-1)}, X^{(i)}, K^{(i+1)})$ instead of $\tau^{(i-1)}_{[\ell-\ell_i]}\ot f^{(i)}(\tK^{(i-1)})$.
      Here $g^{(i)}$ is a quantum simulator of size $s_{g} = \poly(s, s_{\SC}, n, k, q, 2^{\ell}, 1/\eps)$ such that
      \[\left(V^{(i-1)}, X^{(i)}, K^{(i+1)}, \tau^{(i-1)}_{[\ell-\ell_i]}\ot f^{(i)}(\tK^{(i-1)})\right) \mbox{ and } \left(V^{(i-1)}, X^{(i)}, K^{(i+1)}, g^{(i)}(V^{(i-1)}, X^{(i)}, K^{(i+1)})\right)\]
      are $(s_{\SC}+s, \eps/4q)$-indistinguishable.
      The existence and the property of the simulator $g^{(i)}$ is by the Leakage Simulation Lemma.
      Specifically, treating the $(V^{(i-1)}, X^{(i)}, K^{(i+1)})$ as $X$ and $\tau_{[\ell-\ell_i]}^{(i-1)}\ot f^{(i)(\tK^{(i-1)})}$ as $B$ in Theorem~\ref{thm:leak-sim}, then we show the existence of $g^{(i)}$.
      Now we prove that the games $G^q_i$ and $G^q_{i\to i+1}$ are $\eps/4q$-indistinguishable by any $s$-size quantum distinguisher $A$.
      \begin{claim}
        For all $s$-size quantum distinguisher $A$,
        \[\size{ \pr{A(G^q_i) = 1} - \pr{A(G^q_{i\to i+1}) = 1} } \leq \frac{\eps}{4q}.\]
      \end{claim}
      \begin{proof}
        Let \begin{align*}
          & L^q_i = \left(V^{(i-1)}, X^{(i)}, K^{(i+1)}, \tau^{(i-1)}_{[\ell-\ell_i]}\ot f^{(i)}(\tK^{(i-1)})\right) \\
          \mbox{ and } & L^q_{i\to i+1} = \left(V^{(i-1)}, X^{(i)}, K^{(i+1)}, g^{(i)}(V^{(i-1)}, X^{(i)}, K^{(i+1)})\right).
        \end{align*}
        Assuming there exists an $s$-size quantum distinguisher $A$ for games $G^q_i$ and $G^q_{i\to i+1}$, we define the distinguisher $A'$ for $L^q_i$ and $L^q_{i\to i+1}$.
        When the input is $L^q_i$ (resp., $G^q_{i\to i+1}$), the distinguisher $A'$ emulates the game $G^q_i$ (resp., $G^q_{i\to i+1}$) between $\SC$ and the adversary $A$ starting from the leakage in the $i$-th round, then output the decision of $A$.
        The state of the stream cipher after round $i$ is $(X^{(i)}, \tK^{(i)}, K^{(i+1)})$.
        First two terms $X^{(i)}$ and $K^{(i+1)}$ are provided by both $L^q_i$ and $L^q_{i\to i+1}$, while $\tK^{(i)}$ is merely an uniform string.
        For the adversary's memory in both games, it can also be obtained from $L^q_i$ and $L^q_{i\to i+1}$, respectively.
        Therefore, the size of the distinguisher $A'$ is $s_{\SC} + s$.
        For such $A'$, we have
        \[\size{\pr{A'(L^q_i) = 1} - \pr{A'(L^q_{i\to i+1}) = 1} } = \size{ \pr{A(G^q_i) = 1} - \pr{A(G^q_{i\to i+1}) = 1} } > \frac{\eps}{4q},\]
        which contradicts the indistinguishability.
      \end{proof}

    \item {\bf From $G^q_{i\to{i+1}}$ to $G^q_{i+1}$}

      We define the game $G_{i+1}$ from $G_{i\to{i+1}}$ as follows.
      After step~(b) at round $i$, the output $X^{(i)}$ and part of the internal state $K^{(i+1)}$ are replaced by $\tX^{(i)}$ and $\tK^{(i+1)}$, respectively.
      Note that the replacements lead to the following changes.
      \begin{enumerate}
        \item In the step (c) of round $i$, the adversary sees $\tX^{(i)}$ instead, so the classical part of the memory becomes $(V^{(i-1)}, \tX^{(i)})$ after the $i$-th leakage in Game $G_{i+1}$.
        \item $\tlda_i$ in the above is also changed implicitly, since it depends on $X^{(i)}$ and $K^{(i+1)}$.
          Namely, $\tlda^{(i)} = g_i(V^{(i-1)}, \tX^{(i)}, \tK^{(i+1)})$.
        \item At round $(i+2)$ (if exists), $\lda^{(i+2)}$ is also changed to $f^{(i+2)}(\tK^{(i+1)})$.
      \end{enumerate}
      We prove the following claim.
      \begin{claim}
        For all $s$-size quantum distinguisher $A$,
        \[\size{ \ex{A(G^q_{i\to i+1})} - \ex{A(G^q_{i+1})} } \leq \frac{\eps}{4q}.\]
      \end{claim}
      \begin{proof}
        First, we argue that $(K^{(i+1)}, X^{(i)})$ is pseudorandom in $G^q_{i\to i+1}$ against the adversary right before the $i$-th leakage.
        The memory state of the adversary is $(V^{(i-1)}, \tau^{(i-1)})$ at that moment.
        Recall that $(K^{(i+1)}, X^{(i)}) = F(\tK^{(i-1)}, \tX^{(i-1)})$ in Game $G^q_{i\to i+1}$ .
        As long as we can argue that
        \begin{enumerate}
          \item $(\tK^{(i-1)}, \tX^{(i-1)})$ are independent given $V^{(i-1)}, \tau^{(i-1)}$ and
          \item $\Hmin(\tK^{(i-1)}|V^{(i-1)}, \tau^{(i-1)}) \geq k-\ell$,
        \end{enumerate}
        then we can apply Claim~\ref{claim:f} to prove that $(K^{(i+1)}, X^{(i)})$ is $(s_{\Prg}, \eps/4q)$-pseudorandom given $(V^{(i-1)}, \tau^{(i-1)})$.

        For the first condition, we notice that before round $i$, $\tK^{(i-1)}$ was only used by the simulator $g^{(i-2)}$, so it is correlated to $(V^{(i-2)}, \tau^{(i-2)})$.
        On the other hand, $(V^{(i-1)}, \tau^{(i-1)})$ is decided from both $(V^{(i-2)}, \tau^{(i-2)})$ and $\tX^{(i-1)}$, which is generated after the adversary's memory state became $(V^{(i-2)}, \tau^{(i-2)})$.
        Therefore, conditioned on $(V^{(i-1)}, \tau^{(i-1)})$, $\tX_{i-1}$ and $\tK_{i-1}$ are independent.

        To show the second condition, clearly $\tK_{i-1}$ is uniform and has entropy $k$ before seeing any leakage related to $\tK_{i-1}$.
        Before the $i$-th round, the only information about $\tK_{i-1}$ learned by an adversary is $g^{(i-2)}(\tX^{(i-3)}, \tX^{(i-2)}, \tK^{i-1})$.
        After some operations on the leakage, the conditional entropy of $\tK^{(i-1)}$ can only increase.
        Therefore,
        \[\Hmin(\tK^{(i-1)}|V^{(i-1)}, \tau^{(i-1)}) \geq \Hmin(\tK_{i-1}|g^{(i-2)}(\tX^{(i-3)}, \tX^{(i-2)}, \tK^{i-1})) \geq k - \ell.\]
        The last inequality is from Lemma~\ref{thm:quantum-chain-rule-sep} and the fact that the length of $g^{(i-2)}(\tX^{(i-3)}, \tX^{(i-2)}, \tK^{i-1})$ is at most $\ell$.

        So far we have shown that $(K^{(i+1)}, X^{(i)})$ and $(\tK^{(i+1)}, \tX^{(i)})$ are $(s_{\Prg}, \eps/4q)$-indistinguishable conditioned on $(V^{(i-1)}, \tau^{(i-1)})$.
        Assume there exists an $s$-size quantum distinguisher $A$ to distinguish the games $G^q_{i\to i+1}$ and $G^q_{i+1}$ with advantage more than $\eps/4q$.
        Then we can construct an distinguisher $A'$ of size $s_{\Prg}$ to distinguish $(K^{(i+1)}, X^{(i)}, V^{(i-1)}, \tau^{(i-1)})$ from $(\tK^{(i+1)}, \tX^{(i)}, V^{(i-1)}, \tau^{(i-1)})$ with the same advantage.
        Namely,
        \begin{align*}
          &\size{\pr{A'(K^{(i+1)}, X^{(i)}) = 1} - \pr{A'(\tK^{(i+1)}, \tX^{(i)}) = 1}}\\
           = & \size{\pr{A(G^q_{i\to i+1})=1} - \pr{A(G^q_{i+1}) = 1}} > \eps/4q.
        \end{align*}
        Now we calculate the circuit size of $A'$.
        The distinguisher $A'$ emulates the game $G$ between $\SC$ and the adversary $A$ starting from the leakage in the $i$-th round, and output $A$'s decision.
        The state of the stream cipher after the $i$-th leakage is $(X^{(i)}, \tK^{(i)}, K^{(i+1)})$.
        The terms $X^{(i)}$ and $K^{(i+1)}$ (or $\tX^{(i)}$ and $\tK^{(i+1)}$) are given as inputs of $A'$, and $\tK^{(i)}$ is merely an uniform string.
        For the adversary's memory in both games, it is also provided by the input.
        In order to provide the simulated leakage in the $i$-th round, (since the leakage function $f^{(i)}$ is replaced by $g^{(i)}$), we have to run the simulation circuit $g^{(i)}$.
        To sum up, the size of $A'$ is at most $s + s_{\SC} + s_g = s_{\Prg}$, which contradicts the pseudorandomness of $(K^{(i+1)}, X^{(i)})$.
      \end{proof}
  \end{itemize}

  We have shown that both $G^q_{i}, G^q_{i\to i+1}$ and $G^q_{i\to i+1}, G^q_{i+1}$ are $\eps/4q$-indistinguishable by an $s$-size quantum adversary with $\ell$-bounded quantum storage.
  By a hybrid argument, we have, for all $s$-size quantum adversary with $\ell$-bounded quantum storage,
  \begin{align}\label{eq:game-dist}
    \left|\pr{A(G^q_0) = 1}-\pr{A(G^q_q) = 1}\right| \leq q(\eps_{\Prg}+\eps_{\Ext}) \leq \frac{\eps}{2}.
  \end{align}
  Similarly, we can define games $\tG^q_i$ (resp., $\tG^q_{i\to i+1}$) with replacing $X_q$ by $\tX_q$ at the end of the games $G^q_i$ (resp., $G^q_{i\to i+1}$) for all $i$.
  By the same argument, we also have
  \[ \left|\pr{A(\tG^q_0) = 1}-\pr{A(\tG^q_q) = 1}\right| \leq \frac{\eps}{2}. \]
  Note that $G^q_q = \tG^q_q$.
  Combining both inequalities, we have that for all $s$-size quantum adversary with $\ell$-bounded quantum storage,
  \begin{align*}
  &\left|\pr{A(G^q_0) = 1}-\pr{A(\tG^q_0)}\right|\\
  \leq & \left|\pr{A(G^q_0) = 1}-\pr{A(G^q_q) = 1}\right| + \left|\pr{A(\tG^q_0) = 1}-\pr{A(\tG^q_q) = 1}\right|\leq \eps,
  \end{align*}
  which concludes the proof.
\end{proof}

\section{Barrier - Gap Amplification Problem} \label{sec:barrier}

In Section~\ref{sec:leakage-chain-rule}, we have seen the leakage chain rule for quantum relaxed-HILL pseudoentropy for ccq-states.
However, the chain rule for cqq-states, where the source and prior knowledge are both quantum, is still unknown, and nor is the connection between the guessing pseudoentropy and HILL pseudoentropy.

One of the main challenge in extending classical proofs from classical to quantum cases is due to the celebrated Wootters-Zurek no-cloning theorem~\cite{WZ82}.
Here we exhibit another barrier ---  the \emph{gap amplification problem} defined as follows.
Given a quantum distinguisher $A$ (whose input is a quantum state $\rho$), where the acceptance probability is greater than $p$ for YES instances and less than $q$ for NO instances, can we have another quantum distinguisher $A'$ where the gap $p'-q'$ is larger than that in $A$?
If we were able to clone an arbitrary quantum state, then the gap amplification would be easy (as discussed below).
Thus, we can view the gap amplification problem as a special case of the no-cloning theorem.
Moreover, we will show that the impossibility of amplifying the gap implies that imperfect cloning of a single qubit to within a constant in trace distance is impossible.

In the classical case, the gap amplification result provides the robustness of BPP definition in the way that no matter what the acceptance probabilities for both YES and NO instances are, the definitions for BPP are equivalent as long as the gap is non-negligible.
Similarly, in the quantum setting, the gap amplification problem is connected to the amplification of the acceptance probability of quantum proofs in QMA.
The gap amplification problem is trivial in the classical case, as there is no cloning restriction in the classical world.
For a given input, we can make copies of the input, run the original algorithm multiple times, and then use a majority or threshold rule to reduce the error probability via a concentration bound (\eg~Chernoff bound).
However, in the quantum case, due to the no-cloning theorem, it is not obvious that we can do it.
Note that the no-cloning theorem does not directly imply the impossibility of amplification, but we can use the similar concept in proving no-cloning theorem to show the impossibility of amplification.
On the other hand, the impossibility of amplification implies not only no-cloning theorem, but also the imperfect cloning~\cite{GM97} for arbitrary states.

First, we define the gap amplification problem as follows.

\begin{definition}[$\GapAmp$ Problem] \label{def:gap-amplify}
  Let $A:\density{\bbC^M}\to\bin$, $0 < q < p < 1$. 
  We say that a quantum distinguisher $A':\density{\bbC^M}\to\bin$ is a $(p, q)$-amplified version of $A$ if for every input $\ket{\psi}\in \ball{\bbC^M}$,
    \[\Pr[A'(\ketbra{\psi}) = 1]
      \begin{cases}
        > p & \mbox{ if } \Pr[A(\ketbra{\psi}) = 1] \geq p\\
        < q & \mbox{ if } \Pr[A(\ketbra{\psi}) = 1] \leq q\\
      \end{cases}.
    \]
\end{definition}

Then we show that except for trivial cases, one cannot push the error rate ($1-p$ and $q$ for two cases respectively) arbitrarily close to~1.
\begin{theorem} \label{thm:barrier}
  For every real numbers $0 < q < p < 1$, there exists a quantum distinguisher $A:\density{\bbC^2}\to\bin$  such that no $(p, q)$-amplified version of $A$ (even of unbounded size) exists.
\end{theorem}
\begin{proof}
  Let $A$ be a single-qubit measurement in the computational basis $\{\ket{0},\ket{1}\}$.
  Consider the pure states $\ket{\psi} = (\cos\alpha)\ket{0}+ (\sin\alpha)\ket{1}$ and $\ket{\phi} = (\cos\beta)\ket{0}+ (\sin\beta)\ket{1}$, where $\alpha = \sin^{-1}(\sqrt{p})$ and $\beta = \sin^{-1}(\sqrt{q})$.
  Thus $\Pr[A(\ket{\psi}) = 1] = p$ and $\Pr[A(\ket{\phi}) = 1]=q$.

  Let the BPOVM of $A'$ be $\Pi = \begin{bmatrix}a & -b+ci\\-b-ci & d\end{bmatrix}$ for $0\leq a,d \leq 1$ and appropriate real numbers $b$ and $c$ such that that $\Pi\geq 0$.
  Assume that $A'$ is a $(p,q)$-amplified version of $A$ such that
  $\inprod{\Pi, \ketbra{\psi}} > \sin^2\alpha$ and $\inprod{\Pi, \ketbra{\phi}} < \sin^2\beta$.
  That is,
  \begin{align*}
    a\cos^2\alpha - 2b\sin\alpha\cos\alpha + d\sin^2\alpha >& \sin^2\alpha,\\
    a\cos^2\beta - 2b\sin\beta\cos\beta + d\sin^2\beta <& \sin^2\beta.
  \end{align*}
  After dividing the two inequalities by $\sin^2\alpha$ and $\sin^2\beta$, respectively, we obtain
  \begin{align}
    a\cot^2\alpha + d &> 1 + 2b\cot\alpha, \label{eq:a-big}\\
    a\cot^2\beta + d &< 1 + 2b\cot\beta. \label{eq:a-small}
  \end{align}
  Since  $d\leq 1$, we have $a > \frac{2b}{\cot\alpha}$ by Equation~(\ref{eq:a-big}).
  On the other hand, subtracting Equation~(\ref{eq:a-big}) from Equation~(\ref{eq:a-small}) and dividing it by $(\cot\beta-\cot\alpha)$, which is positive by the choices of $\alpha$ and $\beta$, we get $a < \frac{2b}{\cot\beta+\cot\alpha}<\frac{2b}{\cot\alpha}$.
  That gives a contradiction.

\end{proof}

Now we can prove the impossibility of imperfect cloning for a single qubit from Theorem~\ref{thm:barrier}.
\begin{theorem}
  Let $C$ be a quantum circuit with input space $\ball{\bbC^2}$ and output space $\ball{\bbC^{2^2}}$.
  For every $0<\eps< 0.002$, there exists  $\ket{\psi}\in \ball{\bbC^2}$ such that $T(C(\ket{\psi}), \ket{\psi}\otimes\ket{\psi}) > \eps$.
\end{theorem}
\begin{proof}
  Suppose the statement is false. That is, for all $\ket{\psi}\in \ball{\bbC^2}$, we have $T(C(\ket{\psi}), \ket{\psi}\otimes\ket{\psi}) < \eps$.
  Then given a quantum distinguisher $A$ with input space $\ball{\bbC^{2}}$, we define a quantum distinguisher $A'$ with the same input space as follows.
  \begin{enumerate}
    \item Run $C$ on input $\ket{\psi}$.
    \item Run $C\otimes C$ on $C(\ket{\psi})$.
    \item Run the algorithm $A\otimes A\otimes A\otimes A$ on  $C\otimes C(C(\ket{\psi}))$ and obtain four binary outcomes.
    \item Output the majority of the four outcomes. Randomly output $0$ or $1$ if we have a tie vote.
  \end{enumerate}
  We claim that $A'$ is a $(2/3, 1/3)$-amplified version of $A$.
  \begin{align*}
    &T\left(C\ot C\left(C(\ket{\psi})\right), \ket{\psi}^{\ot 4}\right)\\
    \leq &T\left(C\ot C\left(C(\ket{\psi})\right), C\ot C(\ket{\psi}^{\ot 2})\right)+T\left(\left(C(\ket{\psi})\right)^{\ot 2}, \ket{\psi}^{\ot 4}\right).
  \end{align*}
  Note that $T\left(C\ot C\left(C(\ket{\psi})\right), C\ot C(\ket{\psi}^{\ot 2})\right)\leq  T\left(C(\ket{\psi}), \ket{\psi}^{\ot 2}\right)<\eps$
  since a trace-preserving quantum operation does not increase the trace distance~\cite{NC00}.
  On the other hand,
  \begin{align*}
    T(\rho\ot\rho, \sigma\ot\sigma)\leq & \frac{1}{2}T(\rho\ot(\rho+\sigma), \sigma\ot (\rho+\sigma)) +\frac{1}{2}T( (\rho+\sigma)\ot\rho, (\rho+\sigma)\ot\sigma)\\
    =& \frac{1}{4}\Tr\left|(\rho-\sigma)\ot(\rho+\sigma)\right| +\frac{1}{4}\Tr\left| (\rho+\sigma)\ot(\rho-\sigma)\right|\\
    =& \frac{1}{2}\Tr\left|\rho-\sigma\right|\Tr\left|\rho+\sigma\right|= 2  T(\rho, \sigma).
  \end{align*}
  Thus \[T\left(\left(C(\ket{\psi})\right)^{\ot 2}, \ket{\psi}^{\ot 4}\right)\leq 2 T\left(C(\ket{\psi}), \ket{\psi}^{\ot 2}\right)< 2\eps.\]
  Therefore,
  \[T\left(C\ot C\left(C(\ket{\psi})\right), \ket{\psi}^{\ot 4}\right)<3\eps.\]

  Suppose $A(\ket{\psi}) \geq 2/3 \triangleq p$ and $A(\ket{\phi}) \leq 1/3= 1-p$.
  Then
  \begin{align*}
    \Pr[A'(\ket{\psi}) = 1] &> (p^4-3\eps) + 4(p^3(1-p)-3\eps)+3(p^2(1-p)^2-3\eps),\\
    \Pr[A'(\ket{\phi}) = 1] &> ((1-p)^4+3\eps) + 4((1-p)^3p+3\eps)+3(p^2(1-p)^2+3\eps).
  \end{align*}
  If we set $\eps = 0.002$, then $\Pr[A'(\ket{\psi}) = 1] > 0.6927$ and $\Pr[A'(\ket{\phi}) = 1] < 0.3073$. This means that $A'$ is a $(2/3, 1/3)$-amplified version of $A$, which contradicts Theorem~\ref{thm:barrier}.

\end{proof}

\section*{Acknowledgments}
We are grateful to anonymous reviewers for pointing out that some results in an earlier version of this paper were already known.
YHC and SV thank Mark Bun, Aram Harrow and Mehdi Soleimanifar for useful discussion.
KMC is grateful to Krzysztof Pietrzak for an inspiring discussion that led to this research.
CYL acknowledges useful discussions with Todd~A.~Brun and Nengkun~Yu.

\newcommand{\etalchar}[1]{$^{#1}$}

\appendix

\section{Tomography}\label{app:tomography}
In a quantum tomography problem, one wants to learn the behavior or even a description of a quantum circuit or quantum state.
In our applications to the MMWU proof of quantum simulation leakage lemma, one tasks is that, given a quantum distinguisher, find a corresponding BPOVM matrix.
The task is precisely formulated in Definition~\ref{def:qckt-tomography-app}.
Also in the min-max proof, which will be illustrated in Appendix~\ref{app:min-max-proof}, we would like to find a quantum state that maximizes the acceptance probability of a given quantum distinguisher.
This task is formulated in Definition~\ref{def:qckt-max-sat}.
In the most general parameters in the tomography problems, let the dimension of input space be $d$ and the circuit size is $s$.
Then the desired algorithm complexity is $\poly(s, \log d)$, which is too demanding to achieve.
Fortunately, in our application, the $\poly(s, d)$ time complexity is already satisfying.
Our tomography algorithm also uses a solution to the $\QCktValue$ Problem (Definition~\ref{def:qckt-value}), described as follows.

\begin{definition}[$\QCktValue$ Problem]\label{def:qckt-value}
  The $\QCktValue(s, \eps, \gamma)$ problem is a computational problem defined as follows:
  \begin{itemize}
    \item Input: a description of a quantum circuit $C$ of size $s$ with binary output $\{0,1\}$, and an error parameter $0<\eps<1$.
    \item Task: with probability at least $1-\gamma$, output an estimate $\tilde{p}$ of the probability $p = \pr{C = 1}$ such that $|\tilde{p}-p|\leq \eps$.
  \end{itemize}
\end{definition}

\begin{lemma}\label{lemma:qkct-value}
  There exists a uniform quantum algorithm $A$ that solves $\QCktValue(s, \eps, \gamma)$ in time $O(s \log(1/\gamma)/\eps^2))$.
\end{lemma}

\begin{proof}
  The algorithm independently run the circuit $C$ $t$ times and let $\tilde{p}$ be the number of times getting 1, divided by $t$.
  By Chernoff's bound, we have
  \[\pr{|p-\tilde{p}| > \eps} < 2^{-\Omega(t\eps^2)}.\]
  By choosing $t = O(\log(1/\gamma)/\eps^2)$, $2^{-t\eps^2}\leq\gamma$.
  Each trial takes $O(s)$ time.
  Therefore, the total running time is $O(s \log(1/\gamma)/\eps^2)$.
\end{proof}

\begin{remark}\label{remark:qckt-value-speedup}
  It is worth mentioning that by using a quantum speed-up (\eg~\cite{Montanaro20150301}), one can improve the dependence on $1/\eps$ quadratically, although this improvement is not crucial for our purposes.
\end{remark}

\begin{definition}[$\QCktTom$ Problem]\label{def:qckt-tomography-app}
  The $\QCktTom(s, d, \eps, \gamma)$ problem is a computational problem defined as follows:
  \begin{itemize}
    \item Input: a description of a quantum circuit $C:\density{\bbC^{d}}\to\zo$ of size-$s$, and an error parameter $0< \eps <1$.
    \item Task: let $\Pi$ be the corresponding BPOVM of $C$.
      Output an explicit description (as matrices) of BPOVM $\tilde{\Pi}$ such that $\opnorm{\Pi-\tilde{\Pi}}\leq \eps$ with probability $1-\gamma$.
  \end{itemize}
\end{definition}

\begin{lemma}\label{lemma:qkct-tomography-app}
  There exists a (uniform) quantum algorithm that solves the $\QCktTom(s, d, \eps, \gamma)$ Problem in time
  $\poly(s, d, 1/\eps, \log(1/\gamma))$.
\end{lemma}
As discussed in the preliminary, one can mathematically derive $\Pi$ from the description of any quantum circuit $C$ (Equation~(\ref{eqn:circuit-povm})).
However, this calculation is computationally heavy, \eg~involving matrix operations over $d\cdot 2^{m}$ dimensions, where $m$ is the number of ancilla qubits used by $C$, which can be as large as $s$, and thus fails to provide desired efficiency.
\begin{proof}
  The strategy we use is to estimate each entry of $\Pi$, which is a $d$ by $d$ positive semidefinite matrix, by feeding special input states to circuit $C$ and observing the statistics of the output bit (\ie~a tomography process for the POVM $\Pi$ (\eg~\cite{POVM-tomo})).

  Since we exploit a quantum machine to perform the circuit $C$ and the measurement, it only costs $O(s)$ for a quantum machine to execute $C$ once.
  The total running time then depends on the number of executions of $C$ for the desired efficiency.

  To that end, we will leverage the following set of special input states, which form an basis for positive semidefinite operators over the input space.
  Let $\{\ket{1}, \cdots, \ket{d}\}$ be any orthonormal basis in $\bbC^d$. Define the following set of density operators:

  \begin{eqnarray}
    \forall n =1, \cdots, d, & A_{n,n}= \ketbra{n}, \\
    \forall 1\leq n< m \leq d, & A_{n,m}^\re=\ketbra{\psi_{n,m}}, \ket{\psi_{n,m}}=\frac{1}{\sqrt{2}}(\ket{n}+\ket{m}), \\
    \forall 1\leq n< m \leq d, & A_{n,m}^\im=\ketbra{\phi_{n,m}}, \ket{\phi_{n,m}}=\frac{1}{\sqrt{2}}(\ket{n}+i \ket{m}).
  \end{eqnarray}
  Also let
  \begin{align*}
    \alpha_{n,n}(\Pi)&=\tr{A_{n,n} \Pi} \\
    \alpha^{\re}_{n,m}(\Pi)&=\tr{A_{n,m}^{\re} \Pi} \\
    \alpha^{\im}_{n,m}(\Pi)&=\tr{A_{n,m}^{\im} \Pi}
  \end{align*}
  The collection of values $\alpha_{n,n}(\Pi)$ for $n=1, \cdots, d$ and $\alpha^{\re}_{n,m}(\Pi)$ and  $\alpha^{\im}_{n,m}(\Pi)$ for $1\leq n< m \leq d$ uniquely determines any positive semidefinite operator $\Pi$.\footnote{It is not hard to see that $\alpha_{n,n}(\Pi)$ determines all the diagonal entrees. Every off-diagonal entree $(n,m)$ (or its conjugate at $(m,n)$) is then determined by $\alpha^{\re/\im}_{n,m}(\Pi)$ together with the information about the diagonal entree $(n,n)$ and $(m,m)$.}
  It suffices to collect these $\alpha$ values to within small error to approximate $\Pi$.
  We will invoke Lemma~\ref{lemma:qkct-value} for each $\alpha$ value for that purpose.
  Overall, by a union bound, with probability $1-\gamma$, we can collect a set of $\tilde{\alpha}$ values that approximate the original $\alpha$ values with an additive error $\eta$ in time $d^2\cdot O\left( s\log(d/\gamma)/\eta^2\right)=\poly(s, d, \log(1/\gamma), 1/\eta)$.
  Namely, for all $n,m$, we have
  \[
    |\tilde{\alpha}_{n,n}- \alpha_{n,n}(\Pi) |\leq \eta,
    |\tilde{\alpha}^{\re}_{n,m}- \alpha^{\re}_{n,m}(\Pi) |\leq \eta,
    \text{ and } |\tilde{\alpha}^{\im}_{n,m}- \alpha^{\im}_{n,m}(\Pi) |\leq \eta.
  \]
  We can thus solve the following semidefinite program (SDP) to recover an approximate $\tilde{\Pi}$:
  \begin{eqnarray*}
   \text{Goal:} & \text{find a } \tilde{\Pi} \\
   \text{Subject to:} &  |\tilde{\alpha}_{n,n}- \alpha_{n,n}(\tilde{\Pi}) |\leq \eta, \\
   & |\tilde{\alpha}^{\re/\im}_{n,m}- \alpha^{\re/\im}_{n,m}(\tilde{\Pi}) |\leq \eta,\\
   & 0\leq \tilde{\Pi}\leq \I.
  \end{eqnarray*}
  We claim that any feasible solution $\tilde{\Pi}$ is a good approximate of $\Pi$. This is because by definition and the triangle inequality, all the $\alpha$ values of $\Pi$ and  $\tilde{\Pi}$ are close.
  Namely,
  \[
    |\tilde{\alpha}_{n,n}(\tilde{\Pi})- \alpha_{n,n}(\Pi) |\leq 2\eta,  |\tilde{\alpha}^{\re}_{n,m}(\tilde{\Pi})- \alpha^{\re}_{n,m}(\Pi_1) |\leq 2\eta,
  |\tilde{\alpha}^{\im}_{n,m}(\tilde{\Pi})- \alpha^{\im}_{n,m}(\Pi) |\leq 2\eta.
  \]
  This implies the max norm of $\tilde{\Pi}-\Pi$ is small, \ie~$\norm{\tilde{\Pi}-\Pi}_{\max} \leq O(\eta)$. By Equation (\ref{eqn:max_norm}), we have
  \[
  \opnorm{\tilde{\Pi}-\Pi} \leq d \norm{\tilde{\Pi}-\Pi}_{\max} = O(d \eta).
  \]
  It then suffices to choose $\eta=O(\eps/d)$.
  Overall, the above circuit succeeds with probability at least $1-\gamma$ and runs in
  $\poly(s, d, 1/\eps, \log(1/\gamma))$ time.
\end{proof}

\begin{definition}[$\QCktMaxSat$ Problem]\label{def:qckt-max-sat}
  The $\QCktMaxSat(s, d, \eps, \gamma)$ problem is a computational problem defined as follows:
  \begin{itemize}
    \item Input: a description of a quantum circuit $C:\density{\bbC^{d}}\to\zo$ of size-$s$, and an error parameter $0< \eps <1$.
    \item Task: Output an explicit description (as a density matrix) of a quantum state $\rho\in\density{\bbC^{d}}$ such that $C(\rho) > \max_{\sigma} C(\sigma) - \eps$ with probability $1-\gamma$.
  \end{itemize}
\end{definition}

\begin{theorem}\label{thm:qckt-max-sat}
  There exists a (uniform) quantum algorithm $A$ that solves $\QCktMaxSat(s, d, \eps, \gamma)$ problem in time $\poly(s, d, 1/\eps, \log(1/\gamma))$.
\end{theorem}

\begin{proof}
  This theorem follows from Lemma~\ref{lemma:qkct-tomography-app} and an application of a spectrum decomposition.
  Let $\Pi$ be the corresponding BPOVM of $C$.
  By Lemma~\ref{lemma:qkct-tomography-app}, there exists an circuit $A$ that runs in time $\poly(s, d, 1/\eps, \log(1/\gamma))$ and outputs a description of $\tilde{\Pi}$ such that $\opnorm{\tilde{\Pi}-\Pi}\leq \eps/2$ with probability $1-\gamma$.
  That means for all $\tau \in \density{\bbC^{d}}$,
  \begin{equation} \label{eqn:opclose}
     \left | \ip{\tilde{\Pi}}{\tau} -\ip{\Pi}{\tau} \right |\leq \eps/2.
  \end{equation}
  We then run a spectrum decomposition on $\tilde{\Pi}$ and choose $\rho=\ketbra{\psi}$ to be the density operator corresponding to the eigenvector $\ket{\psi}$ with the largest eigenvalue of $\tilde{\Pi}$.
  This step can be done in $\poly(d)$ given that dimension of $\tilde{\Pi}$ is $d$ (\eg~by SDP).
  Thus, we have
  \begin{equation} \label{eqn:max}
    \ip{\tilde{\Pi}}{\rho} \geq \max_{\sigma} \ip{\tilde{\Pi}}{\sigma}.
  \end{equation}
  By Equation~(\ref{eqn:opclose}), we have
  \begin{eqnarray*}
    \ip{\Pi}{\rho} & \geq &  \ip{\tilde{\Pi}}{\rho} -\eps/2 \\
    & \geq & \max_{\sigma} \ip{\tilde{\Pi}}{\sigma}-\eps/2 \\
    & \geq & \max_{\sigma} \ip{\Pi}{\sigma} -\eps/2-\eps/2\\
    &=&\max_{\sigma} \ip{\Pi}{\sigma} -\eps.
  \end{eqnarray*}
  The overall complexity is $\poly(s, d, 1/\eps, \log(1/\gamma))$, which completes the proof.

\end{proof}

\section{Leakage Simulation Lemma by Quantum Min-Max Theorem}\label{app:min-max-proof}

In this section, we provide another proof of the Leakage Simulation Lemma via the nonuniform quantum min-max theorem.
First, we introduce two techniques required in this proof -- epsilon-net and quantum sampling lemma.
Then we prove the leakage simulation lemma, which can lead to the Leakage Chain Rule as shown in Section~\ref{ssec:leakage-chain-rule}.

\subsection{Epsilon-nets for Quantum States and BPOVM.}
An epsilon-net (or $\eps$-net) of a certain set $X$ is meant to be an approximation of every point in $X$ with ``error'' at most $\eps$ by a collection of discrete points.
Suppose $X$ is a set with a metric $\Delta(\cdot, \cdot)$.
A subset $\cN(X, \eps)$ is an $\eps$-net of $X$ if for each $x \in X$, there exists $y\in \net{X, \eps}$ such that $\Delta(x,y)\leq \eps$.
We are interested in $\eps$-nets of (1) the set of pure $d$-dimensional quantum states with respect to the trace distance and (2) the set of BPOVMs with respect to the operator norm.

\begin{proposition}[\eg~\cite{Ver11}] \label{prop:enet-vec}
  For every $0<\eps<0.01$ and $d\in\bbN$, there is an $\eps$-net $\net{\ball{\bbC^d}, \eps}$ of the complex $d$-dimensional unit sphere (which we view as the set of pure quantum states) with respect to the Euclidean norm such that
  the size of $\net{\ball{\mathbb{C}^d}, \eps}$ is bounded by $(1/\eps)^{O(d)}$.
  Moreover, for any pure state $\ket{\psi}\in \ball{\bbC^d}$, there exists  $\ket{\phi}\in\net{\ball{\bbC^d, \eps}}$ with $\eps<0.01$ such that for every  quantum circuit $C$ with binary output, we have
  \[\left|\ex{C(\ketbra{\psi})} - \ex{C(\ketbra{\phi})}\right| \leq 2 \eps.\]
\end{proposition}

The set of all BPOVMs on a $d$-dimension system is denoted by $\bpovm{\mathbb{C}^d}$.
Namely, $\bpovm{\mathbb{C}^d}$ consists of all Hermitian operator $\Pi$ such that $0\leq \Pi \leq \I_d$ and $0\leq \I_d-\Pi \leq \I_d$.

\begin{proposition}[\eg~\cite{SW12}]\label{prop:enet-bpovm}
  For every $0<\eps<\frac{1}{2}$ and $d\in\mathbb{N}$, there is an $\eps$-net $\net{\bpovm{\mathbb{C}^d}, \eps}$ of $d$-dimensional BPOVMs with respect to the operator norm such that the size of $\net{\bpovm{\mathbb{C}^d}, \eps}$ is bounded by $(1/\eps)^{O(d^2)}$.
  Moreover, for every quantum circuit $C$ with $d$-dimension input and binary output, there exists a quantum circuit $C'$ with BPOVM in $\net{\bpovm{\bbC^d}, \eps}$ such that for all quantum states $\rho\in\density{\bbC^d}$, we have
  \[\left|\ex{C(\rho)} - \ex{C'(\rho)}\right| < \eps.\]
\end{proposition}

\subsection{Derandomization via Sampling}
The derandomization lemma says that for every distribution of circuits, there is a ``small'' circuit to approximate the distribution such a given set of functions cannot tell the difference.
In the classical case, it has been used to prove many results, such as Impagliazzo's Hardcore Lemma~\cite{Impagliazzo95}, the nonuniform min-max theorem, Regularity Lemma~\cite{TrevisanTV09} and Dense Model Theorem~\cite{ReingoldTTV08}.
This technique was formally defined as Lemma~3 of~\cite{ChungLP15}.
Herein we generalize it to allow circuits with quantum inputs and outputs, and use the epsilon-net method to handle the fact that there are infinitely many quantum states in a state space.

\begin{lemma}\label{lemma:sample-approx}
  Let $X$ be a finite space with $\size{X} = d_X$ and $\cY$ be a Hilbert space with dimension $d_Y$.
  Let $\overline{\cC}$ be a distribution over $\cC$, a class of quantum circuits with input space $X$ and output space $\density{\cY}$.
  Then for every $\eps \in (0, 1/2)$, there exists a quantum distinguisher $\widehat{C}$ with complexity\footnote{A circuit $\widehat{C}$ has complexity $O(t)$ with respect to $\cC$ if $\widehat{C}$ is composed of $O(t)$ circuits in $\cC$ and $O(t)$ universal gates.} $O(t)$ with respect to $\cC$ such that

  1) For all $x\in X$ and a distinguisher $D: X \times \density{\cY} \to\zo$,
    \[\left| \ex[C\from\overline{\cC}]{D(x, C(x))} - \ex{D(x, \widehat{C}(x))} \right| \leq \eps,\]
    and
  2)
  $\displaystyle t = O\left(\frac{1}{\eps^2}(\log d_X + d_Y^2 \log(1/\eps))\right).$
\end{lemma}

\begin{proof}
  We sample a set of $t$ circuits $C_1, \dots, C_t$ from the distribution $\overline{\cC}$ and let $\widehat{C}$ be a circuit that randomly chooses one of $\{C_1, \dots, C_t\}$ to run.
  Then for every $x\in X$,
  \[\ex[C\from\overline{\cC}]{C(x, \rho)} = \ex[\widehat{C}\from\overline{\cC}]{\widehat{C}(x, \rho)} = \frac{1}{t}\sum_{i = 1}^t \ex{C_i(x, \rho)}.\]
  For every $x\in X$ and a distinguisher $D$, by Chernoff bound, we have
  \begin{align}\label{eq:concentrate}
    \pr[\widehat{C}\from\overline{\cC}]{\left|\ex[C\from\overline{\cC}]{D(x, C(x))} - \ex[\widehat{C}]{D(x, \widehat{C}(x))} \right|>\frac{\eps}{2}} \leq 2^{-t\eps^2/16}.
  \end{align}

  Recall that a quantum distinguisher with input in $\density{\bC^{N}}$ can be represented by a $N$ by $N$ BPOVM matrix.
  If we denote the BPOVM of the distinguisher $D(x, \cdot)$ by $\Pi_x$, then the distinguisher $D$ by a set of BPOVM $\{\Pi_x\}_{x\in X}$.
  Particularly, $D(x, \rho) = \ip{\Pi_x}{\rho}$.
  Thus, Equation~(\ref{eq:concentrate}) can be written as
  \begin{align}\label{eq:concentrate-2}
    \pr[\widehat{C}\from\overline{\cC}]{\left|\ex[C\from\overline{\cC}]{\ip{\Pi_x}{C(x)}} - \ex[\widehat{C}]{\ip{\Pi_x}{\widehat{C}(x)}} \right|>\frac{\eps}{2}} \leq 2^{-t\eps^2/16}.
  \end{align}

  Apply the union bound to Equation~(\ref{eq:concentrate}) over every $x\in X$ and $\Pi_x \in\net{\bpovm{\cY}, \eps/2}$, then
  \begin{align*}
  &\pr[\widehat{C}\from\overline{\cC}]{\forall x\in X, \Pi_x\in\net{\bpovm{\cY}, \eps/2}\;,\;\;\left|\ex[C\from\overline{\cC}]{\ip{x}{C(x)}} - \ex[\widehat{C}]{\ip{x}{\widehat{C}(x)}} \right|>\frac{\eps}{2}} \\
  \leq &|X|\cdot|\net{\bpovm{\cY}, \eps/2}|\cdot 2^{-t\eps^2/16}.
  \end{align*}
  Since $\size{\net{\bpovm{\cY}, \eps/2}} = (1/\eps)^{O(d_Y^2)}$, we can choose $t = O\left(\frac{1}{\eps^2}(\log d_X + d_Y^2 \log(1/\eps))\right)$ such that the quantity on the right-hand side of the inequality is less than one.
  That implies there exists a choice of $t$ circuits $C_1, \dots, C_t$ to form the circuit $\widehat{C}$ such that for all $x\in X$ and $\Pi_x\in\net{\bpovm{\cY}, \eps/2}$, $\widehat{C}$ and the expectation of $\cC$ are $\eps/2$-close measured by $\Pi_x$.
  Apply Proposition~\ref{prop:enet-bpovm}, we can extend the above statement from a BPOVM in the net to any BPOVM (or equivalently, any quantum distinguisher) by losing another $\eps/2$.
  Namely, we have that there exists a choice of $t$ circuits $\{C_1, \dots, C_t\}$ to form the circuit $\widehat{C}$ such that
  \[\forall x\in X, D(x, \cdot)\;,\;\;\left|\ex[C\from\overline{\cC}]{D(x, C(x))} - \ex[\widehat{C}]{D(x, \widehat{C}(x))} \right|>\frac{\eps}{2} + \frac{\eps}{2} = \eps.\]
\end{proof}

\subsection{Leakage Simulation Lemma}

\repeattheorem{thm:leak-sim}

\begin{proof}

  Suppose for contradiction that for all size-$s'$ quantum circuit $C:\bin^n\to\density{\bbC^{2^\ell}}$, there exists a quantum distinguisher $D:\bin^n\times \density{\bbC^{2^\ell}}\to\bin$ of size $s$ such that
  \[\ex{D(\rho_{XB})} - \ex{D(X, C(X))} \geq \eps.\]
  First we transform circuits of bounded size to distributions of circuits of bounded size via the following claim.
  \begin{claim}\label{claim:circuit-to-dist}
    For every distribution $\overline{C}$ over size-$s''$ quantum circuit with $s' = s''\cdot O((n+2^{2\ell})/\eps^2)$, there exists a distinguisher $D$ of size $s$ such that $\ex{D(\rho_{XB})} - \ex[C\sim\overline{C}]{D(X, C(X))} < \eps/2$.
  \end{claim}
  \begin{proof}[proof of Claim~\ref{claim:circuit-to-dist}]
    Suppose that there exists a distribution $\overline{C}$ over size-$s''$ circuit such that for all size-$s$ distinguisher $D$,
    \[\ex{D(\rho_{XB})} - \ex[C\sim\overline{C}]{D(X, C(X))} < \eps/2.\]
    Apply Lemma~\ref{lemma:sample-approx} with $d_X = 2^n$ and $d_Y = 2^\ell$, then there exists a circuit $\widehat{C}$ of size $s' = s''\cdot O((n+2^{2\ell})/\eps^2)$ such that $\ex{D(\rho_{XB})} - \ex{D(X, \widehat{C}(X))} < \eps/2+\eps/2 = \eps$ which contradict the assumption.
  \end{proof}
  Once we have Claim~\ref{claim:circuit-to-dist}, we apply the nonuniform Quantum Min-Max Theorem (Theorem~\ref{thm:min-max}) using the following parameters.
  The strategy space of Player~1 is
  \[\cA = \left\{\mbox{cq-state }\sum_x p_x\ketbra{x}\ot C(x) \middle| C:\zo^n\to\density{\bbC^{2^\ell}} \mbox{ is a quantum circuit of size } s''\right\},\] and the strategy space of Player~2 $\cB$ is the set of all distinguishers with size at most $s$.
  The mapping $f$ is defined to be
  \[f(D) = E[D(\rho_{XB})]\I_{n+\ell} + \Pi_D.\]
  where $\Pi_D$ is the BPOVM of $D$.
  By the Quantum Min-Max Theorem, we know there exists a quantum distinguisher $\tilde{D}$ of size $s\cdot O\left(\frac{1}{\eps^2}(n+\ell)\right)$ such that for all for all $s''$-size circuit $\tilde{C}$,
  \begin{align}\label{eq:quantum-gw-mm}
    \ex{\tilde{D}(\rho_{XB})} - \ex{\tilde{D}(X, \tilde{C}(X))} > \eps/4.
  \end{align}
  Express the BPOVM of $\tilde{D}(x, \cdot)$ to be $\Pi_{x}$ for $x\in\zo^n$, then we have
  \[\ex{\tilde{D}(\rho_{XB})} = \sum_{x\in\bin^n}p_x\ip{\Pi_x}{\rho^x_B}.\]
  Now, define the quantum circuit $C$ is as follows:
  \begin{enumerate}
    \item For input $x\in\bin^n$, apply Lemma~\ref{thm:qckt-max-sat} to solve the $(s, \ell, \eps/8, \eps/8)$-$\QCktMaxSat$ Problem with the quantum circuit $\tilde{D}(x, \cdot)$ to get a description of the quantum state $\sigma_x$ such that with probability $1-\eps/8$,
    \[\inprod{\Pi_x, \sigma_x} \geq \max_{\rho}\inprod{\Pi_x, \rho} - \eps/8.\]
    \item Construct the quantum state $\sigma_x$ and output it.
  \end{enumerate}
  The state can be constructed from its description by a circuit of size polynomial in its description length and dimension~\cite{SBM05}.
  So the running time of $C$ is $\poly(s, n, 2^{\ell}, 1/\eps)$.
  Eventually, we have
  \begin{align*}
    \ex{\tilde{D}(X, C(X))}
    = \sum_{x\in\bin^n}p_x\inprod{\Pi_x, \sigma_x}
    \geq& \left(1-\frac{\eps}{8}\right) \left(\sum_{x\in\bin^n}p_x\max_{\rho}\inprod{\Pi_x, \rho} - \frac{\eps}{8}\right)\\
    \geq& \left(1-\frac{\eps}{8}\right)\left(\max_{C'}\ex{\tilde{D}(X, C'(X))} - \frac{\eps}{8}\right)\\
    \geq& \max_{C'}\ex{\tilde{D}(X, C'(X))} - \frac{\eps}{4}.
  \end{align*}
  which contradicts Equation~(\ref{eq:quantum-gw-mm}).

\end{proof}

\section{Proof of MMWU Theorem via KL-projection}

In this section, we proof the MMWU Theorem considering the setting in Section~\ref{subsec:min-max} and Procedure~\ref{pro:mmwu-minmax}.

First, we have the following facts and lemmas.

\begin{fact}\label{fact:log-approx}
  For $x < 1$, $-\ln (1-x) \geq x + x^2/2 + x^3/3 $; for $X$
\end{fact}

\begin{fact}\label{fact:tr-ineq}
  If $B > C$, then $\tr{AB} > \tr{AC}$.
\end{fact}

\begin{lemma}\label{lemma:kl-proj}
  Let $Y\in\density{\cH}$ and $\cA$ be a convex set in $Y\in\density{\cH}$.
  Let $Y^*$ be a KL-projection of $Y$ on $\cA$, then
  \[\KL{X}{Y^*} + \KL{Y^*}{Y} \leq \KL{X}{Y}\]
  Particularly, due to the non-negativity of the KL-divergence,
  \[\KL{X}{Y^*} \leq \KL{X}{Y}\]
\end{lemma}
\begin{proof}
  Let
  \[Z_\lda = \lda X + (1-\lda) Y^*.\]
  Since $\cA$ is convex, $Z_\lda\in\cA$ for all $0 \leq \lda \leq 1$.
  Because the minimum of $\KL{Z_\lda}{Y}$ happens at $\lda = 0$, the derivative of $\KL{Z_\lda}{Y}$ is nonnegative at $\lda = 0$.
  \begin{align*}
    \frac{d}{d\lda}\KL{Z_\lda}{Y} = & \frac{d}{d\lda} \tr{Z_\lda \log Z_\lda} - \tr{Z_\lda \log Y}\\
    = & \tr{\frac{d Z_\lda}{d\lda}\log Z_\lda} + \tr{Z_\lda \frac{d\log Z_\lda}{d\lda}} -  \tr{\frac{d Z_\lda}{d\lda}\cdot \log Y}
  \end{align*}
  It is straightforward to have $\frac{dZ_\lda}{d\lda} = X - Y^*$.
  For $\tr{Z_\lda \frac{d\log Z_\lda}{d\lda}}$,
  \begin{align*}
    \tr{Z_\lda \frac{d\log Z_\lda}{d\lda}} = & \tr{Z_\lda \frac{d}{d\lda}\left(\sum_{n = 1}^{\infty}-\frac{(-1)^n}{n}(Z_\lda-\I)^n\right)}\\
    = & \tr{
      Z_\lda \sum_{n = 1}^{\infty}-\frac{(-1)^n}{n} \sum_{i = 0}^{n-1} (Z_\lda-\I)^i\cdot \frac{dZ_\lda}{d\lda} \cdot (Z_\lda-\I)^{n-1-i} }\\
    = & \sum_{n = 1}^{\infty}-\frac{(-1)^n}{n} \sum_{i = 0}^{n-1} \tr{
      Z_\lda (Z_\lda - \I)^i \cdot (X-Y^*) \cdot (Z_\lda-\I)^{n-1-i}
    }\\
    \overset{(a)}{=} & \sum_{n = 1}^{\infty}-\frac{(-1)^n}{n} \sum_{i = 0}^{n-1} \tr{
      Z_\lda (Z_\lda - \I)^{n-1} \cdot (X-Y^*)
    }\\
    = & \tr{
      Z_\lda \sum_{n = 1}^{\infty}-\frac{(-1)^n}{n} \sum_{i = 0}^{n-1} (X-Y^*)
    }\\
    = & \tr{
      Z_\lda Z_\lda^{-1} (X-Y^*)
    } = \tr{X-Y^*} = 0,
  \end{align*}
  where the equality $(a)$ is due to the commutativity between $Z_\lda$ and $(Z_\lda-\I)$, and the invariance of trace under cyclic permutations.
  Therefore,
  \begin{align*}
    \left.\frac{d}{d\lda}\KL{Z_\lda}{Y}\right|_{\lda = 0} = & \left.\tr{(X-Y^*)\log Z_\lda}  - \tr{(X-Y^*)\cdot \log Y}\right|_{\lda = 0}\\
    = & \tr{(X-Y^*)\log Y^*}  - \tr{(X-Y^*)\cdot \log Y}\\
    = & \tr{X(\log X - \log Y)} - \tr{X(\log X - \log Y^*)} - \tr{Y^* (\log Y^* - \log Y)}\\
    = & \KL{X}{Y} - \KL{X}{Y^*} - \KL{Y^*}{Y} \geq 0,
  \end{align*}
  which yields the conclusion by an rearrangement.
\end{proof}

\begin{lemma}\label{lemma:mmwu-kl-proof}
  Consider the setting in Section~\ref{subsec:min-max} and Procedure~\ref{pro:mmwu-minmax}.
  We have for all $a\in\cA$,
  \[
    \frac{1}{T}\sum_{t=1}^T \ip{a^{(t)}}{f(b^{(t)})}
    \leq
    \frac{1}{T}\sum_{t=1}^T \ip{a}{f(b^{(t)})} + \left(\eta + \frac{\log d}{\eta T}\right).
  \]
\end{lemma}

\begin{proof}
  Let $L^{(t)} = f(b^{t})$.
  First, recall the definition of relative entropy (KL-divergence) of two quantum states $\rho, \sigma$:
  \[\KL{\rho}{\sigma} = \tr{\rho(\log \rho - \log \sigma)}.\]
  For any density matrix $a$,
  \begin{align*}
    & \KL{a}{a^{(t)}} - \KL{a}{a^{(t+1)''}}\\
    = & -\tr{a \log a^{(t)}} + \tr{a\log a^{(t+1)''})}\\
    = & -\tr{a\left(\log a^{(t)}} + \tr{a\log\left(\frac{\exp(\log a^{(t)} -\eta L^{(t)})}{\tr{\exp(\log a^{(t)} -\eta L^{(t)})}}\right)\right)}\\
    = & -\tr{a\log a^{(t)}} + \tr{a\log a^{(t)}} + \tr{-\eta a L^{(t)})} - \log\left(\tr{\exp(\log a^{(t)} -\eta L^{(t)})}\right)\\
    \overset{(a)}{\geq} & \tr{-\eta a L^{(t)})} - \log\left(\tr{a^{(t)}\exp(-\eta L^{(t)})}\right)\\
    = & -\eta \ip{a}{L^{t}} - \log\left(\tr{a^{(t)}(\I_d - (\I_d - \exp(-\eta L^{(t)})))}\right)\\
    = & -\eta \ip{a}{L^{t}} - \log\left(\tr{a^{(t)}} - \tr{a^{(t)}(\I_d - \exp(-\eta L^{(t)}) )} \right)\\
    = & -\eta \ip{a}{L^{t}} - \log\left(1 - \tr{a^{(t)}(\I_d - \exp(-\eta L^{(t)}) )} \right)\\
    \overset{(b)}{\geq} & -\eta \ip{a}{L^{t}} + \tr{a^{(t)}(\I_d - \exp(-\eta L^{(t)}) )}\\
    \overset{(c)}{\geq} & -\eta \ip{a}{L^{t}} + \tr{a^{(t)}(\I_d - \I_d + \eta L^{(t)} - (\eta L^{(t)})^2)}\\
    = & \eta \ip{a^{(t)}}{L^t} - \eta \ip{a}{L^{t}} -\eta^2.
  \end{align*}
  Inequality (a) is due to the Golden-Thompson inequality.
  For the inequality (b), notice that $0 < \eta L^{(t)} < \I_d$, then $0 < \I_d - \exp(-\eta L^{(t)}) < I_d$, so by Fact~\ref{fact:tr-ineq}, $\tr{a^{(t)}(\I_d - \exp(-\eta L^{(t)}) )} < \tr{a^{(t)}I_d} = 1$.
  Finally, applying Fact~\ref{fact:log-approx}, we get the inequality.
  For the inequality (c), it is due to the definition of an exponential of a matrix and Fact~\ref{fact:log-approx}.

  Since $a^{(t+1)}$ is the KL-projection of $a^{(t+1)''}$, by Lemma~\ref{lemma:kl-proj}, we have
  \[\KL{a}{a^{(t+1)}} \leq \KL{a}{a^{(t+1)''}},\]
  and so
  \begin{align*}
    \KL{a}{a^{(t)}} - \KL{a}{a^{(t+1)}} \geq & \KL{a}{a^{(t)}} - \KL{a}{a^{(t+1)''}}\\
    \geq & \eta \ip{a^{(t)}}{L^t} - \eta \ip{a}{L^{t}} -\eta^2
  \end{align*}
  Now we do the telescoping from $t = 1$ through $t = T$, we get
  \begin{align*}
    \KL{a}{a^{(1)}} - \KL{a}{a^{(T+1)}} \geq \eta\sum_{t = 1}^{T} \ip{a^{(t)}}{L^t} - \sum_{t = 1}^{T} \eta \ip{a}{L^{t}} - T\eta^2
  \end{align*}
  Since KL-divergence is always non-negative and when $a^{(1)} = \frac{1}{d}\I_d$, $\KL{a}{a^{(1)}} = \log d - H(a)$ where $H$ is the von Neumann entropy, we have
  \begin{align*}
    \log d  \geq \eta\sum_{t = 1}^{T} \ip{a^{(t)}}{L^t} - \sum_{t = 1}^{T} \eta \ip{a}{L^{t}} - T\eta^2\\
    \frac{1}{T}\sum_{t = 1}^{T} \ip{a^{(t)}}{L^t} \leq \frac{1}{T}\sum_{t = 1}^{T} \ip{a}{L^{t}} + \eta + \frac{\log d}{\eta T}.
  \end{align*}
\end{proof}

\end{document}